\documentclass[11pt,reqno,draft]{amsart}

\def\a{\alpha}
\def\b{\beta}
\def\c{\gamma}
\def\e{\varepsilon}
\def\ic{{\rm i}}
\def\s{\sigma}
\def\w{\omega}
\def\D{\Delta}
\def\p{\partial}
\def\x{{\hat{x}}}
\def\r{\rho}
\def\cG{{\mathcal G}}
\def\R{\mathbb R}
\def\N{\mathbb N}

\newtheorem{thm}{Theorem}
\newtheorem{lemma}[thm]{Lemma}
\newtheorem{cor}[thm]{Corollary}
\newtheorem{defn}[thm]{Definition}
\theoremstyle{definition}
\newtheorem*{acknowledgement}{Acknowledgement} 
\newtheorem{remark}[thm]{Remark}{\rm}{\rm}
\newcommand{\be}{\begin{equation}}
\newcommand{\ee}{\end{equation}}
\newcommand{\bea}{\begin{eqnarray}}
\newcommand{\eea}{\end{eqnarray}}
\newcommand{\beax}{\begin{eqnarray*}}
\newcommand{\eeax}{\end{eqnarray*}}
\newcommand{\Tr}{\mbox{\rm Tr}}
\newcommand{\supp}{\mbox{\rm supp}}
\newcommand{\dist}{\mbox{\rm dist}}
\newcommand{\mfr}[2]{{\textstyle\frac{#1}{#2}}}
\newcommand{\V}{V^{\rm TF}}
\newcommand{\rmc}{\rm c}
\newcommand{\rmm}{\rm m}
\newcommand{\rme}{\rm e}
\newcommand{\cS}{\mathcal S}

\begin{document}

\title{THE RELATIVISTIC SCOTT CORRECTION FOR 
  ATOMS AND MOLECULES}  

\author[J.~P.~Solovej]{Jan Philip Solovej}
    \address{Department of Mathematics\\ 
    University of Copenhagen\\
    Universitetsparken 5\\
    DK-2100 Copenhagen, Denmark}
    \email{solovej@math.ku.dk}
  
\author[T.~\O stergaard~S\o rensen]{Thomas \O stergaard S\o rensen}  
    \address{Department of Mathematical Sciences\\ 
    Aalborg University\\
    Fredrik Bajers Vej 7G\\
    DK-9220 Aalborg East, Denmark}
    \email{sorensen@math.aau.dk}
     
\author[W.~L.~Spitzer]{Wolfgang L. Spitzer}
    \address{Institut f\"ur Theoretische Physik\\
    Universit\"at Erlangen-N\"urnberg\\
    Staudtstrasse 7\\
    D-91058 Erlangen, Germany}
    \email{Wolfgang.Spitzer@physik.uni-erlangen.de}

\thanks{\copyright\ 2008 by the authors. This article may be
      reproduced, in its entirety, for non-commercial purposes. }

\date{\today}

\begin{abstract} 
We prove the first correction to the leading Thomas-Fermi energy for
the ground state energy of atoms and molecules in a model where the
kinetic energy of the electrons is treated relativistically. The
leading Thomas-Fermi energy, established in \cite{sor}, as well as the
correction given here are of semi-classical nature. Our result on
atoms and molecules is proved from a general semi-classical estimate
for relativistic operators with potentials with Coulomb-like
singularities. This semi-classical estimate is obtained using the
coherent state calculus introduced in \cite{SS}. The paper contains a
unified treatment of the relativistic as well as the non-relativistic
case.
\end{abstract}

\maketitle

\tableofcontents

\section{Introduction and main results}

Our goal in this paper is to study how relativistic effects influence
the energies of atoms and molecules. More specifically, we are aiming
at proving a relativistic analog of the celebrated Scott correction
\cite{Scott,Lieb1,Hughes,Ivrii-Sigal,Siedentop-Weikard1,
  Siedentop-Weikard2, Siedentop-Weikard3,SS}. At present there is no
mathematically well-defined fully relativistically invariant theory of
atoms and molecules. We will here consider a simplified model, which
shows the relevant qualitative features of relativistic effects. In
this model, these effects are introduced by treating the kinetic
energy of electrons of mass $\rmm$ by the operator $\sqrt{-\hbar^2
  \rmc^2\Delta+\rmm^2\rmc^4}-\rmm\rmc^2$ instead of the standard
non-relativistic Laplace operator $-\hbar^2\Delta/2\rmm$. Here $\rmc$
refers to the speed of light and $\hbar$ is Planck's constant. It is
the simplest of a class of models that attempts to include relativistic
effects; see \cite{Herbst,Lieb-Yau}. Although this model does not give
accurate numerical agreement with observations it is from a
qualitative point of view quite realistic.

One of the qualitative features that our model shares with all other
relativistic models is the instability of large atoms or
molecules. The natural parameter to measure relativistic corrections
in atoms and molecules is the dimensionless fine-structure constant
$\a={\rme}^2/\hbar {\rmc}$, where $\rme$ is the electron charge. As we
will explain below, if  $Z\alpha$ is too large ($Z$ is the atomic
number) then atoms are unstable. In our model the critical value of
$Z\alpha$ is too small compared with experimental results and with
what is assumed to be the correct critical value, namely $Z\alpha=1$.

Our main interest here is the behavior of the total ground state
energy of large atoms and molecules. Because of the relativistic
instability problem mentioned above one cannot simply consider the
limit of large atomic number $Z$.  One is forced to look at the
simultaneous limit of small fine-structure constant $\alpha$ in such a
way that the product $Z\alpha$ remains bounded. Of course, $\alpha$
has a fixed value which experimentally is approximately $1/137$. Thus
considering the limit $\alpha$ tending to zero is strictly speaking
not physically correct.  Likewise, considering the limit of $Z$
tending to infinity is in contradiction with the fact that the
experimentally observed values of $Z$ are bounded (by 92 for the
stable atoms). Studying the limit $Z\to\infty$ and $\alpha\to0$ with
$Z\alpha$ bounded allows us to make a precise mathematical statement
about the asymptotics. There is numerical evidence that the
asymptotics is indeed a good approximation to the total ground state
energy for the physical values of $Z$ and $\alpha$.

The first to, at least heuristically, suggest to consider $Z\alpha$ as
a separate parameter in the limit $Z\to\infty$ was Schwinger
\cite{Schwinger0}. In this original paper, Schwinger finds
discrepancies of his estimates of relativistic corrections with
numerical evidence. Later \cite{Dmitrieva-Plindov}, more corrections
are taken into account and excellent agreement is found. This accuracy
however goes beyond a rigorous mathematical treatment. We content
ourselves with giving a rigorous treatment of the simplified model
with the correct qualitative behavior.

The first rigorous treatment of the limit $Z\to\infty$ with $Z\alpha$
bounded was given by one of us is the paper \cite{sor}, where the
leading asymptotics of the ground state energy was found. It turns out
it does not depend on $Z\alpha$. The goal of the present paper is the
first correction to the leading asymptotics, i.e., the Scott
correction and, in particular, to show that it depends on
$Z\alpha$. The work in \cite{sor} was generalized to another
relativistic model in \cite{Cassanas-Siedentop}. 

We now introduce the molecular many-body Hamiltonian we consider in
this paper. Let ${\rme}$ and ${\rmm}$ denote the electric charge and
mass of an electron.  Let ${\bf Z}{\rme}=
(Z_1{\rme},\ldots,Z_M{\rme})$, where $Z_1,\ldots,Z_M>0$, be the
charges of the $M$ nuclei. We consider the Born-Oppenheimer
formulation where these nuclei are at fixed positions ${\bf
  R}=(R_1,\ldots,R_M)\in\R^{3M}$. We have $N$ electrons. As explained
above the relativistic kinetic energy of the $j$-th electron is equal
to $\sqrt{-\hbar^2{\rmc}^2\D_j + {\rm m^2c^4}}-{\rm mc^2}$, where
$\Delta_j$ is the Laplacian with respect to the \(j\)-th electron
coordinate \(y_j\in\R^3\), \(j=1,\ldots,N\). The potential energy of
the electrons is composed of the attraction to the nuclei,
\be\label{V} 
  {\rme}V({\bf Z}{\rme},{\bf R},y) = \sum_{k=1}^M
  \frac{Z_k{\rm e}^2}{|y-R_k|}\,, 
\ee 
and the electron-electron repulsion,
$$ \sum_{1\le i<j\le N} \frac{{\rme}^2}{|y_i-y_j|}\,.
$$ 
The total energy of the electrons is described by the Hamiltonian,
\beax 
  H_{\rm rel} =\sum_{j=1}^N \Big[\sqrt{-\hbar^2{\rmc}^2\D_j +
  {\rmm}^2{\rmc}^4}-{\rmm}{\rmc}^2 - {\rme}V({\bf Z}{\rme},{\bf R},y_j)
  \Big]+\sum_{1\le i<j\le N}\frac{{\rme}^2}{|y_i-y_j|} \,.
\eeax
Let us now introduce the fundamental constants. Namely, let
$a=\hbar^2/{\rm m e^2}$ be the Bohr radius, and $R_\infty =
\frac{1}{2}{\rm m e^4}/\hbar^2$ Rydberg's constant. Then by a change
of coordinates $y_j\to x_j = y_j/a$, we see that
\beax 
  \lefteqn{(2R_\infty)^{-1} H_{\rm rel} =:  H({\bf Z},{\bf R};\alpha) 
  = H(Z_1,\ldots,Z_M,R_1,\ldots,R_M;\a)}\\
  &=& \sum_{j=1}^{Z}\Big[\sqrt{-\a^{-2}\D_j + \a^{-4}} - \a^{-2} - 
  V({\bf Z},{\bf R},x_j) \Big] + \sum_{1\le i<j\le N}\frac{1}{|x_i-x_j|}\,, 
\eeax
where again $\a$ is the fine-structure constant. For \(\alpha=0\) the
kinetic energy of the $j$-th electron is \(-\frac12\D\).

Here we have set $N=Z=\sum_{k=1}^MZ_k$ so that the molecule is
neutral. In particular, this means that $Z$ must be an integer. From
now on we study the operator $H({\bf Z},{\bf R};\a)$. This operator
acts as an unbounded operator on the anti-symmetric tensor product,
$\bigwedge^{Z} L^2(\R^3\times\{-1,1\})$, where $\pm1$ refers to the
spin variables. We are interested in the ground state energy 
$$ 
  E({\bf Z},{\bf R};\a) = \inf \sigma\big( H({\bf Z},{\bf R};\a)\big)\,,
$$
and, in particular, in an asymptotic expansion of this when
$Z\to\infty$.

The ground state energy $E({\bf Z},{\bf R};\a)$ is finite if
$\max_k\{Z_k\alpha\}\leq 2/\pi$, but $E({\bf Z},{\bf R};\a)=-\infty$
if $\max_k\{Z_k\alpha\}> 2/\pi$ (see
\cite{Herbst,Lieb-Yau})\footnote{Here, and in the sequel, operators 
  are defined as the Friedrichs extension for the corresponding form
  sum, originally defined on \(C_0^\infty\)-functions (here, for
  instance, $\bigwedge^{Z} C_0^\infty(\R^3\times\{-1,1\})$).}. This is
the relativistic instability discussed above. Therefore we must
require the atomic numbers to be smaller than or equal to
$2/(\pi\alpha)$ which is approximately $87$. This is of course in
contradiction with the experimental fact that larger stable atoms
exist and is one reason why our model can only be qualitatively
correct. (We want to emphasise that the instability we are discussing
here is not the nuclear instability causing atoms larger than atomic
number 92 to be unstable. The relativistic instability we discuss here
is only believed to manifest itself for atomic numbers greater than
137.) 

The true energy of the molecule should include the nuclear-nuclear
repulsion. Since the nuclei are considered fixed here the
nuclear-nuclear repulsion is simply a constant which we have omitted.

As discussed above the leading asymptotics of $ E({\bf Z},{\bf R};\a)$
will be independent of the relativistic parameter $\alpha$. It will be
given by what is called Thomas-Fermi theory. The seminal contribution
by Lieb and Simon~\cite{Lieb-Simon} was to put Thomas-Fermi theory on
a solid mathematical foundation and to prove that in the
non-relativistic case the Thomas-Fermi energy of a molecule is indeed
the correct leading asymptotic energy for the true ground state energy
as $Z\to\infty$. This is the result that was generalized to our
relativistic model in \cite{sor}.

The main result of this paper is the following asymptotic result on
the ground state energy.  
\begin{thm}[Relativistic Scott correction] \label{main theorem} 
  Let ${\bf z}=(z_1,\ldots,z_M)$ with $z_1,\ldots,z_M>0$,
  $\sum_{k=1}^M z_k=1$, and ${\bf r}=(r_1,\ldots,r_M)\in\R^{3M}$ with
  $\min_{k\ne\ell}|r_k-r_\ell|>r_0$ for some $r_0>0$ be given. Define
  ${\bf Z}=(Z_1,\ldots,Z_M)=Z{\bf z}$ and ${\bf R}=Z^{-1/3}{\bf
    r}$. Then there exist a constant $E^{\rm TF}({\bf z},{\bf r})$ and
  a universal (independent of \({\bf z}\), \({\bf r}\) and \(M\))
  continuous, non-increasing function $\cS:[0,2/\pi]\to\R$ with
  $\cS(0)=1/4$ such that as $Z=\sum_{k=1}^M Z_k\to\infty$ and
  $\alpha\to0$ with $\max_k\{Z_k\alpha\}\leq 2/\pi$ we have
\begin{equation} \label{main expansion}
   E({\bf Z},{\bf R};\a) = Z^{7/3}E^{\rm TF} ({\bf z},{\bf r}) +
   2\sum_{k=1}^{M}Z_k^2\cS(Z_k\alpha)
   + {\mathcal O}(Z^{2-1/30})\,.
\end{equation}
Here  the error term means that $|{\mathcal O}(Z^{2-1/30})|\leq
CZ^{2-1/30}$, where the constant $C$ only depends on $r_0$ and
\(M\). As before, \(\sqrt{-\a^{-2}\D+\a^{-4}}-\a^{-2}=-\frac12\D\)
when \(\a=0\).
\end{thm}
\begin{remark}
A less detailed version of our result was announced in
\cite{solovej-oberwolfach}.
\end{remark}
\begin{remark}
Several features of our result and its proof should be stressed: 
  \(\, \)
\begin{enumerate}
\item[\rm (i)] The constant $E^{\rm TF} ({\bf z},{\bf r})$ is
  determined in Thomas-Fermi theory. 
\item[\rm (ii)] The fact that ${\bf R}=Z^{-1/3}{\bf r}$ is the
  relevant scaling of the nuclear coordinates may, as we shall see, be
  understood from Thomas-Fermi theory.  
\item[\rm (iii)] A characterization of the function $\cS$ is given
  explicitly in Corollary~\ref{cor:chrarc-S}  below (see also
  Lemma~\ref{lem:Coulomb}). Its continuity is proved in
  Theorem~\ref{thm:TF-semicl}. 
\item[\rm (iv)] The asymptotic result is uniform in the parameters
  $Z_k\alpha\in[0,2/\pi]$, $k=1,\ldots,M$.
\item[\rm (v)] The result contains, as a special case, the
  non-relativistic situation $Z_k\alpha=0$ and, in particular, the
  non-relativistic limit is controlled due to the continuity of the
  function $\cS$ and the uniformity in the parameters $Z_k\alpha$.  In
  order to get the non-relativistic limit it is important that all
  estimates have the correct non-relativistic behavior. This is an
  important issue in this work.  Note that in the non-relativistic
  case the value $\cS(0)=1/4$ is explicitly known, whereas this is not
  the case for any other value. This is because the eigenvalues of
  Hydrogen are explicitly known in the non-relativistic case, but not
  in this relativistic case. The technique to prove a Scott-correction
  without knowing explicitly the eigenvalues for the one-body
  Hydrogen(like) operator was invented by Sobolev~\cite{Sobolev}.
\item[\rm (vi)] The proof of Theorem~\ref{main theorem} does not rely
  on knowing the non-relativistic case, but treats both the
  relativistic and non-relativistic case simultaneously. 
\item[\rm (vii)] The situation near the critical value
  $Z_k\alpha=2/\pi$ is understood since the function $\cS$ is
  continuous up to the critical value $2/\pi$. This is, however, a
  less important point since we do not know whether the model we study
  gives a good description near the critical value.
\end{enumerate}
\end{remark}
The Scott correction was predicted by Scott~\cite{Scott} to be the
first correction to the Thomas-Fermi energy. In the non-relativistic
setting, this was mathematically established by Hughes~\cite{Hughes},
Siedentop and Weikard~\cite{Siedentop-Weikard1, Siedentop-Weikard2,
  Siedentop-Weikard3} for atoms (and by Bach~\cite{Bach} for ions) and
later by Ivri{\u\i} and Sigal~\cite{Ivrii-Sigal} for molecules. Later
a different proof was given by two of us for molecules
\cite{SS}. Based on methods in \cite{Ivrii-Sigal}, Balodis
Matesanz~\cite{Matesanz} gave a proof for the Scott correction of
matter. The Scott correction for operators with magnetic fields was
studied by Sobolev~\cite{Sobolev, Sobolev2} (in the non-interacting
case). 

In \cite{Fefferman-Seco}, Fefferman and Seco derived rigorously the
second correction to Thomas-Fermi theory for atoms, which is of the
order $Z^{5/3}$. This was predicted by Dirac~\cite{Dirac} and
Schwinger~\cite{Schwinger}. It is apparently still an open problem to
prove this for molecules and to find the relativistic correction to
this order.

The main approach to proving the energy asymptotics for large atoms
and molecules goes back to Lieb and Simon~\cite{Lieb-Simon} and is to
use semi-classical estimates.  The $Z$-scaling makes it possible to
relate the many-body problem to a one-body spectral problem, which may
be treated semi-classically, where the semi-classical parameter is
$h=Z^{-1/3}$. Here, several techniques have been developed.  Lieb and
Simon used Dirichlet--Neumann bracketing. This is however not refined
enough to get beyond the leading term.  The Weyl calculus
\cite{Robert-book} is the most advanced and precise method as far as
optimal semi-classical error estimates are concerned, but it also will
not directly give the Scott correction. Ivri{\u\i} and
Sigal~\cite{Ivrii-Sigal} used Fourier intergral operator techniques to
establish the non-relativistic Scott correction for molecules. Hughes
and Siedentop--Weikard used methods that were designed particularly
for spherically symmetric models, i.e., the atomic case.

A simple method, which is particularly well adapted to many-body
problems is that of coherent states. It was pioneered by
Lieb~\cite{Lieb1} and Thirring~\cite{Thirring} to give very short
proofs that Thomas-Fermi theory is correct to leading order. It is one
of our major contributions here to use an improved calculus of
coherent states as developed by two of us in~\cite{SS} to the
relativistic setting.

One feature of our work is that we give a general semi-classical
estimate for relativistic one-body operators for potentials with
singularities such as the Thomas-Fermi potential (see
Theorem~\ref{thm:TF-semicl} below). This is derived by first proving a
localised semi-classical estimate for potentials with some smoothness
(see Theorem~\ref{local semi-classics}). The proof here is not much
different from the one presented in \cite{SS} for non-relativistic
Schr\"odinger operators. We do not claim that our error estimates are
sharp given the regularity we assume on the potential, but only that
they are sufficient to prove the Scott correction. In this connection
we point out that in order to prove the Scott correction it is enough
that the error relative to the leading term is smaller by more than
one power of the semi-classical paramter $h$. In our case the relative
error in Theorem~\ref{local semi-classics} is $h^{6/5}$.

The relativistic kinetic energy is more cumbersome to work with than
the Laplace operator and large parts of the rest of our proof from
\cite{SS} have to be done differently. A main issue is to be able to
localise into separate regions.  Since the relativistic kinetic energy
is a non-local operator, localisation estimates are more involved than
in the non-relativistic setting. The philosophy is that localisation
errors should behave as if we were working with non-relativistic local
error terms up to some exponentially small tails (see
Theorem~\ref{IMS-error-est}). 

The proof of the main theorem presented in Section \ref{main proof} is
based on the general semi-classical estimate
Theorem~\ref{thm:TF-semicl} and the use of a correlation estimate (see
Theorem~\ref{thm:correlation}) to reduce to the one-body problem.

After we had announced our results in \cite{solovej-oberwolfach},
Frank, Siedentop, and Warzel~\cite{FSW} found a proof for the atomic
case based on the method of Siedentop and Weikard
\cite{Siedentop-Weikard1, Siedentop-Weikard2, Siedentop-Weikard3},
also \cite{FSW2} for the model studied in
\cite{Cassanas-Siedentop}. This approach seems to be restricted to the
spherical case. This work does also not, contrary to the present work,
make any special treatment of the non-relativistic limit or the
continuity of the function $\cS$. 

\subsection{Main semi-classical result}\label{sect:main s-c}

We consider the semi-classical approximation for the {\it
  relativistic} operator 
$$T_{\beta}(-\ic h\nabla)-V(\x)\,,$$
where 
\begin{align}\label{def:kinetic energy}
  T_{\beta}(p)=
  \begin{cases}
    \ \sqrt{\beta^{-1}p^2+\beta^{-2}}-\beta^{-1} \ ,&
    \beta\in(0,\infty)\\ 
    \qquad\  \frac12p^2 \ ,\qquad\quad\, & \beta=0
  \end{cases}.
\end{align}
We will consider potentials $V:\R^3\to\R$ with Coulomb singularities
of the form $z_k|x-r_k|^{-1}$, $k=1,\ldots,M$, at points
$r_1,\ldots,r_M\in\R^3$ and with charges
$0<z_1,\ldots,z_M\leq2/\pi$. Define 
\begin{equation}\label{ddefinitionBIS}
  d_{\bf r}(x)=\min\big\{|x-r_k|\ \big|\ k=1,\ldots,M\big\}\,
  \ , \quad {\bf r}=(r_1,\ldots,r_M)\in\R^{3M}\,.
\end{equation}
We assume that for some $\mu\geq0$ the potential \(V\) satisfies
\begin{equation}\label{eq:VcondD}
  \big|\partial^\eta \big(V(x)+\mu\big)\big|\leq
  \left\{\begin{array}{ll}C_{\eta,\mu} 
  d_{\bf r}(x)^{-1-|\eta|} &\hbox{ if }\mu\ne0\\ 
  C_\eta \min\{d_{\bf r}(x)^{-1},d_{\bf r}(x)^{-3}\}\,d_{\bf
    r}(x)^{-|\eta|} &\hbox{ if }\mu=0 
\end{array}\right.
\end{equation}
for all $x\in\R^3$ with $d_{\bf r}(x)\ne0$ and all multi-indices
$\eta$ with $|\eta |\leq 3$, and 
\begin{equation}\label{eq:VcondS}
  \big|V(x)-z_k|x-r_k|^{-1}\big|\leq Cr_{\rm min}^{-1}+C
\end{equation}
for $|x-r_k|<r_{\rm min}/2$ where $r_{\rm
  min}=\min_{k\ne\ell}|r_k-r_\ell|$. 
Note, in particular, that the Thomas-Fermi potential $\V({\bf z}, {\bf
  r}, \cdot)$ discussed in \eqref{eq:tfpotgeneral} below satisfies
these requirements, by Theorem~\ref{thm:tfestimate}. So does the
potential \(V(x)=\frac{2}{\pi|x|}-1\) (with \(M=1\), \(r_0=0\), and
\(d_{\bf r}(x)=|x|\)). 

The main new result in this section is the relativistic Scott
correction to the semi-classical expansion for potentials of this
form. It will be proved in Section~\ref{sect:semi-tf} below. The power
$-3$ in \eqref{eq:VcondD} is not optimal.  

\begin{thm}[{\bf Scott-corrected relativistic semi-classics}]
  \label{thm:TF-semicl}  
There exists a continuous, non-increasing function
\(\mathcal{S}:[0,2/\pi]\to\R\) with \(\mathcal{S}(0)=1/4\), such that
for all $h>0$, $0\leq\beta\leq h^2$, \(T_\b\) as in \eqref{def:kinetic
  energy},  and all potentials $V:\R^3\to\R $ satisfying
\eqref{eq:VcondD} and \eqref{eq:VcondS} with $r_{\min}>r_0>0$ and
$\max\{z_1,\ldots,z_M\}\leq 2/\pi$, we have
\bea\label{TF-Scott}\nonumber
  \Big|\Tr\big[T_{\beta}(-\ic h\nabla)-V(\x)\big]_- - (2\pi h)^{-3}
  \int \big[\mfr{1}{2}p^2 - V(v)\big]_-\, dvdp
  - h^{-2}\sum_{k=1}^Mz_k^2
  {\mathcal S}(\beta^{1/2}h^{-1}z_{k})
  \Big|\nonumber \\
  && \hspace{-6cm}\le 
  C h^{-2+1/10}\,.
\eea   
Here, $[x]_{-}=\min\{x,0\}$. The constant $C>0$ depends only on  $M$,
$r_0$, $\mu$ and the other constants in \eqref{eq:VcondD} and
\eqref{eq:VcondS}.  

Moreover, we can find a density matrix $\gamma$, whose density
\(\rho_{\gamma}\) satisfies (with \(\|\cdot\|_{6/5}\) the
\(L^{6/5}\)-norm) 
\begin{equation}\label{eq:rhogammaint}
  \left|\int\rho_\gamma(x)\,dx-2^{1/2}(3\pi^2)^{-1} h^{-3}\int 
  |V(x)_-|^{3/2}\,dx\right|\leq C h^{-2+1/5}
\end{equation}
and 
\begin{equation}\label{eq:rhogamma6/5}
  \big\|\rho_\gamma-2^{1/2}(3\pi^2)^{-1}
  h^{-3}|V_-|^{3/2}\big\|_{6/5}\leq C h^{-2-1/10}\,,
\end{equation}
such that 
\begin{eqnarray}\label{TF-Scott-trial}\nonumber
   \Tr\big[(T_{\beta}(-\ic h\nabla)-V(\x))\gamma\big]\leq 
   (2\pi h)^{-3} \int \big[\mfr{1}{2}p^2 - V(v)\big]_-\, dvdp 
   + h^{-2}\sum_{k=1}^Mz_k^2 {\mathcal S}(\beta^{1/2}h^{-1}z_{k})
   \nonumber \\
   && \hspace{-6cm}{}+ C h^{-2+1/10}\,.
\end{eqnarray}
\end{thm}
\begin{remark}\label{rem:non-rel-in-s-c}
The term proportional to \(h^{-2}\) is called the {\it Scott
  correction}. If \(\b=h^2\) then it only depends on the charges
\(z_k, k=1,\ldots,M\), of the Coulomb-singularities. Notice that the
function in the semi-classical integral is the {\it non-relativistic}
energy. This is also the reason why the leading Thomas-Fermi energy is
independent of \(\beta\). 
\end{remark}
Applying this theorem to the potential \(V(x)=\frac{2}{\pi|x|}-1\)
(which satisfies \eqref{eq:VcondD} and \eqref{eq:VcondS} with \(M=1\),
\(r_0=0\), and \(d_{\bf r}(x)=|x|\)), and using a simple scaling
argument, gives the following explicit characterization of the
function \(\mathcal{S}\) in Theorem~\ref{thm:TF-semicl} (see details
in Lemma~\ref{cor:alt-S} in Section~\ref{sect:semi-tf} below). 
\begin{cor}[{\bf Characterization of the Scott-correction
    \(\mathcal{S}\)}]\label{cor:chrarc-S} 
The function \(\mathcal{S}\) satisfies, uniformly for
$\alpha\in[0,2/\pi]$,
\begin{align}
  \mathcal S(\alpha)=\lim_{\kappa\to0}&\Big(\,
  \Tr\big[H_{\rm C}
  +\kappa\big]_- - (2\pi)^{-3}\int\big[\mfr12
  p^2-|v|^{-1}+\kappa\big]_-\,dpdv\,\Big)\,,
\end{align}
where
\begin{align}\label{def:H_C-BIS}
  H_{\rm C}(\alpha)=
  \begin{cases}
     \sqrt{-\alpha^{-2}\D+\alpha^{-4}}-\alpha^{-2}
     -|{\x}|^{-1}\ , & \alpha\in(0,2/\pi] \\
     {}-\frac12\D -|{\x}|^{-1} \ ,& \alpha=0
   \end{cases}\,.
\end{align}
\end{cor}
\begin{remark} Another characterization of the function
  \(\mathcal{S}\) is given in Lemma~\ref{lem:Coulomb} in
  Section~\ref{sect:semi-tf} below.
\end{remark}
\begin{remark}
  The result in Corollary~\ref{cor:chrarc-S} was proved in
  \cite[Theorem 7.4]{sor-thesis}, but only {\it pointwise} and only for
  \(\alpha\in(0,2/\pi)\). 
\end{remark}

\section{Preliminaries}\label{sect:prelim}

\subsection{Analytic tools}\label{sect:Analytic tools}

We recall here the main analytic tools which we use throughout this
paper. We do not prove all of them here but give the standard
references. Various constants are denoted by the same letter $C$
although its value may change from one line to the next.  

Let $p\ge1$, then a complex-valued function $f$ (and only those will
be considered here) is said to be in $L^p(\R^n)$ if the norm $\|f\|_p
= \left(\int |f(x)|^p \,dx\right)^{1/p}$ is finite. We denote by
\(\langle\ ,\ \rangle\) the inner product on \(L^2(\R^n)\); it is
linear in the second and anti-linear in the first entry. For any $1\le
p\le t\le q\le\infty$ we have the inclusion $L^p\cap L^q\subset L^t$,
since by H\"older's inequality $\| f\|_t \le
\|f\|_p^\lambda\|f\|_q^{1-\lambda}$ with $\lambda p^{-1}+(1-\lambda)
q^{-1}=t^{-1}$. We denote by \(\hat f\) the Fourier transform of
\(f\in L^2(\R^n)\), given by \(\hat f(p)=(2\pi)^{-n/2}\int {\rm
  e}^{-{\rm i}x\cdot p} f(x)\,dx\) for Schwartz functions on \(\R^n\), 
and extended by continuity to \(L^2(\R^n)\).

We denote $x_{-}=\min\{x,0\}$, and let $\chi_A$ be the characteristic
function of the set $A$; we write $\chi=\chi_{(-\infty,0]}$ for the
characteristic function of $(-\infty,0]$. We call $\c$ a density
matrix on $L^2(\R^n)$ if it is a trace class operator on $L^2(\R^n)$
satisfying the operator inequality ${\bf 0}\le\c\le {\bf 1}$. The
density of a density matrix $\gamma$ is the $L^1$-function
$\rho_\gamma$ such that
$\Tr(\gamma\theta)=\int\rho_\gamma(x)\theta(x)dx$ for all $\theta\in
C_0^\infty(\R^n)$ considered as a multiplication operator.

We also need an extension to many-particle states. Let $\psi\in
\bigotimes^N L^2(\R^3\times\{-1,1\})$ be an $N$-body wave-function.
Its one-particle density $\rho_\psi$ is defined by 
$$
  \rho_\psi(x) = \sum_{j=1}^N \sum_{s_1=\pm1}\cdots\sum_{s_N=\pm1}\int
  |\psi(x_1,s_1,\ldots,x_N,s_N)|^2\,\delta(x_j-x)\,dx_1\cdots dx_N \,.
$$

The following two inequalities we recall are crucial in many of our
estimates. They serve as replacements for the Lieb-Thirring inequality
\cite{Lieb-Thirring} used in the non-relativistic case.

\begin{thm}[{\bf Daubechies inequality}]\label{Daubechies}{\bf
    One-body case:}  Let $m>0$, $f(u)=
\sqrt{u^2+m^2}-m$, and $F(s)=
\int_0^s [f^{-1}(t)]^n\, dt$, where \(f^{-1}\) denotes the inverse
function of \(f\). Assume that $V\in L^1_{\rm loc}(\mathbb R^n)$, and
let \({}-\D\) be the Laplacian in \(\R^n\). Then
\bea\label{Daub inequBIS} 
   \Tr\big[\sqrt{-\D+m^2}-m + V(\x)\big]_{-} &\ge& {}-C \int
   F\big(|V(x)_{-}|\big)\,dx\,, 
\eea 
where $x_- = \min\{x,0\}$, and $C$ is some positive constant.

\noindent{\bf Many-body case:} Let
$\psi\in\bigwedge^NL^2(\R^3\times\{-1,+1\})$ and let $\rho_{\psi}=\rho$
be its one-particle density. Then
\bea \label{Daub manybodyBIS}
  \Big\langle\psi,\sum_{j=1}^N
  \big[\sqrt{-\Delta_j+m^2}-m\big]\psi\Big\rangle
  \ge \int G[\rho(x)]\,dx\,,
\eea
where (with \(C_0=0.326\))
\bea\label{eq:Daub functions}
  G(\rho) = (3/8)m^4C_0g[(\rho/C_0)^{1/3}m^{-1}]-m\rho\,,
\eea
with \(g(t)=t(1+t^2)^{1/2}(1+2t^2)-\log[t+(1+t^2)^{1/2}]\).
\end{thm}
The asymptotic behaviour of $G$ for small, respectively large $\rho$
is given by 
\be
   \label{behaviour G}
   G(\rho) \underset{\rho\to0}{\sim} (3/10m)C_0^{-2/3}\rho^{5/3}\, , \quad
   G(\rho) \underset{\rho\to\infty}{\sim} (3/4)C_0^{-1/3}\rho^{4/3}\,. 
\ee
By a simple scaling, and using  the definition of $T$ and
\eqref{behaviour G}, respectively, we see that 
\bea\label{eq:Daub usefull} 
  \lefteqn{\Tr\big[\sqrt{-\alpha^{-2}\D+m^2\alpha^{-4}}-m\alpha^{-2} +
  V(\x)\big]_-}
  \\
  &&\qquad \ge {}-Cm^{n/2}\int |V(x)_{-}|^{1+n/2}\,dx-C\alpha^{n}\int
  |V(x)_{-}|^{1+n}\,dx\,,
  \nonumber
\eea
and
\bea  \label{Daub manybodyBIS2}
  \Big\langle\psi,\sum_{j=1}^N
  \big[\sqrt{-\alpha^{-2}\Delta_j+m^2\alpha^{-4}}
  -m\alpha^{-2}\big]\psi\Big\rangle 
  \ge {}C\int \min\{m^{-1}\rho(x)^{5/3},\alpha^{-1}\rho(x)^{4/3}\}
  \,dx\,.
\eea
Both \eqref{eq:Daub usefull} and \eqref{Daub manybodyBIS2} also holds
for \(\alpha=0\), where we let
$\sqrt{-\alpha^{-2}\Delta+m^2\alpha^{-4}}-m\alpha^{-2} = -\Delta/2m$,
when $\alpha=0$. The original proofs of the inequalities \eqref{Daub
  inequBIS} and \eqref{Daub manybodyBIS} can be found in
\cite{Daubechies} (for \(\alpha=0\), in \cite{Lieb-Thirring}).

\begin{thm}[{\bf Lieb-Yau inequality}]\label{thm:L-Y} Let \(n=3\). Let $C>0$ and
$R>0$ and let
\be\label{H_{CR}} 
  H_{C,R} = \sqrt{-\D} - \frac{2}{\pi|{\hat x}|}  - C/R\,.
\ee
Then, for any density matrix $\gamma$ and any function $\theta$ with
support in $B_R = \{x\,|\,|x|\le R\}$ we have that 
\be 
  \Tr\big[\bar{\theta}\gamma \theta H_{C,R}\big] \ge -4.4827 \,C^4
  R^{-1} \{3/(4\pi R^3)  
  \int |\theta(x)|^2\, dx\}\,.
\ee
\end{thm}
Note that when $\theta=1$ on $B_R$ then the term inside the brackets
$\{ \}$ equals 1. 

We will need the following new operator inequality. The proof can be
found in Appendix~\ref{sect:App}. 
\begin{thm}[{\bf Critical
    Hydrogen inequality}]\label{thm:new-critical}
 Let \(n=3\). For any
  \(s\in[0,1/2)\) there exists constants \(A_s, B_s>0\) such that
\be\label{eq:new-critical}
    \sqrt{-\D}-\frac{2}{\pi|{\hat x}|}\geq A_s(-\Delta)^{s}-B_s\,.
\ee
\end{thm}

We also use the following standard notation for the Coulomb energy,
$$ D(f) = D(f,f) = \mfr{1}{2} \int \overline{f(x)}\,|x-y|^{-1} f(y)\, dx dy\,.
$$

\begin{thm}[{\bf Hardy-Littlewood-Sobolev inequality}] There exists a
  constant $C$ such that 
\be \label{Hardy-Littlewood-Sobolev} 
    D(f)\le C\,\| f \|_{6/5}^2\,.
\ee
\end{thm}

The sharp constant $C$ has been found by Lieb \cite{Lieb:sob}; see
also \cite{Lieb-Loss}. It can be shown by Fourier transformation that
$f\mapsto\sqrt{D(f)}$ is a norm. This fact will play a role in the
proof of the upper bound in our main Theorem~\ref{main theorem}. 

In order to localise the relativistic kinetic energy we shall use the
equivalent of the IMS-formula for the operator $-\D/2m$. In the
sequel, as before,
$\sqrt{-\alpha^{-2}\Delta+m^2\alpha^{-4}}-m\alpha^{-2} = -\Delta/2m$,
when $\alpha=0$.  
\begin{thm}[{\bf Relativistic IMS formula}] \label{IMS} Let
  $(\theta_u)_{u\in\mathcal M}$ be a family of positive bounded
  \(C^1\)-functions on $\R^3$ with bounded derivatives, and let $d\mu$
  be a positive measure on 
  $\mathcal M$ such that $\int_{\mathcal M} \theta_u(x)^2\,  
  d\mu(u)=1$ for all $x\in\R^3$. Then for any $f\in H^{1/2}(\R^3)$,
\bea \label{eq:IMS}
     \lefteqn{(f, \big(\sqrt{-\alpha^{-2}\D + m^2\alpha^{-4}}
       -m\alpha^{-2}\big) 
    f)}\\ \quad&=& \int_{\mathcal M} (\theta_u f, 
    \big(\sqrt{-\alpha^{-2}\D + m^2\alpha^{-4}}-m\alpha^{-2}\big) 
    \theta_u f)\, d\mu(u) - (f,L^{}f)\,,\nonumber
\eea
where the operator $L$ is of the form 
\begin{align} \label{IMS-error1} 
   L&=\int_{\mathcal M} L_{\theta_u}\,d\mu(u)\,, 
\end{align}
with \(L_{\theta_u}\) the bounded operator with kernel
\begin{align}
   L_{\theta_{u}}(x,y)&= (2\pi)^{-2}m^2\alpha^{-3} |x-y|^{-2}
   K_2(m\alpha^{-1}|x-y|)  
    \big[\theta_u(x)-\theta_u(y)\big]^2\,.
  \label{IMS-error2} 
\end{align}
Here, $K_2$ is a modified Bessel function of the second kind. For
\(\alpha=0\), \(L_{\theta_u}\) is multiplication by
\((\nabla\theta_{u})^2/2m\),  where
$\sqrt{-\alpha^{-2}\Delta+m^2\alpha^{-4}}-m\alpha^{-2} = -\Delta/2m$,
when $\alpha=0$. 
\end{thm}
A proof (and the definition of \(K_2\), and some of its properties)
can be found in Appendix~\ref{sect:App}. 

The following bound on the localisation error will be crucial.
\begin{thm}[{\bf Localisation error}] \label{IMS-error-est} 
Let \(\Omega\subset\R^3\) and \(\ell>0\). Let \(\theta\) be a
Lipschitz continuous function satisfying \(0\le\theta\le1\),
\(\dist(\Omega^{c},\supp\,\nabla\theta)\ge\ell\), and \(\theta\) is
constant on \(\Omega^c\). 

Then for all \(m>0\), \(\alpha\ge0\) there exists a positive operator
\(Q_\theta\)  such that the following operator inequality holds: 
\be\label{IMS-err-op-est}
  L_\theta\le C\,m^{-1}\|\nabla\theta\|_{\infty}^2\chi_{\Omega}+Q_\theta,
\ee
with
\be \label{est:traceQ}
  \Tr[Q_\theta]\le
  C\,m\alpha^{-2}\ell^{-1}{\rm e}^{-m\alpha^{-1}\ell}
  \|\nabla\theta\|_{\infty}^2|\Omega|\,,
\ee
for a constant \(C>0\), independent of \(m,\alpha,\ell, \theta\), and
\(\Omega\). Here, \(\chi_\Omega\) and $|\Omega|$ are the
characteristic function and the volume, respectively, of the set
\(\Omega\). For \(\alpha=0\), \(Q_{\theta}\equiv0\). 
\end{thm}
A proof can be found in Appendix~\ref{sect:App}. Note that the first
term, $C\,m^{-1}\|\nabla\theta\|_{\infty}^2\chi_{\Omega}$, on the
right side of \eqref{IMS-err-op-est} is similar to the error in the
non-relativistic IMS formula for the operator $-\D/2m$, except in this
case one has $\|\nabla\theta\|_{\infty}^2\chi_{{\rm supp}
  \nabla\theta}/2m$ as the only error. 

When localising, we shall make use of the following.
\begin{thm}[{\bf Partition of \(\R^n\)}] \label{partition}
Consider $\varphi\in C^\infty_0(\mathbb R^n)$ with support in the unit
ball $\{|x|\leq1\}$ and satisfying $\int\varphi(x)^2\,dx=1$. Assume
that $\ell:\R^n\to \R$ is a $C^1$-map satisfying $0<\ell(u)\leq1$ and
$\|\nabla\ell\|_\infty<1$. Let $J(x,u)$ be the Jacobian of the map
$u\mapsto\frac{x-u}{\ell(u)}$, i.e.,
$$
  J(x,u) =
  \ell(u)^{-n}\Big|\det\Big[\frac{(x_i-u_i)\partial_j\ell(u)}{\ell(u)}  
  + \delta_{ij}\Big]_{ij}\Big|\,.
$$
We set
$\varphi_{u}(x)=\varphi\big(\frac{x-u}{\ell(u)}\big)
\sqrt{J(x,u)}\,\ell(u)^{n/2}$.  
Then, for all $x\in \R^n$, 
\begin{equation}\label{eq:philoc}
  \int_{\R^n} \varphi_{u}(x)^2\,\ell(u)^{-n}\, du = 1 \,,
\end{equation}
and for all multi-indices $\eta\in\N^n$ we have
\begin{equation}\label{eq:phiuestimate}
  \|\p^\eta\varphi_{u}\|_\infty \leq \ell(u)^{-|\eta|}
  C_\eta\max_{|\nu|\leq|\eta|}\|\p^\nu\varphi\|_\infty \,,
\end{equation}
where $C_\eta$ depends only on $\eta$.
\end{thm}
This is Theorem 22 in \cite{SS}.

We will consider potentials $V:\R^3\to\R$ with Coulomb singularities
of the form $z_k|x-R_k|^{-1}$, $k=1,\ldots,M$, at points
$R_1,\ldots,R_M\in\R^3$ and with charges $0<z_1,\ldots,z_M\leq2/\pi$. 
Recall that (see \eqref{ddefinitionBIS}; replace \({\bf r}\) by \({\bf R}\))
\begin{equation}\label{ddefinition}
  d_{\bf R}(x)=\min\big\{|x-R_k|\ \big|\ k=1,\ldots,M\big\}\,
  \ , \quad {\bf R}=(R_1,\ldots,R_M)\in\R^{3M}\,.
\end{equation}
To treat such potentials we will need the following combination of
Theorems~\ref{Daubechies}  and \ref{thm:L-Y}. The proof can be found
in Appendix~\ref{sect:App}. 
\begin{thm}[{\bf Combined Daubechies-Lieb-Yau inequality}]\label{thm:L-Y-D}
Let \(R_1,\ldots,R_M\in\R^3\), and assume $W\in L^1_{\rm loc}(\R^3)$
satisfies 
\begin{equation}\label{eq:cond-D-L-Y}
   W(x)\geq{}-\frac{\nu}{d_{\bf R}(x)}-C\nu
   m\alpha^{-1}\quad\hbox{when}\quad d_{\bf R}(x)<\alpha m^{-1},
\end{equation}
with $\alpha\nu\leq2/\pi$ and $m>0$, $\alpha\geq0$, and \(d_{\bf R}\)
as in \eqref{ddefinition}. Assume also that the minimal distance
between nuclei satisfies $\min_{k\ne \ell}|R_k-R_\ell|>2\alpha
m^{-1}$. Then 
\bea\label{ineq:ImprovedDLY}\nonumber
  \lefteqn{\Tr\big[\sqrt{-\alpha^{-2}\Delta+m^2\alpha^{-4}}
   -m\alpha^{-2}+W({\hat x})\big]_{-}} 
   &&\\&\geq&
  {}-C\nu^{5/2}\alpha^{1/2}m
  -Cm^{3/2}\!\!\!\!\!\!\!\!
   \int\limits_{d_{\bf R}(x)>\alpha m^{-1}}
   \!\!\!\!\!\!\!\!|W(x)_-|^{5/2}\,dx 
  {}-C\alpha^{3}\!\!\!\!\!\!\!\!\int\limits_{d_{\bf R}(x)>\alpha
    m^{-1}}\!\!\!\!\!\!\!\!|W(x)_-|^{4}\,dx\,,
\eea
where as before $\sqrt{-\alpha^{-2}\Delta+m^2\alpha^{-4}}-m\alpha^{-2}
= -\Delta/2m$, when $\alpha=0$. 
\end{thm}
Finally, we come to the two inequalities which bound the many-body
ground state energy in terms of a corresponding one-body energy.

\begin{thm}[{\bf Correlation inequality}] \label{thm:correlation}
Let $\rho:\R^3\to\R$ be non-negative with  $D(\rho)<\infty$ and let
$\Phi:\mathbb R^3\to\mathbb R$ be a spherically symmetric,
non-negative function with support in the unit ball such that $\int
\Phi(x)\,dx<\infty$. For $s>0$, let $\Phi_s(x) =
s^{-3}\Phi(x/s)$. Then, for some constant $C$ independent of $N$ and
$s$, we have\footnote{We denote convolution by $*$, i.e., $(f*g)(x) =
  \int f(y) g(x-y)\, dy$. We also abuse notation and write
  $\rho*|x|^{-1}$ instead of $\big(\rho*|\cdot|^{-1}\big)(x)$.} 
\be\label{inequ:correlation} 
  \sum_{1\le i<j\le N} |x_i -x_j|^{-1} 
  \ge \sum_{j=1}^N (\rho * |x|^{-1} * \Phi_s)(x_j) - D(\rho) - C N
  s^{-1} \,. 
\ee
\end{thm}
The proof can be found in Appendix~\ref{sect:App}.

\begin{thm}[{\bf Lieb's Variational Principle}] \label{Lieb's
    Variational Principle} 
Let $\gamma$ be a density matrix on $L^2(\R^3)$ satisfying $2\,\Tr\c=
2\int \rho_\c(x)\, dx \leq Z$ (i.e., less than or equal to the total
number of electrons) with kernel $\rho_\c(x)=\c(x,x)$. Then  
\be\label{variational principle} 
  E({\bf Z},{\bf R};\a)\le
  2\,\Tr\big[\big(\sqrt{-\a^{-2}\D + \a^{-4}} 
  -\a^{-2} - V({\bf Z},{\bf R},\x)\big)\c\big] + D(2\rho_\gamma) \,.  
\ee 
\end{thm}
The factors 2 above are due to the spin degeneracy, see~\cite{Lieb2}.

\subsection{Thomas-Fermi theory }\label{sec:tf}

Consider ${\mathbf z}=(z_1,\ldots,z_M)\in\R_+^M$ and ${\mathbf
  r}=(r_1,\ldots,r_M)\in\R^{3M}$.  Let $0\le\rho\in L^{5/3}(\mathbb
R^3)\cap L^1(\mathbb R^3)$ then the (non-relativistic) Thomas-Fermi
(TF) energy functional, ${\mathcal E}^{{\rm TF}}$, is defined as 
$$ 
  {\mathcal E}^{{\rm TF}}(\rho) = \mfr{3}{10}(3\pi^2)^{2/3} \int
  \rho(x)^{5/3}\, dx - \int V({\mathbf z},{\mathbf r},x)\rho(x)\, dx+
  D(\rho) \,, 
$$
where $V$ is as in (\ref{V}).

By the Hardy-Littlewood-Sobolev inequality the Coulomb energy,
$D(\rho)$, is finite for functions $\rho\in L^{5/3}(\mathbb R^3)\cap
L^1(\mathbb R^3)\subset L^{6/5}(\mathbb R^3)$. Therefore, the
TF-energy functional is well-defined. Here we only state without proof
the properties about TF-theory which we use throughout the paper. The
original proofs can be found in \cite{Lieb-Simon} and \cite{Lieb1}.  

\begin{thm}[{\bf Thomas-Fermi minimizer}] 
For all ${\mathbf z}=(z_1,\ldots,z_M)\in \R_+^M$ and ${\mathbf
  r}=(r_1,\ldots,r_M)\in\R^{3M}$ there exists a unique non-negative
$\rho^{{\rm TF}}({\mathbf z},{\mathbf r},x)$ such that $\int
\rho^{{\rm TF}}({\mathbf z},{\mathbf r},x)\,dx=\sum_{k=1}^Mz_k$ and
$$
  {\mathcal E}^{{\rm TF}}(\rho^{{\rm TF}}) = \inf
  \big\{{\mathcal E}^{{\rm TF}} (\rho)\ \big| \ 0\le\rho\in
  L^{5/3}(\R^3)\cap L^1(\R^3)\big\} \,.
$$
We shall denote by $E^{{\rm TF}}({\mathbf z},{\mathbf r})={\mathcal
  E}^{{\rm TF}}(\rho^{{\rm TF}})$ the TF-energy. Moreover, let
\begin{equation}\label{eq:tfpotgeneral}
  V^{{\rm TF}}({\mathbf z},{\mathbf r},x) = V({\mathbf z},{\mathbf
    r},x) -\rho^{{\rm TF}}({\mathbf z},{\mathbf r},\cdot) * |x|^{-1}\,
\end{equation}
be the TF-potential, then $V^{{\rm TF}}>0$ and $\rho^{{\rm TF}}>0$,
and $\rho^{{\rm TF}}$ is the unique solution in $L^{5/3}(\R^3)\cap
L^1(\R^3)$ to the TF-equation:
\begin{equation}\label{eq:tfeqgeneral} 
   V^{{\rm TF}}({\mathbf z},{\mathbf r},x) =
   \mfr{1}{2}(3\pi^2)^{2/3}
   \rho^{{\rm TF}}({\mathbf z},{\mathbf r},x)^{2/3}\,.
\end{equation}
\end{thm}

Very crucial for a semi-classical approach is the {\it scaling}
behavior of the TF-potential.  It says that for any positive parameter
$h$, 
\begin{eqnarray}\label{scaling} 
  V^{{\rm TF}}({\mathbf z},{\mathbf r},x) &=&
  h^{4} V^{{\rm TF}}(h^{-3}{\mathbf z},h{\mathbf r},hx)\,,
  \\
  \label{scaling:rho}
  \rho^{{\rm TF}}({\mathbf z},{\mathbf r},x) &=& h^{6}\rho^{{\rm
      TF}}(h^{-3}{\mathbf z},h{\mathbf r},hx)\,,
  \\ 
  \label{scaling:E-TF}
  E^{\rm TF}({\mathbf z},{\mathbf r})&=&
  h^{7}E^{\rm TF}(h^{-3}{\mathbf z},h{\mathbf r})\,.  
\end{eqnarray}
By $h{\mathbf r}$ we mean that each coordinate is scaled by $h$, and
likewise for $h^{-3}{\mathbf z}$ and $hx$. By the TF-equation
\eqref{eq:tfeqgeneral}, the equations \eqref{scaling} and
\eqref{scaling:rho} are obviously equivalent. Notice that the
Coulomb-potential (the potential $V$ in \eqref{V}) has the claimed
scaling behavior. The rest follows from the uniqueness of the solution
of the TF-energy functional. 

We shall now establish the crucial estimates that we need about the
TF-potential. For each $k=1,\ldots,M$ we define the function
\begin{equation}\label{eq:Wdefinition}
  W_k({\mathbf z},{\mathbf r},x)=V^{\rm TF}({\mathbf z},{\mathbf r},x)
  -z_k|x-r_k|^{-1}\,.
\end{equation}
The function $W_k$ can be continuously extended to $x=r_k$. 
 
The first estimate in the next theorem is very similar to
a corresponding estimate in \cite{Ivrii-Sigal} (recall that 
the function \(d_{\bf r}\) was defined in \eqref{ddefinitionBIS}). 
\begin{thm}[{\bf Estimate on $V^{\rm TF}$}]\label{thm:tfestimate}
Let ${\mathbf z}=(z_1,\ldots,z_M)\in \R_+^M$ and
${\mathbf r}=(r_1,\ldots,r_M)\in \R^{3M}$.
For all multi-indices $\eta\in\mathbb N^3$ and all $x$ with $d_{\bf
  r}(x)\ne0$ we have 
\begin{equation}\label{eq:tfdf}
  \left|\partial^\eta_x\V({\mathbf z},{\mathbf r},x)\right|\leq C_\eta
  \min\{d_{\bf r}(x)^{-1}, d_{\bf r}(x)^{-4}\}\,
  d_{\bf r}(x)^{-|\eta|}\,, 
\end{equation}
where $C_\eta>0$ is a constant which depends on $\eta$, $z_1,\ldots,z_M$,
and $M$.

Moreover, for $|x-r_k|<r_{\min}/2$, where $r_{\min}=\min_{k\ne
  \ell}|r_k-r_\ell|$, we have  
\begin{equation}\label{eq:westimate}
  {}-C\leq W_k({\mathbf z},{\mathbf r},x)
  \leq Cr_{\min}^{-1}+C \,,
\end{equation}
where the constants $C>0$ here depend on $z_1,\ldots,z_M$, and $M$. 
\end{thm}
\begin{cor}[{\bf Estimate on $\rho^{\rm
      TF}*|x|^{-1}*(\delta_0-\Phi_t)$}]\label{estimate related to
    correlation} 
Let $\Phi:\mathbb R^3\to\mathbb R$ be a spherically symmetric,
positive function with support in the unit ball and integral 1, and
for $t>0$, let $\Phi_t(x) = t^{-3}\Phi(x/t)$. 
If $\rho^{{\rm TF}} (x)= \rho^{{\rm TF}}({\mathbf z},{\mathbf r},x)$ then
\be \label{correl:upper bound}
    0\,\le \,\rho^{{\rm TF}} * |x|^{-1} - \rho^{{\rm TF}} * |x|^{-1} * 
    \Phi_t \,\le \, 
    \left\{\begin{array}{cl} C\,t
    \min\{d_{\bf r}(x)^{-1/2}, d_{\bf r}(x)^{-2}\}\,
    & \text{for }\ d_{\bf r}(x) \ge 2t\\C\,
    t^{1/2}\,& \text{for }\ d_{\bf r}(x)<2t\end{array}\right. 
\ee    
with the function \(d_{\bf r}\) from \eqref{ddefinitionBIS}, and some
constant $C>0$ depending on $z_1,\ldots,z_M$, and $M$.
\end{cor}

For the proof of \eqref{eq:tfdf} and \eqref{eq:westimate} we refer
to~\cite{SS}. (Note that in \cite{SS} it is claimed that
\(W_k({\mathbf z},{\mathbf r},x)\ge0\). This is not correct, but the
proof in \cite{SS} does give that \(W_k({\mathbf z},{\mathbf r},x)\ge{}-C\).)
The proof of \eqref{correl:upper bound} can be found in
Appendix~\ref{sect:App}.  
\begin{remark}\label{rem:uniform-in-z_k}
As is seen from the proofs in \cite{SS} and in
Appendix~\ref{sect:App}, the constants in Theorem~\ref{thm:tfestimate}
and Corollary~\ref{estimate related to correlation} only depend on
\(z_0>0\) when \(z_1,\ldots,z_M\in(0,z_0]\). 
\end{remark}

The relation of Thomas-Fermi theory to semi-classical analysis is that
the semi-classical density of a gas of non-interacting
(non-relativistic) electrons moving in the Thomas-Fermi potential $\V$
is simply the Thomas-Fermi density. More precisely, the semi-classical
approximation to the density of the projection onto the eigenspace
corresponding to the negative eigenvalues of the Hamiltonian
$-\frac{1}{2}\Delta-\V$ is 
\bea\label{eq:sc=tf-density}
  2 \!\!\!\!\!\!\!\!\!\!\!\!\!\!
  \int\limits_{\frac{1}{2}p^2-\V({\mathbf z},{\mathbf r},x)\le0}\,
  \!\!\!\!\!\!\!\!\!\!\!\!\!\!\!\!
  \frac{dp}{(2\pi)^3} \,= \,2^{3/2} (3\pi^2)^{-1} (\V)^{3/2}({\mathbf
  z},{\mathbf r},x) \,= \,\rho^{\rm TF}({\mathbf z},{\mathbf r},x)\,.
\eea
Here the factor two on the very left is due to the spin degeneracy.
Similarly, the semi-classical approximation to the energy of the gas,
i.e., to the sum of the negative eigenvalues of
$-\frac{1}{2}\Delta-\V$, is  
\bea\label{eq:sc=tf}\nonumber
  2\int \big[\mfr{1}{2}p^2-\V({\mathbf z},{\mathbf r},x)\big]_-\,
  \frac{dxdp}{(2\pi)^3}&=&{} -{\mfr {4\sqrt {2}}{15{\pi }^{2}}}\int
  \,\V({\mathbf z},{\mathbf r},x)^{5/2}\,dx\\ &= &E^{{\rm TF}}({\mathbf
    z},{\mathbf r}) + D\big(\rho^{{\rm TF}}({\mathbf z}, {\mathbf
    r},\cdot)\big)\,.
\eea
Since (by Theorem~\ref{thm:tfestimate}) the Thomas-Fermi potential
$\V({\bf z}, {\bf  r}, \cdot)$ in \eqref{eq:tfpotgeneral} satisfies
\eqref{eq:VcondD} and  \eqref{eq:VcondS} 
(uniformly for
\(z_1,\ldots,z_M\in(0,2/\pi]\); see Remark~\ref{rem:uniform-in-z_k}),
Theorem~\ref{thm:TF-semicl} 
implies that the density given in \eqref{eq:sc=tf-density} and the
energy given in 
\eqref{eq:sc=tf} are the leading order terms also for the {\it
  relativistic} gas, i.e., for the 
operator $T_{\b}(-\ic h\nabla)-V^{\rm TF}$, \(0\le\b\le h^2\), with
\(T_\b\) as in \eqref{def:kinetic energy}. That the Thomas-Fermi
energy is correct to leading order for $T_{h^2}(-\ic h\nabla)-V^{\rm
  TF}$ was proved in \cite{sor}. Theorem~\ref{thm:TF-semicl}
establishes the first correction---the Scott correction---to the
leading order. 

  \section{Proof of the relativistic Scott correction for the
    molecular ground state energy} 
  \label{main proof}

In this section we prove Theorem~\ref{main theorem}. Except for the
correlation inequality we proceed in exactly the same manner as in the
non-relativistic case \cite{SS}. In \cite{SS} correlations were
controlled by the Lieb-Oxford inequality \cite{Lieb-Oxford}. Applying
this inequality, correlations can be estimated by the integral $\int
\rho^{4/3}$ involving the electronic density $\rho$. Using the
non-relativistic Lieb-Thirring inequality such an integral can be seen
to be of lower order than the total energy. In the present
relativistic case the Daubechies inequality \eqref{Daub manybodyBIS}
{\it a priori} only allows us to conclude that the integral
$\int\rho^{4/3}$ is of the same order as the total energy. We
therefore follow a different strategy.

\begin{proof}[Proof of Theorem~\ref{main theorem} {\rm ({\bf
      Lower bound})}]
Let $\psi$ be a (normalised) ground state wave function and let
$s>0$. We will use the correlation inequality
\eqref{inequ:correlation} with $\rho(x)=\rho^{\rm TF}({\bf Z},{\bf
  R},x)$. Let $\Phi_s$ be a function as in
Theorem~\ref{thm:correlation}. We shall choose $s=Z^{-5/6}$.  

As above (see \eqref{ddefinitionBIS}) we have $d_{\bf
  r}(x)=\min\big\{|x-r_k|\ \big|\ k=1,\ldots,M\big\}$. Note that for
the physical positions of the nuclei we then have 
$$
  d_{\bf R}(x)=\min\big\{|x-R_k|\ \big|\
  k=1,\ldots,M\big\}=Z^{-1/3}d_{\bf r}(Z^{1/3}x)\,.
$$

{F}rom the estimate in \eqref{correl:upper bound} with $t=Z^{1/3}s$ we
obtain from the Thomas-Fermi scaling \eqref{scaling:rho} that
$$
  \big|\rho^{\rm TF}({\bf Z},{\bf R},\cdot)*|x|^{-1}
  -\rho^{\rm TF}({\bf Z},{\bf R},\cdot)*|x|^{-1}*\Phi_s(x)\big|
  \leq C Z^{3/2} s (g(x)+Z^{1/6})\,,
$$
where
$$
  g(x)=
  \left\{\begin{array}{ll}
        (2s)^{-1/2} &\hbox{if } d_{\bf R}(x)<2s\\
        d_{\bf R}(x)^{-1/2}&\hbox{if } 2s\leq d_{\bf R}(x)\leq Z^{-1/3}\\
        0&\hbox{if } Z^{-1/3}<d_{\bf R}(x)
    \end{array}\,.
  \right.
$$
We find from the correlation estimate (Theorem~\ref{thm:correlation})
that 
\begin{eqnarray} 
  \lefteqn{\big\langle \psi,H({\bf Z},{\bf R};\alpha)\psi\big\rangle 
  }
  \nonumber\\
  &\ge&\sum_{j=1}^{Z} \big\langle \psi,\big[\sqrt{-\a^{-2}\D_j +\a^{-4}}
  -\a^{-2} - V({\bf Z},{\bf R},\x_j)\big]
  \psi\big\rangle 
  \nonumber\\
  && {}+ \,\big\langle\psi,\sum_{j=1}^{Z} 
  \big(\rho^{{\rm TF}}({\bf Z},{\bf R},\cdot) * |x|^{-1} * \Phi_s)(\x_j)
  \psi\big\rangle
  - D\big(\rho^{{\rm TF}}({\bf Z},{\bf R},\cdot)\big) - C s^{-1}Z
  \nonumber\\
  &\ge&2\,\Tr\big[\sqrt{-\a^{-2}\D +\a^{-4}} -\a^{-2} - V^{{\rm TF}}({\bf
  Z},{\bf R},\x)-C Z^{3/2} s g(\x)\big]_- 
  - D\big(\rho^{{\rm TF}}({\bf Z},{\bf R},\cdot)\big)
  \nonumber\\     
  &&{}-CsZ^{8/3}
  - Cs^{-1}Z\,.\label{eq:HZR-lowerbound}
\end{eqnarray}

To control the error term with $g$ above we shall use the combined
Daubechies-Lieb-Yau inequality (Theorem~\ref{thm:L-Y-D}) to estimate
$$
  \varepsilon\,\Tr\big[\sqrt{-\a^{-2}\D +\a^{-4}} -\a^{-2} -
  V^{{\rm TF}}({\bf Z},{\bf R},\x)-C \varepsilon^{-1}
  Z^{3/2} s
  g(\x)\big]_-
$$
for some $0<\varepsilon<1$ which we will choose to be
$\varepsilon=Z^{-1/2}$. We use Theorem~\ref{thm:L-Y-D} with $m=1$ and
$\nu=\max_kZ_k$. Then by assumption $\nu\alpha\leq2/\pi$. We must also
check that the assumption \eqref{eq:cond-D-L-Y} is satisfied, i.e.,
that for $d_{\bf R}(x)<\alpha$ we have 
$$
  {}-V^{{\rm TF}}({\bf Z},{\bf R},x)-C \varepsilon^{-1}Z^{3/2} s 
  g(x)\geq
  {}-\frac{\nu}{d_{\bf R}(x)}-C\nu\alpha^{-1}\,.
$$
This follows from the definition of $g$ and the estimate on the TF
potential in \eqref{eq:westimate} together with the Thomas-Fermi
scaling 
\eqref{scaling} if
$$
  \mfr12\a< s <C(\nu Z^{-1})^2(Z\alpha)^{-2}\varepsilon^2Z,\qquad
   r_{\min}^{-1}+1<C(\nu Z^{-1})(Z\a)^{-1}Z^{2/3},
$$
which, for $Z$ large enough, is a consequence of the assumptions in
the theorem and the choices of $\varepsilon$ and $s$. Note, in
particular, that $\nu Z^{-1}=\max_k{z_k}\geq M^{-1}$ (since \(\sum_k
z_k=1\)) and by assumption $Z\a \leq\min_k{2/(\pi z_k)}\leq
2M/\pi$. The constants $C$ above depend only on $z_1,\ldots,z_M$, and
$M$. 

According to the Thomas-Fermi estimate \eqref{eq:tfdf}, the
Thomas-Fermi scaling \eqref{scaling}, the definition of $g$, and the
choices of $s$ and $\varepsilon$ we have  
$$
  V^{{\rm TF}}({\bf Z},{\bf R},x)+C \varepsilon^{-1}Z^{3/2} s
  g(x)\leq C\min\{d_{\bf R}(x)^{-4},Zd_{\bf R}(x)^{-1}\}\,.
$$
Thus the combined Daubechies-Lieb-Yau inequality gives, since $\nu\leq
Z$ and $Z\alpha\leq2M/\pi$, that 
\begin{eqnarray*}
  \lefteqn{\varepsilon\,\Tr\big[\sqrt{-\a^{-2}\D +\a^{-4}} -\a^{-2} -
  V^{{\rm TF}}({\bf Z},{\bf R},\x)-C \varepsilon^{-1}
  Z^{3/2} s g(\x)]_-}\\
  &\geq&{}-C \varepsilon Z^2-C\varepsilon \int
  \big(\min\{d_{\bf R}(x)^{-4},Zd_{\bf R}(x)^{-1}\}\big)^{5/2}dx\\&&
  {}-C\varepsilon \a^3\int_{d_{\bf R}(x)>\a}\big(\min\{d_{\bf R}(x)^{-4},
  Zd_{\bf R}(x)^{-1}\}\big)^{4}dx\\
  &\geq&{}-C \varepsilon
  (Z^2+Z^{7/3})\geq{}-C\varepsilon Z^{7/3}\,.
\end{eqnarray*}

We return to the main term in \eqref{eq:HZR-lowerbound}. Using the
Thomas-Fermi scaling property \eqref{scaling} and replacing $x$ by
$Z^{-1/3}x$ we arrive at
\begin{eqnarray*}
  \lefteqn{\Tr\big[\sqrt{-\a^{-2}\D +\a^{-4}} -\a^{-2} - V^{{\rm
        TF}}({\bf Z},{\bf R},\x)\big]_- }
  \\
  &=&Z^{4/3}\kappa^{-1}
  \,\Tr\big[\sqrt{-\beta^{-1} h^2\D+\beta^{-2}}-\beta^{-1} 
  - \kappa V^{{\rm TF}}({\bf
  z},{\bf r},\x)]_{-}\,, 
\end{eqnarray*}
where we have chosen
\begin{equation}\label{eq:kappabetahchoice}
  \kappa=\min_k\frac{2}{\pi z_k}\geq Z\alpha\,,\qquad
  h=\kappa^{1/2} Z^{-1/3}\,,\qquad\beta=Z^{4/3}\alpha^2\kappa^{-1}\,.
\end{equation}
We shall use $\beta$ and $h$ as the semi-classical parameters 
when we apply Theorem~\ref{thm:TF-semicl}. It is therefore
important that $\beta\leq h^2$. This follows since
$\beta^{-1} h^2=(Z\alpha)^{-2}\kappa^2\geq1$. Note also that
\(2/\pi\le\kappa\le2M/\pi\) since \(z_k\le1\), \(k=1,\ldots,M\), and
\(\sum_kz_k=1\). 

Putting this together with the estimate above into
\eqref{eq:HZR-lowerbound} we obtain (using the inequality
$\Tr[X+Y]_-\geq\Tr[X]_-+\Tr[Y]_-$ for operators $X$ and $Y$ bounded
from below (with a common core), and the choices of $\varepsilon$ and
$s$) that 
\begin{eqnarray*}
  \lefteqn{\big\langle \psi,H({\bf Z},{\bf R};\alpha)\psi\big\rangle 
  }&&\\ &\geq&
  2(1-\varepsilon)Z^{4/3}\kappa^{-1}\,\Tr\big[\sqrt{-\beta^{-1}h^2\D
    +\beta^{-2}} -\beta^{-1} - \kappa V^{{\rm TF}}({\bf z},{\bf
    r},\x)\big]_-\\ &&{}-C\varepsilon Z^{7/3}-CsZ^{8/3}
  - Cs^{-1}Z- D\big(\rho^{{\rm TF}}({\bf Z},{\bf R},\cdot)\big)\\
  &\geq&2Z^{4/3}\kappa^{-1}\,\Tr\big[\sqrt{-\beta^{-1}h^2\D +\beta^{-2}}
  -\beta^{-1} - \kappa V^{{\rm TF}}({\bf z},{\bf r},\x)\big]_- 
  \\&&{}- D\big(\rho^{{\rm TF}}({\bf Z},{\bf R},\cdot)\big)-CZ^{2-1/6}\,.
\end{eqnarray*}
Now we apply the semi-classical result for potentials with
Coulomb-like singularities from Theorem~\ref{thm:TF-semicl} to
\(\kappa V^{{\rm TF}}({\bf z},{\bf r},\cdot)\) (recall that
\(2/\pi\le\kappa\le2M/\pi\) which ensures that the constants in
\eqref{eq:VcondD} and \eqref{eq:VcondS} are uniform in \(\kappa\)),
and the calculation in \eqref{eq:sc=tf}. 
Then 
\beax
  \lefteqn{2Z^{4/3}\kappa^{-1}\,\Tr\big[\sqrt{-\beta^{-1}h^2\D
  +\beta^{-2}} -\beta^{-1} - \kappa V^{{\rm TF}}({\bf z},{\bf
  r},\x)\big]_- 
  }\\
  &=&Z^{7/3}\Big(E^{{\rm TF}}({\mathbf z},{\mathbf r}) +
  D\big(\rho^{{\rm TF}}({\mathbf z}, {\mathbf r},\cdot)\big) \Big) 
  +2 \sum_{k=1}^M Z_k^2{\mathcal S}(Z_k\alpha)+ {\mathcal
    O}(Z^{2-1/30}) 
  \\
  &=& E^{{\rm TF}}({\bf Z},{\bf R}) + D\big(\rho^{{\rm TF}}({\bf
    Z},{\bf R},\cdot)\big) + 2\sum_{k=1}^M Z_k^2{\mathcal
    S}(Z_k\alpha)+ {\mathcal O}(Z^{2-1/30}) \,. 
\eeax 
Note here that $\kappa$ cancels in the leading semi-classical term and
in the Scott-term (the term with $\mathcal S$). 
Also, \(2/\pi\le\kappa\le2M/\pi\) ensures that the error is uniform in
\(\kappa\). 
Here we have again used the TF scaling $E^{{\rm TF}}({\bf Z},{\bf
  R})=Z^{7/3}E^{{\rm TF}}({\mathbf z}, {\mathbf r})$ and $D\big(\rho^{{\rm
    TF}}({\bf Z},{\bf R},\cdot)\big)=Z^{7/3}D\big(\rho^{{\rm TF}}({\mathbf z},
{\mathbf r},\cdot)\big)$. This finishes the proof of the lower bound.
\end{proof}
\begin{proof}[Proof of Theorem~\ref{main theorem} {\rm ({\bf
      Upper bound})}]
The starting point now is Lieb's variational principle,
Theorem~\ref{Lieb's Variational Principle}. By a simple rescaling the
variational principle states that for any density matrix $\gamma$ on
$L^2(\R^3)$ with $2\,\Tr\gamma\leq Z$ we have 
\begin{eqnarray*}
  E({\bf Z},{\bf R};\alpha)&\leq& 2Z^{4/3}\,\Tr\big[\big(
        \sqrt{-\alpha^{-2}Z^{-2}\Delta+\alpha^{-4}Z^{-8/3}}-\alpha^{-2}Z^{-4/3}
      -V({\bf z},{\bf r},\x)\big)\gamma\big]\\&&
  {}+ Z^{7/3}D(2 Z^{-1}\rho_\gamma)\,.
\end{eqnarray*}
As for the lower bound we bring the TF-potential 
into play:
\begin{eqnarray}
  Z^{-4/3}E({\bf Z},{\bf R};\alpha)&\leq&2\,\Tr\big[
    \big(\sqrt{-\alpha^{-2}Z^{-2}\Delta+\alpha^{-4}Z^{-8/3}}-\alpha^{-2}Z^{-4/3}
      -\V({\bf z},{\bf r},\x)\big)\gamma\big]
  \nonumber\\
  &&{}+ZD\big(2Z^{-1}\rho_\gamma-\rho^{\rm TF}({\bf z},{\bf r},\cdot)\big)
  -ZD\big(\rho^{\rm TF}({\bf z},{\bf r},\cdot)\big)\nonumber\\
  &=&2\kappa^{-1}\,\Tr\big[
    \big(\sqrt{-\beta^{-1}h^2\Delta+\beta^{-2}}-\beta^{-1}
      -\kappa \V({\bf z},{\bf r},\x)\big)\gamma\big]
  \nonumber\\
  &&{}+ZD\big(2Z^{-1}\rho_\gamma-\rho^{\rm TF}({\bf z},{\bf r},\cdot)\big)
  -ZD\big(\rho^{\rm TF}({\bf z},{\bf r},\cdot)\big)\,,\label{eq:upper1}
\end{eqnarray}
where $\kappa$, $h$, and $\beta$ are chosen as in
\eqref{eq:kappabetahchoice} in the proof of the lower bound.  Note
that with this choice of $h$ and $\kappa$ we have from
\eqref{eq:tfeqgeneral} that
$$
  2^{1/2}(3\pi^2h^3)^{-1}(\kappa\V({\bf z},{\bf r},x))^{3/2}
  =Z\rho^{\rm TF}({\bf z},{\bf r},x)/2\,.
$$
We now choose a density matrix $\widetilde{\gamma}$ 
according to Theorem~\ref{thm:TF-semicl} with $V(x)=\kappa \V({\bf
  z},{\bf r},x)$. 

Since $\int\rho^{\rm TF}({\bf z},{\bf r},x)\,dx=\sum_{k=1}^Mz_k=1$ we
see from \eqref{eq:rhogammaint} that $2\,\Tr\widetilde\gamma\leq
Z(1+CZ^{-1/3-1/15})$ (recall that \(\kappa^{-1}\le\pi/2\)). Thus if we
define $\gamma=(1+CZ^{-1/3-1/15})^{-1}\widetilde{\gamma}$ we see that
the condition $2\,\Tr\gamma\leq Z$ is satisfied. 

Using the Hardy-Littlewood-Sobolev and \eqref{eq:rhogamma6/5}
inequalities we find that 
$$
  ZD\big(2Z^{-1}\rho_{\widetilde\gamma}-\rho^{\rm TF}({\bf z},{\bf
    r},\cdot)\big) \leq
  CZ^{-1}\big\|\rho_{\widetilde\gamma}-Z\rho^{\rm TF}({\bf z},{\bf 
    r},\cdot)/2\big\|_{6/5}^2 \leq CZ^{2/3-4/15}\,,
$$
and thus 
\begin{eqnarray}
  ZD\big(2Z^{-1}\rho_{\gamma}-\rho^{\rm TF}({\bf z},{\bf
    r},\cdot)\big) &\leq&
  C(1+CZ^{-1/3-1/15})^{-2}Z
  D\big(2Z^{-1}\rho_{\widetilde\gamma}-\rho^{\rm TF}({\bf z},{\bf
    r},\cdot)\big)
  \nonumber\\&&{}+
  CZ^{1/3-2/15}D\big(\rho^{\rm TF}({\bf z},{\bf
      r},\cdot)\big)\leq C Z^{2/3-4/15}\,,\label{eq:upper2}
\end{eqnarray}
where we have used the triangle inequality for $\sqrt{D}$, and
that $D\big(\rho^{\rm TF}({\bf z},{\bf r},\cdot)\big)\leq C$.

Finally, if we use \eqref{TF-Scott-trial} and \eqref{eq:sc=tf} we get
as for the lower bound that
\begin{eqnarray*}
  \lefteqn{2Z^{4/3}\kappa^{-1}\,\Tr\big[
  \big(\sqrt{-\beta^{-1}h^2\Delta+\beta^{-2}}-\beta^{-1}
  -\kappa \V({\bf z},{\bf r},\x)\big)\widetilde\gamma\big]}&&\\
  &\leq&E^{{\rm TF}}({\bf Z},{\bf R}) + D\big(\rho^{{\rm TF}}({\bf
  Z},{\bf R},\cdot)\big)
  + 2\sum_{j=1}^M Z_k^2{\mathcal S}(Z_k\alpha)+ {\mathcal O}(Z^{2-1/30})\,.
\end{eqnarray*}
Since $E^{\rm TF}({\bf Z},{\bf R})\geq {}-C Z^{7/3}$ and 
$D\big(\rho^{\rm TF}({\bf Z},{\bf R},\cdot)\big)\ge0$ we see that the
same estimate holds for $\widetilde\gamma$ replaced by $\gamma$. This
together with $D\big(\rho^{\rm TF}({\bf Z},{\bf R},\cdot)\big)
=Z^{7/3}D\big(\rho^{\rm TF}({\bf z},{\bf r},\cdot)\big)$,
\eqref{eq:upper1}, and \eqref{eq:upper2} finishes the proof of the
upper bound.
\end{proof}
The function \(\mathcal{S}\) is continuous and non-increasing, and
\(\mathcal{S}(0)=1/4\), 
according to Theorem~\ref{thm:TF-semicl}. 

This finishes the proof of Theorem~\ref{main theorem}. 

\section{Relativistic semi-classics for potentials with Coulomb-like
  singularities}
\label{sect:semi-tf}

In this section we prove Theorem~\ref{thm:TF-semicl}. The theorem will
follow from using Theorem~\ref{corollary} below (a rescaled version of
the local semi-classical results for {\it regular} potentials in
Theorem~\ref{local semi-classics} in Section~\ref{Local semi-classical
  estimates} below). We localise (Theorem~\ref{IMS}) the operator
using multi-scale analysis (Theorem~\ref{partition}), and control the
localisation errors (Theorem~\ref{thm:L-Y-D}). Near the 
singularities of the potential, we compare with the Coulomb
potential. To be able to do this, we first prove a Scott-corrected
semi-classical result for a localised relativistic Hydrogen operator
(Lemma~\ref{lem:Coulomb} below). The ingredients of the proof of the
latter are the same (rescaled semi-classics, localisation and
multi-scale analysis, and estimating localisation errors).

\begin{thm}[{\bf Rescaled semi-classics}] \label{corollary}
Let $n\geq 3$ and let $\phi\in C^{n+4}_0(\R^n)$ be supported in a ball
$B_\ell$ of radius $\ell>0$. Let $V\in C^3(\overline{B_\ell})$ be a
real potential, and let $T_\b(p)=\sqrt{\b^{-1}p^2+\b^{-2}}-\b^{-1}$ be
the kinetic energy. Let $H_\b=T_\b(-\ic h\nabla)+V(\hat{x})$, $h>0$,
and $\sigma_\b(v,q) = T_\b(q) + V(v)$. Then for all $h, \beta, f>0$
with $\beta f^2\le1$, we have 
\begin{equation} \label{eq:phiHphilf}
  \left|\Tr[\phi H_\b\phi]_{-} - (2\pi h)^{-n}\int \phi(v)^2
    \s_\b(v,q)_{-} \,dv dq \right| 
  \le C h^{-n+6/5} f^{n+4/5}\ell^{n-6/5} \,,
\end{equation}
where the constant $C$ is independent of $\b$ and depends only on 
\begin{equation}\label{eq:phivdependence}
  \sup_{|\eta|\le n+4}\|\ell^{|\eta|}\p^\eta\phi\|_{\infty}
  \quad\hbox{  and }\quad 
  \sup_{|\eta|\le3}\|f^{-2}\ell^{|\eta|}\p^\eta V\|_{\infty} \,.
\end{equation}
Moreover, there exists a density matrix $\gamma$ such that 
\begin{equation}
  \Tr[\gamma\phi H_\b\phi]\leq (2\pi h)^{-n}\int \phi(v)^2
  \s_\b(v,q)_{-} \,dv dq  
  +C h^{-n+6/5} f^{n+4/5}\ell^{n-6/5}\label{eq:gammaproplf} \,.
\end{equation}
The density  $\rho_\gamma$ satisfies
\begin{equation}\label{eq:rhogammaproplf1}
  \left|\rho_\gamma(x) - (2\pi h)^{-n}\omega_n 
  |V_-|^{n/2} (2+\b|V_-|)^{n/2}(x)\right| 
  \leq Ch^{-n+9/10}f^{n-9/10}\ell^{-9/10}
\end{equation}
for (almost) all $x\in B_\ell$, and
\bea\label{eq:rhogammaproplf2}
  \lefteqn{\left|\int\phi(x)^2\rho_\gamma(x)\,dx
           -(2\pi h)^{-n}\omega_n\int\phi(x)^2 
           |V_-|^{n/2} (2+\b |V_-|)^{n/2}(x)\,dx\right| 
           }\nonumber
  \\ & \leq & Ch^{-n+6/5}f^{n-6/5}\ell^{n-6/5} \,,\hspace{7cm}
\eea
where $\omega_n$ is the volume of the unit ball $B_1$ in $\R^n$.
The constants $C>0$ in the above estimates again depend on
the parameters as in \eqref{eq:phivdependence}.
\end{thm}

\begin{proof} We introduce the unitary scaling operator
$(U\psi)(x)=\ell^{-n/2}\psi(\ell^{-1}x)$. Then
$$ 
  U^* \phi \big[T_\b(-\ic h\nabla) + V(\hat{x})\big] \phi \,U = 
  f^2 \phi_\ell \big[T_{\beta f^2}(-\ic h f^{-1}\ell^{-1} \nabla) +
  V_{f,\ell}(\x)\big] \phi_\ell\,,
$$
where $\phi_\ell(x) = \phi(\ell x)$, and $V_{f,\ell}(x) = f^{-2}V(\ell
x)$. Thus, 
$$ 
  \Tr[\phi H_\b\phi]_{-} = f^2\,\Tr\Big[\phi_\ell \big[T_{\beta f^2}(-\ic
  h f^{-1}\ell^{-1} \nabla) + V_{f,\ell}(\x)\big] \phi_\ell\Big]_{-}\,.
$$
Note that $\phi_\ell$ and $V_{f,\ell}$ are supported  in a ball of
radius $1$ and that for all multi-indices $\eta$,
$$ 
  \| \p^\eta \phi_\ell \|_\infty = \|\ell^{|\eta|}\p^\eta \phi\|_\infty
  \quad{\rm{ and }}\quad \| \p^\eta V_{f,\ell} \|_\infty =
  f^{-2}\|\ell^{|\eta|}\p^\eta V\|_\infty\,. 
$$
Let $\s_{f,\ell,\b}(u,q) = T_{\b f^2}(q) + V_{f,\ell}(u)$. By
Theorem~\ref{local semi-classics} in Section~\ref{Local semi-classical
  estimates} below there is a constant $C$ depending on
the parameters as in (\ref{eq:phivdependence}) so that, as long as \(\beta
f^2\le1\), 
\be 
  \left|\Tr[\phi H_\b\phi]_- - (2\pi h f^{-1} \ell^{-1})^{-n} f^2
  \int \phi_\ell(u)^2 \s_{f,\ell,\beta}(u,q)_-\, du dq \right|
  \le C f^2 (hf^{-1}\ell^{-1})^{-n+6/5} \,.
\ee
A simple change of variables gives 
$$
  (2\pi h f^{-1} \ell^{-1})^{-n} f^2 \int \phi_\ell(u)^2
  \s_{f,\ell,\beta}(u,q)_-\, du dq 
  = (2\pi h)^{-n} \int \phi(v)^2 \s_\b(v,q)_{-} \,dv dq \,,
$$
and we have proved \eqref{eq:phiHphilf}. 

Now, let $\gamma_{f,\ell,\beta}$ be the density matrix for 
$\phi_\ell \big[T_{\beta f^2}(-\ic h f^{-1}\ell^{-1} \nabla) +
V_{f,\ell}(\x)\big] \phi_\ell$,  
which according to Lemma \ref{lm:upperbound} satisfies
\beax 
  f^2\lefteqn{\,\Tr\Big[\phi_\ell \big[T_{\beta f^2}(-\ic h
  f^{-1}\ell^{-1} \nabla) + V_{f,\ell}(\x)\big]
  \phi_\ell\gamma_{f,\ell,\beta}\Big]}\\      
  &\le& (2\pi hf^{-1}\ell^{-1})^{-n} f^2\int \phi_\ell(u)^2
  \s_{f,\ell,\beta}(u,q)_-\, du dq + C (hf^{-1}\ell^{-1})^{-n+6/5} \,,  
\eeax
\beax
  \lefteqn{\left|\r_{\gamma_{f,\ell,\beta}}(x) - (2\pi
           hf^{-1}\ell^{-1})^{-n}\omega_n |V_{f,\ell}(x)_-|^{n/2} 
           (2+\b f^2|V_{f,\ell}(x)_-|)^{n/2}(x)\right|
          }
  \\& \le & C (hf^{-1}\ell^{-1})^{-n+9/10}\,,\hspace{9cm}
\eeax
\beax
  \lefteqn{\left|\int \phi_\ell(x)^2\r_{\gamma_{f,\ell,\beta}}(x)\, dx - 
  (2\pi hf^{-1}\ell^{-1})^{-n} \omega_n \int \phi_{\ell}(x)^2
  |V_{f,\ell}(x)_-|^{n/2} (2+\b f^2|V_{f,\ell}(x)_-|)^{n/2}\,dx\right| 
  }\\
  &\le&C (hf^{-1}\ell^{-1})^{-n+6/5}\,.\hspace{10.8cm}
\eeax
The density matrix $\gamma=U\gamma_{f,\ell,\beta}U^*$, whose density is
$\rho_\gamma(x) = \ell^{-n} 
\r_{\gamma_{f,\ell,\beta}}(x/\ell)$, then satisfies the properties in 
\eqref{eq:gammaproplf}--\eqref{eq:rhogammaproplf2}.
\end{proof}

\noindent{\bf Multi-scale Analysis.} The rescaled semi-classics of
Theorem~\ref{corollary} will be used in balls of varying radius. 
This idea goes back to Ivri{\u\i} \cite{Ivrii-Sigal,bookIvrii}. We
introduce a variable scale $\ell$ and a corresponding family of
localisation functions \(\{\varphi_u\}_{u\in\R^3}\), which will also
be used in the proof of Theorem~\ref{thm:TF-semicl}. 
\begin{defn}[{\bf Scale for multi-scale analysis}]
Let $0<\ell_0<1$ be a parameter
that we shall choose explicitly below, and let
\(r_1,\ldots,r_M\in\R^3\). 
Define 
\begin{equation}\label{eq:ldefinition}
  \ell(x)=\mfr{1}{2}\Bigl(1+\sum_{k=1}^M(|x-r_k|^2
  +\ell_0^2)^{-1/2}\Bigr)^{-1}\,. 
\end{equation}
\end{defn}
Note that $\ell$ is a smooth function (due to $\ell_0$) with 
\begin{equation}\label{eq:bound nabla ell}
   0<\ell(x)<1/2\,\quad\hbox{and}\quad\|\nabla\ell\|_\infty<1/2\,.
\end{equation}
Note also that in terms of the function $d\equiv d_{\bf r}$ from
\eqref{ddefinitionBIS} we have 
\begin{equation}\label{eq:elldcom}
  \mfr{1}{2}(1+M)^{-1}\ell_0\leq
  \mfr{1}{2}(1+M(d(x)^2+\ell_0^2)^{-1/2})^{-1}
  \leq\ell(x)\leq\mfr{1}{2}(d(x)^2+\ell_0^2)^{1/2}\,.
\end{equation}
In particular, we have 
\begin{equation}\label{eq:elldcom2}
  \ell(x)\geq\mfr{1}{2}(1+M)^{-1}\min\{d(x),1\} \,.
\end{equation}
We fix a localisation function $\varphi\in C_0^\infty(\R^3)$ with
support in $\{|x|<1\}$ and such that
$\int\varphi(x)^2\,dx=1$. According to Theorem~\ref{partition} we can
find a corresponding family of functions $\varphi_u\in
C_0^\infty(\R^3)$, $u\in\R^3$, where $\varphi_u$ is supported in the
ball $\{|x-u|<\ell(u)\}$, with the properties that
\begin{equation}\label{eq:phiuprop}
  \int_{\R^3}\varphi_u(x)^2\ell(u)^{-3}\,du\,=\,1\quad\hbox{and}\quad
  \|\partial^{\eta}\varphi_u\|_\infty\,\leq\, C\ell(u)^{-|\eta|}\,,
\end{equation} 
for all multi-indices $\eta$, where $C>0$ depends only on $\eta$ and
$\varphi$. For $d(u)>2\ell_0$ we have $\ell(u)\leq \sqrt{5}d(u)/4$ and
hence for all $x$ with $|x-u|<\ell(u)$ we have (note that $d(u)\le
d(x)+|x-u|$ and $\sqrt{5}/4<1$) that
\begin{equation}\label{eq:ell and d}
  \ell(u)< d(u)\quad\hbox{and}\quad d(u)\leq C d(x)\,.
\end{equation}

As a first step towards the Scott correction for Coulomb-type
potentials we deal with the Scott correction for a single relativistic
Hydrogen atom. This method for proving the existence of a Scott
correction in the semi-classical expansion of the sum of eigenvalues
of an operator with a (homogeneous) singular potential without
explicitly knowing the eigenvalues was first used by
Sobolev~\cite{Sobolev} when studying (non-relativistic) operators with
magnetic fields. 

\begin{lemma}[{\bf Scott-corrected localised Hydrogen}]
  \label{lem:Coulomb}
There exists a non-increasing function ${\mathcal
  S}:[0,2/\pi]\to\mathbb{R}$, with \(\mathcal{S}(0)=1/4\), such that,
if $\phi_r(x)=\phi(x/r)$, $r\in(0,\infty)$, with $\phi\in C^7(\R^3)$, 
$0\le\phi\le1$, satisfying \(\sqrt{1-\phi}\in C^{1}(\R)\) and
$$
 \phi(x)=\left\{\begin{array}{lcr}1&\mbox{for}&|x|
 \le1\\0&\mbox{for}&|x|\ge2\end{array} 
 \right. \,,
$$
then there exists $C>0$ depending only on $\phi$ such that for all
\(\alpha\in[0,2/\pi]\) and \(r\in(0,\infty)\),
\be\label{eq:local Hydro} 
   \left|\Tr[\phi_r H_{\rm C}(\alpha) \phi_r]_- - (2\pi)^{-3} \int\phi_r(v)^2
   \big[\mfr{1}{2} q^2 - |v|^{-1}\big]_- \, dv dq  -
   {\mathcal S}(\alpha) \right| < C r^{-1/10} \,,
\ee
where
\begin{align}\label{def:H_C}
  H_{\rm C}(\alpha)=
  \begin{cases}
     \sqrt{-\alpha^{-2}\D+\alpha^{-4}}-\alpha^{-2}
     -|{\x}|^{-1}\ , & \alpha\in(0,2/\pi] \\
     {}-\frac12\D -|{\x}|^{-1} \ ,& \alpha=0
   \end{cases}\,.
\end{align}
\end{lemma}
As emphasised in Remark~\ref{rem:non-rel-in-s-c}, the function in the
semi-classical integral in \eqref{eq:local Hydro} is the 
{\it non-relativistic} energy. See also Lemma~\ref{cor:alt-S}
below for an alternative description of the function \(\mathcal{S}\).
\begin{remark}
 A result similar to the one in Lemma~\ref{lem:Coulomb} was proved in
 \cite[Theorem 7.1]{sor-thesis}, but without uniform control in
 \(\alpha\) and only for
 \(\alpha\in(0,2/\pi)\). 
\end{remark}
\begin{proof}[Proof of Lemma~\ref{lem:Coulomb}] 
We fix $\alpha\in[0,2/\pi]$ and write $H_{\rm C}=H_{\rm C}(\alpha)$.
We define for $r>0$ 
\be  
   S_r=\Tr\big[\phi_r H_{\rm C}(\alpha)
  \phi_r\big]_- 
   - (2\pi)^{-3} \int\phi_r(v)^2 \big[\mfr{1}{2}q^2 -
   |v|^{-1}\big]_- \, dv dq \,. 
\ee
We will show that $S_r$ has a limit as $r\to\infty$.

Let $R>2r$. We estimate the difference between $\Tr[\phi_R H_{\rm C}
\phi_R]_-$ and $\Tr[\phi_r H_{\rm C} \phi_r]_-$ semi-classically. The
leading semi-classical term involves the relativistic energy which is
then compared to the non-relativistic energy. Below all constants will
depend only on \(\phi\) and in particular {\it not} on
$\alpha\in[0,2/\pi]$.

Denote $\psi_r=\sqrt{1-\phi_r^2}$. By the relativistic IMS formula
\eqref{eq:IMS}, 
$$ 
  H_{\rm C} = \phi_r H_{\rm C}\phi_r + \psi_r H_{\rm C} \psi_r -
  L_{\phi_r}-L_{\psi_r}\,, 
$$ 
where $ L_{\phi_r}$ and $L_{\psi_r}$ are given by \eqref{IMS-error1}
and \eqref{IMS-error2}  
(\(\mathcal M=\{1,2\}\)). We multiply with $\phi_R$ and get that
$$
  \phi_R H_{\rm C}\phi_R = \phi_r H_{\rm C} \phi_r + 
  \phi_R \psi_rH_{\rm C}\psi_r\phi_R -
  \phi_R(L_{\phi_r}+L_{\psi_r})\phi_R \,. 
$$
We have used that $\phi_R\phi_r=\phi_r$ since $R>2r$. Now, let
$\gamma_R=\chi(\phi_R H_{\rm C} \phi_R)$ be the projection onto the
negative part of $\phi_R H_{\rm C} \phi_R$.  Then, by the variational
principle and Theorem~\ref{IMS-error-est} (with $m=1$, $\ell=r$,
$\Omega=B(0,3r)$, and $\theta=\phi_r$ and \(\psi_r\), respectively), 
\bea\label{eq:firstLoc}\nonumber 
  \lefteqn{\Tr[\phi_R H_{\rm C}\phi_R]_-}
  \\\nonumber &=& \Tr[\gamma_R\phi_r H_{\rm C} \phi_r] +
  \Tr[\gamma_R \phi_R \psi_r H_{\rm C} \psi_r \phi_R] \,-\Tr[\gamma_R
  \phi_R(L_{\phi_r}+L_{\psi_r})\phi_R]
  \\
  &\ge&\Tr[\gamma_R \phi_r (H_{\rm C}-Cr^{-2}) \phi_r]
  +\Tr[\gamma_R \phi_R \psi_r(H_{\rm C}-Cr^{-2}\phi_{3r})\psi_r\phi_R]
  \\\nonumber &&\,{}- Cr^{-2}\,.  
\eea 
Here, \(C\) is independent of \(\alpha\).  We treat the part of the
localisation error coming from the first term in
\eqref{eq:firstLoc}. We split $H_{\rm C} = (1-\varepsilon) H_{\rm C} +
\varepsilon H_{\rm C}$ for some \(0<\varepsilon<1\) to be chosen and
use the second term to control the error term. 

By Theorem~\ref{thm:L-Y-D} (with \(M=1\), \(R_0=0\), $d(x)=|x|$,
\(m=1\) and $\nu=1$),
\beax \label{Lieb-Yau1}
  \lefteqn{\Tr[\gamma_R\phi_r(\varepsilon H_{\rm
      C}-Cr^{-2})\phi_r]}\nonumber\\ 
  &=&
  \varepsilon\,\Tr\big[(\phi_r\gamma_R\phi_r)\big\{\sqrt{-\alpha^{-2}\D 
  +\alpha^{-4}}-\alpha^{-2}-  
  (|\hat{x}|^{-1}+Cr^{-2}\varepsilon^{-1})\big\}\big]
  \\&
  \geq& {}- C\varepsilon\Big(\alpha^{1/2}+\!\!\!\!
  \int\limits_{\alpha<|x|<2r}\!\!\!\!(|x|^{-1}
  +C\varepsilon^{-1}r^{-2})^{5/2}\,dx+ 
  \alpha^3\!\!\!\!\!\!\int\limits_{\alpha<|x|<2r}\!\!\!\!(|x|^{-1}
  +C\varepsilon^{-1}r^{-2})^4\,dx\Big) 
  \\
  &\geq&{}- C\varepsilon\Big(1+r^{1/2}
  +\varepsilon^{-5/2}r^{-2}+\varepsilon^{-4}r^{-5})\,,
\eeax
assuming $\varepsilon^{-1}r^{-2}\leq C\alpha^{-1}$ and using that
$\alpha\leq2/\pi$. We may choose $\varepsilon=r^{-1}$ if we assume
that $r>1$ (note that then indeed $\varepsilon^{-1}r^{-2}=r^{-1}<1\leq
2\alpha^{-1}/\pi$). We then obtain 
$$
  \Tr[\gamma_R\phi_r(\varepsilon H_{\rm C}-Cr^{-2})\phi_r]\geq {}-Cr^{-1/2}.
$$

As a result, we have shown that 
\beax 
  \lefteqn{\Tr[\phi_R H_{\rm C}\phi_R]_-}
  \\
  &\geq& (1-\varepsilon)\,\Tr[\gamma_R\, \phi_r H_{\rm C} \phi_r] + 
  \Tr[\gamma_R \phi_R \psi_r(H_{\rm C}-Cr^{-2}\phi_{3r})\psi_r\phi_R]
  - C r^{-1/2}
  \\
  &\geq& \Tr[\phi_r H_{\rm C} \phi_r]_- + 
  \Tr[\phi_R \psi_r(H_{\rm C}-Cr^{-2}\phi_{3r})\psi_r\phi_R]_-
  - C r^{-1/2}\,.
\eeax

We will treat the term $\Tr[\phi_R\psi_r(H_{\rm
  C}-Cr^{-2}\phi_{3r})\psi_r\phi_R]_-$ by our semi-classical estimates
in Section~\ref{Local semi-classical estimates} below. We first
rescale. Define the unitary scaling operator
$(U\varphi)(x)=R^{-3/2}\varphi(R^{-1}x)$. Then 
\begin{align}\label{eq:unitary scaling}\nonumber
  \widetilde{H}_{\rm C}\,
  :&= \,U^* (H_{\rm C}-Cr^{-2}\phi_{3r})U \,
  \\&= \,R^{-1}
  \big(\sqrt{-\alpha^{-2}\D + R^2\alpha^{-4}} - R\alpha^{-2}
   - |\x|^{-1}-C R r^{-2}\phi_{3r/R}(\x)
  \big) \nonumber
  \\&=\, R^{-1}\big( T_{\b}(-\ic h\nabla)-|\x|^{-1} -C R
  r^{-2}\phi_{3r/R}(\x)\big)
\end{align}
with $\beta=\alpha^{2}R^{-1}\,(<R^{-1})$ and $h=R^{-1/2}$. Let
$\phi_{R,r}=\phi_R\psi_r=\phi_R\sqrt{1-\phi^2_r}$ and 
$\psi(x)=\phi_{R,r}(Rx)$ (see \eqref{def:kinetic energy} for
\(T_\beta\)). In this way, $\phi_{R}\psi_r(H_{\rm
  C}-Cr^{-2}\phi_{3r})\psi_r\phi_{R}$ and $\psi \widetilde{H}_{\rm
  C}\psi$ are unitarily equivalent.

Now, let $\ell$ and $\varphi_u$ be the functions in
\eqref{eq:ldefinition} and \eqref{eq:phiuprop}, respectively, when
$M=1$, $r_1=0$, and $\ell_0=h^2=R^{-1}$. 
By another relativistic IMS localisation we get that
\beax 
   \psi \widetilde{H}_{\rm C} \psi &=&R^{-1}
   \!\!\!\! \!\!\!\! 
   \int\limits_{r/3R<|u|<5/2} \!\!\!\! \!\!\!\!
   \psi\varphi_u 
   \big(T_{\b}(-\ic
   h\nabla)-|\x|^{-1}-C R r^{-2}\phi_{3r/R}(\x)\big)\varphi_u \psi
   \,\ell(u)^{-3}\, du 
   \\
  &&{}-R^{-1}\!\!\!\! \!\!\!\!
   \int\limits_{r/3R<|u|<5/2} \!\!\!\! \!\!\!\! \psi
   L_{\varphi_u}\psi \,\ell(u)^{-3} \, du\,. 
\eeax
We have used that $\psi\varphi_u=0$ for $|u|\not\in[r/3R,5/2]$,
 since $\ell(u)\le\frac12(|u|^2+\ell_0^2)^{1/2}$ (see \eqref{eq:elldcom})
and $\supp\,\psi\subset\{r/R\le|x|\le2\}$,
$\supp\,\varphi_u\subset\{|x-u|\le\ell(u)\}$.

Concerning $L_{\varphi_u}$, Theorem~\ref{IMS-error-est} with
$\ell=\ell(u)/2$, $m=R$, and $\Omega=\Omega_u=\{|x-u|\le 3\ell(u)/2\}$
gives 
$$ 
  L_{\varphi_u} \le C R^{-1}\ell(u)^{-2}\chi_{\Omega_u} + Q_{\varphi_u}\,,
$$
with
\bea\label{eq:Tr Q}
  \Tr[Q_{\varphi_u}] \le C R\alpha^{-1}\ell(u)^{-1}{\rm
    e}^{-\alpha^{-1}R\ell(u)/2}\,. 
\eea
Here we have used \eqref{eq:phiuprop}. 

Notice that if the supports of \(\varphi_u\) and
\(\varphi_{u'}\) overlap then \(|u-u'|\le\ell(u)+\ell(u')\) and thus
\begin{align}\label{eq:overlap two phi-u's}
  \ell(u')\le\ell(u)+\|\nabla\ell\|_{\infty}(\ell(u)+\ell(u'))\,.
\end{align}
Therefore, since \(\|\nabla\ell\|_\infty<1/2\), we have that
\(\ell(u')\le C\ell(u)\) and thus \(\ell(u)^{-1}\le
C\ell(u')^{-1}\). Similarly, \(\ell(u)\le C\ell(u')\), and so
\(\chi_{\Omega_u}\leq\chi_{\{|x-u|\le C\ell(u')\}}\) if the supports
of \(\varphi_u\) and \(\varphi_{u'}\) overlap.

Using this and \eqref{eq:phiuprop} 
we get for all $x\in\R^3$
\bea \label{eq:trick to join errors}\nonumber
  \int
  \ \big(\ell(u)^{-2}\chi_{\Omega_u}(x)\big)\,\ell(u)^{-3}\,du
  &=&\int
  \big(\ell(u)^{-2}\chi_{\Omega_u}(x)\big)\,\big(
  \int
   \varphi_{u'}^2(x)\,\ell(u')^{-3}\,du'\big)\,\ell(u)^{-3}\,du  
  \\&\le& C\int
  \ \varphi_{u'}(x)\ell(u')^{-2}\varphi_{u'}(x)\,\ell(u')^{-3}\,du'.
\eea
Rewriting the last integral with $u$ as integration variable we get
\beax
  \psi \widetilde{H}_{\rm C} \psi&\ge&
  R^{-1}\int
  \psi\varphi_u\big(T_\b(-\ic
  h\nabla)-|\x|^{-1}-Ch^{2}\ell(u)^{-2}\big)\varphi_u\psi\,\ell(u)^{-3}\,du 
  \\&&{}-R^{-1}\int
  \psi Q_{\varphi_u}\psi\,\ell(u)^{-3}\,du.
\eeax
Here we have also used that $R r^{-2}\phi_{3r/R}(x)\leq
Ch^{2}\ell(u)^{-2}$ for $x$ in the support of $\varphi_u$. This is a
consequence of $\ell(u)\leq \mfr12 |u|+\mfr12 \ell_0\leq \mfr12
|x|+\mfr12\ell(u)+\ell_0$ for $x$ in the support of $\varphi_u$ which
implies that $\ell(u)\leq |x|+\ell_0\leq C r/R$ for $x$ in the support
of $\varphi_u$ and $\phi_{3r/R}$. 

We will now use Theorem~\ref{corollary} (with
\(\phi=\psi\varphi_u\), \(\ell=\ell(u)\),
\(B_{\ell}=\{|x-u|\le\ell(u)\}, f=f(u)=|u|^{-1/2}\)) on  
$$
  \psi\varphi_u\big(T_\b(-\ic
  h\nabla)-|\x|^{-1}-Ch^{2}\ell(u)^{-2}\big)\varphi_u\psi\,,
$$
for \(u\) fixed with \(|u|\in[r/3R,5/2]\). We claim that 
\be \label{est:derivatives}
  \|\p_x^{\eta} (\psi\varphi_u)\|_\infty \le C_{\eta} \ell(u)^{-|\eta|}\,
  \text{ for all } \eta\in\N^3\,.
\ee
This follows from \eqref{eq:phiuprop}, \eqref{eq:ell and d}, and the
estimate $|\p^{\eta}\psi(x)| \le C_{\eta}|x|^{-|\eta|}$. It suffices
to check the latter for $1\le |x|\le 2$ and $r/R \le |x| \le 2r/R$,
due to the support properties of \(\psi\). 
Furthermore, for \(r>3\), \(|x|^{-1}+Ch^{2}\ell(u)^{-2}\)
is smooth (as  a function of \(x\)) on \(B_{\ell}\) (use
\eqref{eq:elldcom}, \(\ell_0=R^{-1}\),  and \(|u|\ge r/3R\)), and
satisfies
\be \label{est:derivativesV}
  \sup_{|x-u|<\ell(u)}\big|\p_x^{\eta}(|x|^{-1}+Ch^{2}\ell(u)^{-2})
  \big| 
  \le C_{\eta} f(u)^2\ell(u)^{-|\eta|}\, 
  \text{ for all } \eta\in\N^3\,,
\ee
with $f(u)=|u|^{-1/2}$. For the Coulomb potential, 
this is trivial. For the other term, only the statement
for \(\eta=0\) is non-trivial; it follows from \eqref{eq:elldcom2},
\(h=R^{-1/2}\), and \(|u|\ge r/3R\). Finally, the condition
$f(u)^2\beta\le 1$ is also satisfied (when \(r\ge3\)), since \(|u|\ge
r/3R\) and \(\beta<R^{-1}\).  

{F}rom Theorem~\ref{corollary} we conclude that
\begin{align*}
  \Tr[\phi_R &\psi_r(H_{\rm C}-Cr^{-2}\phi_{3r})\psi_r\phi_R]_- = 
  \Tr[\psi \widetilde{H}_{\rm C}\psi]_-
  \\
  &\ge \,R^{-1}(2\pi h)^{-3} \!\!\!\! \!\!\!\!\!\!\!
  \int\limits_{r/3R<|u|<5/2} \!\!\!\! \!\!\!\!
 \psi(v)^2\varphi_u(v)^2 \,
  \big[T_\beta(q) -
  |v|^{-1}-Ch^2\ell(u)^{-2}]_-
  \,\ell(u)^{-3}\, du dv dq 
  \\
  &\quad {}- C R^{-1}h^{-2+1/5} \!\!\!\! \!\!\!\!
  \int\limits_{r/3R<|u|<5/2} \!\!\!\! \!\!\!\!
  f(u)^{4-1/5}\ell(u)^{-1-1/5}\, du
  \\
  &\ \  {}-\,R^{-1} \!\!\!\! \!\!\!\!
  \int\limits_{r/3R<|u|<5/2} \!\!\!\! \!\!\!\!
  \Tr\big[\psi Q_{\varphi_u}\psi\big]  
  \,\ell(u)^{-3}\, du\,.
\end{align*}
Integrating the semi-classical error using \(f(u)=|u|^{-1/2}\),
\eqref{eq:elldcom2}, and \(R>r\) 
gives the lower bound $-CR^{-1} h^{-2+1/5} (R/r)^{1/10}={}-C 
r^{-1/10}$. 

{F}rom \eqref{eq:Tr Q} it follows, using \eqref{eq:elldcom2},
\(\alpha\leq2/\pi\), and \(R>r\), that
\begin{eqnarray*}
 R^{-1} \!\!\!\! \!\!\!\!\!\!
  \int\limits_{r/3R<|u|<5/2} \!\!\!\! \!\!\!\!\!\!
  \Tr\big[\psi Q_{\varphi_u}\psi\big]  
  \,\ell(u)^{-3}\, du
  &\le& C\!\!\!\! \!\!\!\!\!\!
  \int\limits_{r/3R<|u|<5/2} \!\!\!\! \!\!\!\!\!\!
  \alpha^{-2}{\rm e}^{-\alpha^{-1}R\ell(u)/2}\ell(u)^{-3}\,du\\&
  \le& C r^{-1}{\rm e}^{-r/8}\,.
\end{eqnarray*}
Since \(\supp\,\varphi_u\subset\{v\,|\,|u-v|\le \ell(u)\}\) and 
$|u|\leq5/2$ we have \(|v|\le |u|+\ell(u)\le C\ell(u)\) on
\(\supp\,\varphi_u\). Using this, integrating in \(u\) (using
\eqref{eq:phiuprop}), we get 
\beax
 \lefteqn{  R^{-1}(2\pi
   h)^{-3}\!\!\!\! \!\!\!\!\!\!\!
  \int\limits_{r/3R<|u|<5/2}
   \!\!\!\! \!\!\!\!  
  \psi(v)^2\varphi_u(v)^2 \,
  \big[T_\beta(q) -
  |v|^{-1}-Ch^2\ell(u)^{-2}]_-
  \,\ell(u)^{-3}\, du dv dq} \\
 &\ge&
  R^{-1}\frac{1}{(2\pi h)^3}
  \int 
 \psi(v)^2\,
  \big[\sqrt{\beta^{-1}q^2+\beta^{-2}}-\beta^{-1}-
  |v|^{-1}-Ch^2|v|^{-2}\,]_-
  \, dv dq\,.
\eeax

In order to compare this latter relativistic semi-classical expression
with the non-relativistic semi-classical one we use the inequality
$|x_- - y_-|\le |x-y|$ and a Taylor expansion of $\sqrt{t^2+1}-1$
to arrive at
\begin{eqnarray}\label{eq:semiclassicalcomp}
  \lefteqn{\int\Big|\big[\mfr{1}{2} q^2 -a\big]_-
  -[\sqrt{\beta^{-1}q^2+\beta^{-2}}-\beta^{-1}-a-b]_-\Big|\,dq}
  &&\nonumber\\
  &\hspace{1cm}\leq&
  C\beta (\beta(a+b)^2+2(a+b))^{7/2}+Cb(\beta(a+b)^2+2(a+b))^{3/2}
\end{eqnarray}
for all $a,b>0$. This gives, using $h^2=R^{-1}$ and $\beta<R^{-1}$, that
\begin{eqnarray}\label{eq:change s-c}\nonumber
  \lefteqn{
  \left|\int \psi(v)^2 \,\Big(\big[\mfr{1}{2} q^2 -  
  |v|^{-1}\big]_- - \big[\sqrt{\beta^{-1}q^2+\beta^{-2}}-\beta^{-1} -
  |v|^{-1}
  -Ch^2|v|^{-2}
  \big]_-\Big)\, dvdq \right|
  }\hspace{5truecm}&&
  \\&
  \le& C R^{-1}\int\limits_{r/R<|v|< 2}
  |v|^{-7/2}\, dv\ \le\ C(Rr)^{-1/2}\,, 
\end{eqnarray}
since \(R>r\geq1\). 

Thus undoing the scaling we arrive at
\begin{eqnarray*}
  R^{-1}(2\pi h)^{-3}\!\!\!\! \!\!\!\!\!\!\!
  \int\limits_{r/3R<|u|<5/2}
  \!\!\!\! \!\!\!\!  
  \lefteqn{\psi(v)^2\varphi_u(v)^2 \,
  \big[T_\beta(q) -
  |v|^{-1}-Ch^2\ell(u)^{-2}]_-
  \,\ell(u)^{-3}\, du dv dq} \\
  &\ge&(2\pi)^{-3} \int \phi_{R,r}(v)^2\,
  \big[\tfrac12q^2 - |v|^{-1}]_- \, dv dq 
  - C r^{-1/10}\,.
\end{eqnarray*}
Summarizing, we have proved that there exists a constant
\(C=C(\phi)\), independent of \(\alpha\in[0,2/\pi]\),
 such that for \(r\) large enough, and \(R>2r\), 
\bea\label{eq:lower bound TF-Scott}\nonumber
  \lefteqn{\Tr[\phi_R H_{\rm C}\phi_R]_-}
  \\&\ge&
  \Tr[\phi_r H_{\rm C} \phi_r]_- + 
  (2\pi)^{-3} \int \phi_{R,r}(v)^2\,
  \big[\tfrac12q^2 - |v|^{-1}]_- \, dv dq 
  - C r^{-1/10}\,.
\eea

Next, we want to bound $\Tr[\phi_R H_{\rm C}\phi_R]_-$ from above by
$\Tr[\phi_r H_{\rm C}\phi_r]_-$  by constructing a density matrix. To
this end, we first set $\gamma_r=\chi(\phi_r H_{\rm C} \phi_r)$. Then
we let $\widetilde{\gamma}_u$ be the density matrix we get when we use
Theorem~\ref{corollary} for the rescaled operator
$\psi\varphi_u\widetilde{H}_{\rm C}\varphi_u\psi$ (now with
\(\widetilde{H}_{\rm C}=U^*H_{\rm C}U\) with \(U\) as in
\eqref{eq:unitary scaling}), for fixed \(u\) with
\(|u|\in[r/3R,5/2]\), and set $\gamma_u=U\varphi_u\widetilde{\gamma}_u
\varphi_u U^*$. Finally, we define 
\be 
  \gamma = \phi_r\gamma_r\phi_r +
  \!\!\!\! \!\!\!\!
  \int\limits_{r/3R<|u|<5/2} \!\!\!\! \!\!\!\!
  \psi_r \,\gamma_u \,\psi_r \,\ell(u)^{-3} \,du\,.
\ee
Since \({\bf 0}\le\widetilde\gamma\le{\bf 1}\) and
\(\int\varphi_u^2(x)\ell(u)^{-3}\,du=1\), 
$$ 
  {\bf 0} \le \int \gamma_u\,\ell(u)^{-3}\,du \le {\bf 1}\,,
$$
and so we see, by multiplication with $\psi_r$ on both sides, that
${\bf 0} \le \gamma \le{\bf 1}$. Also, \(\gamma\) is trivially trace 
class. 
By the variational principle we obtain that
\beax 
  \Tr[\phi_R H_{\rm C}\phi_R]_- &\le& \Tr[\phi_R H_{\rm C}
  \phi_R\gamma] 
  \\
  &=&\Tr\big[\phi_R\phi_r H_{\rm C}\phi_r\phi_R\,
  \chi(\phi_r H_{\rm C} \phi_r)\big] 
  \\
  &&+\,
  \!\!\!\! \!\!\!\!
  \int\limits_{r/3R<|u|<5/2} \!\!\!\! \!\!\!\!
    \Tr[\psi_r
  \phi_R H_{\rm C} \phi_R \psi_r\gamma_u]\, \ell(u)^{-3} \,du 
  \\
  &\le& \Tr[\phi_r H_{\rm C}\phi_r]_- + 
  \!\!\!\! \!\!\!\!
  \int\limits_{r/3R<|u|<5/2} \!\!\!\! \!\!\!\!
  \Tr[\psi\varphi_u \widetilde{H}_{\rm C} \varphi_u\psi
  \widetilde{\gamma}_u]\, \ell(u)^{-3}  
  \,du\,.
  \\  
\eeax
Here we have used that $\phi_R\phi_r H_{\rm C}\phi_r\phi_R=\phi_r
H_{\rm C}\phi_r$, since \(R>2r\), and again scaled the operators
inside the trace in the integrand. Using Theorem~\ref{corollary} we
can bound the integral from above by 
\beax
   \lefteqn{ 
   R^{-1} (2\pi h)^{-3} \int\psi(v)^2 \varphi_u(v)^2 
   \big[T_\b(q)-|v|^{-1}\big]_- \ell(u)^{-3}\, du dv dq
   }
   \\
   && \qquad\quad+ \ C  R^{-1} h^{-2+1/5} 
   \!\!\!\! \!\!\!\!
   \int\limits_{r/3R<|u|<5/2} \!\!\!\! \!\!\!\!
   f(u)^{4-1/5}\ell(u)^{-1-1/5}\, du\,.
\eeax
As in the case of the lower bound, the error term is bounded by
$Cr^{-1/10}$. 

Integrating with respect to $u$ in the semi-classical expression
above, changing back coordinates, and using
\eqref{eq:change s-c}, we conclude that
\bea\label{eq:upper bound TF-Scott}
  \Tr[\phi_R H_{\rm C}\phi_R]_{-}
  \le
  \Tr[\phi_r H_{\rm C}\phi_r]_{-}
  +(2\pi)^{-3}\int \phi_{R,r}(v)^2 \big[\tfrac12q^2 -
  |v|^{-1}\big]_- \,dv dq
  + Cr^{-1/10}\,.
\eea

Combining \eqref{eq:lower bound TF-Scott} and \eqref{eq:upper bound
  TF-Scott} we have shown that for $R>2r$, 
\beax 
  &&
  \!\!\!\!\!\!\!\!\!\!\!\!\!\!\!\!
  \big|S_R-S_r\big|\\
  &\le&\Big|\Tr[\phi_{R}H_{\rm C}\phi_{R}]_{-}- \Tr[\phi_{r}H_{\rm
    C}\phi_{r}]_{-} 
  +\, (2\pi)^{-3}\int
  \big(\phi_{r}(v)^2-\phi_{R}(v)^2\big)\,\big[\mfr{1}{2} q^2 - 
  |v|^{-1}\big]_- \,dvdq \Big|
  \\
  &\le&C r^{-1/10}\,.  
\eeax 
Hence, $\{S_n\}_{n\in\mathbb N}$ is a Cauchy-sequence and with
$\mathcal{S}=\mathcal{S}(\alpha)$ the limiting value we have
$$
  \big|S_r-\mathcal{S}\big|\leq Cr^{-1/10}\,.
$$
This proves \eqref{eq:local Hydro}. That \(\mathcal{S}\) is
non-increasing follows from the fact that \(T_{\alpha^2}(p)\) (see
\eqref{def:kinetic energy}) is decreasing in \(\alpha\). Finally, that
\(\mathcal{S}(0)=1/4\) is a well-known fact \cite{SS}.
\end{proof}

\begin{proof}[Proof of Theorem \ref{thm:TF-semicl}] 
Using the combined Daubechies-Lieb-Yau inequality 
(see Theorem~\ref{thm:L-Y-D}) with $\alpha=\beta^{1/2}h^{-1}(\leq 1)$
and $m=h^{-2}$ we may assume that $h$ is bounded by some
constant, which we may choose small depending on $M$ and $r_0$, using
that \(z_k\le 2/\pi\), \(k=1,\ldots,M\), and that \(\mathcal{S}\) is a
bounded function (since it is non-increasing; see Lemma~\ref{lem:Coulomb}).

In order to control the region close to and far away from all the
nuclei we introduce localisation functions $\theta_\pm\in C^{1}(\R)$
with the properties that \(0\le\theta_{\pm}\le1\) and 
\begin{enumerate}
  \item $\theta_-^2+\theta_+^2=1$,
  \item $\theta_-(t)=1$ if $t<1$ and $\theta_-(t)=0$ for $t>2$. 
\end{enumerate}
Let $0<r<r_0/4$ and $0<r_0<R$
and define $\Phi_\pm(x)=\theta_\pm(d(x)/R)$ and
$\phi_\pm(x)=\theta_\pm(d(x)/r)$ (with \(d=d_{\bf r}\) as in
\eqref{ddefinitionBIS}).
We choose (assuming $h$ is small enough)
\begin{equation}\label{eq:defrR}
  r=\delta^{-1} h^2\quad\hbox{and}\quad
  R=\left\{\begin{array}{ll}Ch^{-1},&\hbox{if } \mu=0\\R_\mu,&\hbox{if
      } \mu\ne0 
\end{array}\right. ,
\end{equation}
where $\delta=h^{}<1/2$ and $R_\mu=C\mu^{-1}$ is chosen such that
${}-V(x)\geq0$ for $d(x)\geq R_\mu$ (see \eqref{eq:VcondD}). We will
keep writing $\delta$ and $R$ in the calculations below to show why
these choices are optimal. Clearly,
$\Phi_-^2+\Phi_+^2=1$, $\phi_-^2+\phi_+^2=1$, and
$\phi_-^2+\Phi_-^2\phi_+^2+\Phi_+^2=1$. 
Note also that 
$$
  \phi_-(x)=\sum_{k=1}^M\theta_{r,k}(x)\quad\hbox{ with
  }\quad\theta_{r,k}(x)=\theta_-(|x-r_k|/r)\,. 
$$
{\bf Step 1: Lower bound on the quantum energy.}

By the relativistic IMS formula \eqref{eq:IMS} and
Theorem~\ref{IMS-error-est} with \(m=h^{-2}\),
$\alpha=\beta^{1/2}h^{-1}(\leq 1)$, and either $\ell=R$,
$\Omega=\{d(x)\le3R\}$, and \(\theta=\Phi_{\pm}\) respectively, or
$\ell=r$, $\Omega=\{d(x)\le3r\}$, and $\theta=\theta_{r,k}$,
$k=1,\ldots,M$, or $\theta=\phi_+$ respectively, we find that
\bea\label{eq:first TF-loc}  \nonumber
  \lefteqn{T_\b(-\ic h\nabla) -V(\x)}&&
  \\&=&\nonumber\Phi_+\big(T_\b(-\ic h\nabla)
  -V(\x)\big)\Phi_+ + \Phi_-\big(T_\b(-\ic h\nabla)
  -V(\x)\big)\Phi_-
  - L_{\Phi_-} - L_{\Phi_+}
  \\&=& \sum_{k=1}^M\theta_{r,k}\big(T_\b(-\ic h\nabla)
  -V(\x)\big)\theta_{r,k} 
  +\Phi_-\phi_+\big(T_\b(-\ic h\nabla) -V(\x)\big)\phi_+\Phi_-
  \nonumber 
  \\&&{}+
  \Phi_+\big(T_\b(-\ic h\nabla)
  -V(\x)\big)\Phi_+-\Phi_-(\,\sum_{k=1}^{M}L_{\theta_{r,k}}
  + L_{\phi_+})\Phi_- - L_{\Phi_-} -  
  L_{\Phi_+}\,,
\eea
with 
\bea \label{eq:good trace}
  L_{\Phi_\pm} &\le& {C} h^2
  \|\nabla\Phi_{\pm}\|_{\infty}^2\,\chi_{\{d(x)\le 3R\}} +
  Q_{\Phi_\pm}\,,\\
  \Tr[Q_{\Phi_\pm}]&\le& C\beta^{-1}R^{-1}
  {\rm e}^{-(\beta^{1/2}h)^{-1}R}\,\|\nabla\Phi_{\pm}\|_{\infty}^2 
  \,\big|\{d(x)\le 3R\}\big|\,, 
\eea
and (with, by abuse of notation,
\(L_{\phi_{-}}=\sum_{k=1}^{M}L_{\theta_{r,k}}\)) 
\bea \label{eq:good trace2}
  L_{\phi_\pm} &\le& {C} h^2
  \|\nabla\phi_{\pm}\|_{\infty}^2\,\chi_{\{d(x)\le 3r\}} +
  Q_{\phi_\pm}\,,\\
  \Tr[Q_{\phi_\pm}]&\le& C\beta^{-1}r^{-1}{\rm
    e}^{-(\beta^{1/2}h)^{-1}r}\,\|\nabla\phi_{\pm}\|_{\infty}^2
  \,\big|\{d(x)\le 3r\}\big|\,. 
\eea
Using \(\big|\{d(x)\le 3R\}\big|\le 36\pi M R^3\),
\(\|\nabla\Phi_{\pm}\|_{\infty}\le CR^{-1}\) (and the corresponding
estimates for $r$ and $\phi_\pm$), \(\beta\le h^2\), and \(h\) small,
it follows that 
\beax
   \Tr[Q_{\Phi_\pm}]\le C h^{2}R^{-2}{\rm e}^{-h^{-2} R/2}\leq C_Nh^N,\quad
   \Tr[Q_{\phi_\pm}]\le C h^{2}r^{-2}{\rm e}^{-h^{-2} r/2}\leq C_Nh^N\,,
\eeax 
for any $N>0$ by the choices \eqref{eq:defrR}.

Hence we have that
\bea\label{eq:TF-loc-2}  \nonumber
  \lefteqn{\Tr\big[T_\b(-\ic h\nabla) -V(\x)\big]_-}&&
  \\&\geq& \sum_{k=1}^M\Tr\big[\theta_{r,k}\big(T_\b(-\ic h\nabla)
  -V(\x)-Ch^2r^{-2}\big)\theta_{r,k}\big]_-\nonumber\\&& 
  {}+\Tr\big[\Phi_-\phi_+\big(T_\b(-\ic h\nabla)
  -V(\x)-Ch^2r^{-2}\chi_{\{d(x)\le 3r\}} 
  -Ch^2R^{-2} \big)\phi_+\Phi_-\big]_-  \nonumber
  \\&&{}+
  \Tr\big[\Phi_+\big(T_\b(-\ic h\nabla)
  -V(\x)-Ch^2R^{-2}\chi_{\{d(x)\le 3R\}}
  \big)\Phi_+\big]_--Ch^{-2+1/10}\,.
\eea
Each of the first three terms above will be compared to the
corresponding semi-classical expression.
We shall treat the three terms by different methods. The
first term will be calculated using the Scott correction for Hydrogen
in Lemma~\ref{lem:Coulomb}.
The second term will be computed using the local rescaled
semi-classics in 
Theorem~\ref{corollary}. The last term is an
error term which we will treat first.

\noindent{\bf Control of the third term in \eqref{eq:TF-loc-2}.}

We use the Daubechies inequality \eqref{eq:Daub usefull} with
$m=h^{-2}$ and $\alpha=\beta^{1/2}h^{-1}(\leq 1)$. In the case $\mu=0$
we obtain, using the choice \eqref{eq:defrR} of $R$,
\begin{eqnarray}
   \lefteqn{\Tr\big[\Phi_+(T_\b(-\ic h\nabla)
   -V(\x)-Ch^2R^{-2}\chi_{\{d(x)\le 3R\}} )\Phi_{+}\big]_{-}}\nonumber\\
   &\ge& {} 
   -Ch^{-3}M\int\limits_{|x|>R}|x|^{-15/2}\,dx
   \nonumber
   {}-CM\int\limits_{|x|>R}|x|^{-12}\,dx-Ch^2R^{-2}-Ch^8R^{-5}
   \nonumber\\&\ge&
   {}-C\big(h^{-3}R^{-9/2}+R^{-9}+h^2R^{-2}-h^8R^{-5}\big)\geq
   {}-Ch^{3/2}\,.\label{eq:vtfouter} 
 \end{eqnarray}
The case $\mu\ne0$ gives a smaller error since \({}-V\ge0\) on the
support of \(\Phi_{+}\) in this case.

\noindent{\bf Control of the first term in \eqref{eq:TF-loc-2}.}

Using
\eqref{eq:VcondS} and \eqref{eq:defrR} we have 
\begin{eqnarray*}
  \lefteqn{\sum_{k=1}^M\Tr\big[\theta_{r,k}\big(T_\b(-\ic h\nabla)
  -V(\x)-Ch^2r^{-2}\big)\theta_{r,k}\big]_-}\\ 
  &\geq& \sum_{k=1}^M\Tr\left[\theta_r\big(T_\b(-\ic
    h\nabla) -z_k|\x|^{-1}-C\delta^2h^{-2}\big)
  \theta_r\right]_-\,,
\end{eqnarray*}
where we have written $\theta_r(x)=\theta_-(|x|/r)$. 
We have used here that 
\begin{equation}\label{eq:Crminest}
  Cr_{\min}^{-1}+C\leq Cr_0^{-1}+C\leq C\delta^2h^{-2}\,.
\end{equation}
It is this relation which sets a lower bound on $\delta$.
We will control the error using the combined Daubechies-Lieb-Yau
inequality in Theorem~\ref{thm:L-Y-D} with $m=h^{-2}$ and
$\alpha=\beta^{1/2}h^{-1}(\le1)$. 
Note that $m\alpha^{-1}=\beta^{-1/2}h^{-1}\geq h^{-2}$. 
Thus using Theorem~\ref{thm:L-Y-D} we find, for all density matrices
$\gamma$ and all $\e\geq \delta^2$, that
$$
  \e\,\Tr\left[\gamma\big(\theta_r(T_\b(-\ic h\nabla) -
  z_k|\x|^{-1}-C\e^{-1}\delta^2h^{-2})\theta_r\big)\right]
  \geq{}-C(\e\delta^{-1/2}+\e^{-3/2}\delta^2   
  +\e^{-3}\delta^5)h^{-2}.
$$
Thus for all density matrices $\gamma$ and all $\e\geq \delta^2$ we have 
\begin{equation}
  C\delta^2h^{-2}\,\Tr[\gamma\theta_r^2]\leq
  \e\,\Tr\left[\gamma\theta_r(T_\b(-\ic 
  h\nabla) -
  z_k|\x|^{-1})\theta_r\right]+C(\e\delta^{-1/2}+\e^{-3/2}\delta^2 
  +\e^{-3}\delta^5)h^{-2}\,.
\end{equation}
Hence 
\begin{eqnarray}\label{eq:phi_-error}
  \lefteqn{\sum_{k=1}^M\Tr\big[\theta_{r,k}\big(T_\b(-\ic h\nabla)
  -V(\x)-Ch^2r^{-2}\big)\theta_{r,k}\big]_-}&&\\\nonumber
  &\geq&(1-\e)\sum_{k=1}^M\Tr\big[\theta_r(T_\b(-\ic
  h\nabla) -
  z_k|\x|^{-1})\theta_r\big]_{-} - C(\e\delta^{-1/2}+\e^{-3/2}\delta^2  
  +\e^{-3}\delta^5)h^{-2}\,.
\end{eqnarray}
For the corresponding semi-classical expression we 
have from \eqref{eq:VcondS} and \eqref{eq:Crminest} (using
\(\delta<1/2\) and \(|x_{-}-y_{-}|\le |x-y|\)) that
\begin{eqnarray}
  \left|(2\pi
    h)^{-3}\!\!\!\int\phi_-(v)^2\big[\tfrac12p^2-V(v)\big]_-\,dvdp
  -\sum_{k=1}^M(2\pi 
  h)^{-3}\!\!\!\int\theta_r(v)^2\big[\tfrac12p^2
  - z_k|v|^{-1}\big]_-\,dvdp\right|\nonumber\\   
  \leq C\delta^{1/2}h^{-2}.\label{eq:phi_-scerror}
\end{eqnarray}
A simple rescaling applied to the local Hydrogen result in 
Lemma~\ref{lem:Coulomb} gives that 
\begin{eqnarray}
  \lefteqn{\biggl|\Tr\big[\theta_r(T_\b(-\ic
  h\nabla) - z_k|\x|^{-1})\theta_r\big]_--(2\pi
  h)^{-3}\!\!\!\int\theta_r(v)^2\big[\tfrac12p^2-z_k|v|^{-1}\big]_-\,dvdp} 
  \hspace{3.5cm}&&\nonumber\\&&-z_k^2 
  h^{-2}\mathcal{S}(\beta^{1/2}h^{-1}z_k)\biggr|
  \leq C h^{-2}(h^{-2}r)^{-1/10}=Ch^{-2}\delta^{1/10}.\label{eq:phi_-S}
\end{eqnarray}
Combining \eqref{eq:phi_-error}, \eqref{eq:phi_-scerror}, and
\eqref{eq:phi_-S}, using that $\mathcal{S}$ is a bounded function
(since it is non-increasing; see Lemma~\ref{lem:Coulomb}), that
\(\delta<1/2\), and that
$$
  (2\pi
  h)^{-3}\int\theta_r(v)^2\big[\tfrac12p^2-z_k|v|^{-1}\big]_-\,dvdp\leq 
  Ch^{-3}r^{1/2}=C h^{-2}\delta^{-1/2}\,,
$$
and choosing $\e=\delta$, we conclude that 
\begin{eqnarray}
  \Tr\big[\phi_-\big(T_\b(-\ic h\nabla)
    -V(\x)-Ch^2r^{-2}\big)\phi_-\big]_-&\geq&(2\pi
  h)^{-3}\!\!\!\int\phi_-(v)^2\big[\tfrac12q^2-V(v)\big]_-\,dvdq\nonumber
  \\&&\!\!\!\!\!\!\!\!\!\!\!\!\!\!\!\!
  {}+
  h^{-2}\sum_{k=1}^{M}z_k^2\mathcal{S}(\beta^{1/2}h^{-1}z_k)-C\delta^{1/10}h^{-2}\!. 
  \label{eq:vtfscott} 
\end{eqnarray}

\noindent{\bf Control of the second term in \eqref{eq:TF-loc-2}.}

Here we use the local rescaled semi-classics in
Theorem~\ref{corollary}. Before we apply our semi-classical estimates
on the support of $\Phi_-\phi_+$ we localise using the functions
$\varphi_u$ from (\ref{eq:phiuprop}) for general $M$ and with
\(\ell(u)\) as in \eqref{eq:ldefinition}, with $\ell_0=r/4$. From
\eqref{eq:defrR} and the choice of \(\delta\) it follows that
\(\ell_0<1\) for \(h\) small enough. If $x$ is in the support of
$\Phi_-\phi_+$ and in the support of $\varphi_u$ then
$d(u)>r/2=2\ell_0$ since (using \eqref{eq:elldcom}) 
$$
  r\leq d(x)\leq d(u)+\ell(u)<d(u)+\max\{d(u),\ell_0\}\,,
$$
and also $d(u)\leq 2R+1$ since $\ell(u)<1/2$.
Using again the relativistic IMS localisation \eqref{eq:IMS} we thus
have 
\begin{eqnarray}\label{eq:localise in TF}
  \lefteqn{
    \Phi_-\phi_+\big(T_\b(-\ic h\nabla)
    -V(\x)-Ch^2r^{-2}\chi_{\{d(x)\le3r\}}
    -Ch^2R^{-2} \big)\phi_+\Phi_-} 
  &&\nonumber\\
  &&\begin{array}{ll}\\\!\!\!\!\!\!=\!\!\!\!\!\!\!\!
    \displaystyle\int\limits_{r/2< 
      d(u)<2R+1}\!\!\!\!\!\!\!\!\!\!\!\!\!\!\!
  \Phi_-\phi_+\varphi_u\big(T_\b(-\ic h\nabla)-V(\x)-Ch^2r^{-2}\chi_{\{d(x)\le3r\}}
    -Ch^2R^{-2} \big)\varphi_u 
  \phi_+\Phi_-\,\ell(u)^{-3}\,du
  \end{array}\nonumber\\
  &&\quad\,-\,\!\!\!\!\!\!\!\!
  \int\limits_{r/2<d(u)<2R+1}\!\!\!\!\!\!\!\!\!\!\!\!\!\!\!
  \Phi_-\phi_+L_{\varphi_u}\phi_+\Phi_-  
  \ell(u)^{-3}\,du\,. 
\end{eqnarray}

Concerning $L_{\varphi_u}$, Theorem~\ref{IMS-error-est} with
$\ell=\ell(u)/2$ and $\Omega=\Omega_u=\{|x-u|\le 3\ell(u)/2\}$ (and
$m=h^{-2}$ and $\alpha=\beta^{1/2}h^{-1}(\leq 1)$) gives,
using \eqref{eq:phiuestimate}, that
\begin{align}\label{eq:another error}
 L_{\varphi_u} \le C h^2\ell(u)^{-2}\chi_{\Omega_u} + Q_{\varphi_u}\,,
\end{align}
with
\bea\label{eq:Tr Q2}
  \Tr[Q_{\varphi_u}] \le C \beta^{-1}{\rm
    e}^{-(\beta^{1/2}h)^{-1}\ell(u)}\, 
  \leq C_Nh^N
\eea
for all $N>0$ as a consequence of \(\ell(u)\ge(1+M)^{-1}r/8\) and
\(\beta\le h^2\). Thus
\beax
  \Tr\Big[\!\!\!\!\!\!\!\!\!\!\!\!\!
  \int\limits_{r/2< d(u)<2R+1}\!\!\!\!\!\!\!\!\!\!\!\!
  \Phi_{-}\phi_+Q_{\varphi_{u}}\phi_+\Phi_{-}\ell(u)^{-3}\,du\Big]
  &\le& C_N h^{N}\ \text{ for all } N>0\,.
\eeax

By the same arguments as in the proof of Lemma~\ref{lem:Coulomb} above
(see \eqref{eq:overlap two phi-u's} and \eqref{eq:trick to join
  errors}) we join the new localisation error term (from
\eqref{eq:localise in TF}, \eqref{eq:another error}) with the previous
localisation errors from \eqref{eq:good trace} and \eqref{eq:good
  trace2}. Since $\ell(u)\leq \max\{d(u),r/4\}$ we have $R^{-2}\leq C
\ell(u)^{-2}$ for $d(u)\leq 2R+1$ (and $h$ small enough when
\(\mu=0\); for \(\mu\neq0\), use \(\ell(u)<1/2\)) and, by
\eqref{eq:ell and d} (valid on the support of \(\varphi_u\) when
\(d(u)>r/2=2\ell_0\)), 
$$
  r^{-2}\chi_{\{d(x)\le3r\}}(x)\varphi_u(x)^2\leq C\ell(u)^{-2}\varphi_u(x)^2\,.
$$ 
This way, we have proved that
\begin{eqnarray}
  \lefteqn{\Tr\big[\Phi_-\phi_+\big(T_\b(-\ic h\nabla)
  -V(\x)-Ch^2r^{-2}\chi_{\{d(x)\le3r\}}
  -Ch^2R^{-2} \big)\phi_+\Phi_-\big]_- } 
  &&\label{eq:tfphilow}\\\nonumber
  &\geq&\int\limits_{r/2< d(u)<2R+1}\!\!\!\!\!\!\!\!\!\!\!\!
  \Tr\big[\phi_+\varphi_u\left(T_\b(-\ic h\nabla) - V(\x) - 
  Ch^2\ell(u)^{-2}
  \right)\varphi_u\phi_+\big]_- \,
  \ell(u)^{-3}\,du \,\,-\,C h^{-2+1/10}\,.
\end{eqnarray}
Note that there is no need to write $\Phi_-$ on the right side, since
in general $\Tr(\Phi A\Phi)_-\geq \Tr A_-$ for any self-adjoint
operator $A$ and any function $0\le\Phi\le1$.

For \(u\) such that $d(u)>r/2=2\ell_0$ and \(d(u)<2R+1\) we have from 
\eqref{eq:VcondD} and \eqref{eq:defrR} that
\begin{align*}
  \sup_{|x-u|\leq\ell(u)}|\partial^\eta (V(x)-Ch^2\ell(u)^{-2})|&\leq C
  f(u)^2\ell(u)^{-|\eta|}\hbox{ for }|\eta|\leq3\,,\\ 
  \|\partial^\eta(\phi_+\phi_u)\|_\infty&\leq
  C_\eta\ell(u)^{-|\eta|}\hbox{ for }|\eta|\leq7\,, 
\end{align*}
with 
\begin{equation}\label{eq:definition-f}
  f(u)= \left\{\begin{array}{ll}
      d(u)^{-1/2} &\hbox{ if }\mu\ne0\\ 
  \min\{d(u)^{-1/2},d(u)^{-3/2}\} &\hbox{ if }\mu=0
  \end{array}\right. .
\end{equation}
We have also used
that $d(u)\geq \delta^{-1}h^2/2\geq h^2$ and $\min\{1,d(u)\}\leq
C\ell(u)$. 

We are therefore in a position to use the rescaled semi-classics in
Theorem~\ref{corollary} on the ball $\{|x-u|\leq\ell \}$ with
$\ell=\ell(u)$, $f=f(u)$, and $\phi=\phi_+\phi_u$ for each $u$ with 
$r/2\le d(u)\le 2R+1$. Note in particular that $\beta f^2(u)\le\beta
d(u)^{-1}\leq 2\beta/ r=2\beta \delta h^{-2}\leq 2\delta\leq 1$ .
We conclude that for all $u$ with
$r/2\le d(u)\le 2R+1$,
\begin{eqnarray}
  \lefteqn{\Big|\Tr\big[\phi_+\varphi_u\big(T_\b(-\ic
  h\nabla)-V(\x)-Ch^2\ell(u)^{-2}\big)\varphi_u\phi_+\big]_-}&& 
  \nonumber\\&&{}-(2\pi h)^{-3}\int\phi_+(v)^2\varphi_u(v)^2
  \big[\sqrt{\beta^{-1}q^2+\beta^{-2}}-
  \beta^{-1}-V(v)-Ch^2\ell(u)^{-2}\big]_-\, dvdq\Big|
  \nonumber\\&\label{eq:v-tfsclow}
  \leq& C h^{-2+1/5} f(u)^{4-1/5}\ell(u)^{2-1/5}\,.
\end{eqnarray}

The semi-classical integral may be estimated using \eqref{eq:semiclassicalcomp}
\begin{eqnarray}
  \lefteqn{\left|\int\big[\sqrt{\beta^{-1}q^2+\beta^{-2}}-
  \beta^{-1}-V(v)-Ch^2\ell(u)^{-2}\big]_-\,dq
  -\int\big[\mfr{1}{2}q^2-V(v)\big]_-\,dq\right|}&&
  \nonumber\\&\leq&
  Ch^2\big(|V(v)|+h^2\ell(v)^{-2}\big)^{7/2}
  +Ch^2\ell(v)^{-2}\big(|V(v)|+h^2\ell(v)^{-2}\big)^{3/2}\nonumber\\
  &\leq&
  C\big(h^2|V(v)|^{7/2}+h^2\ell(v)^{-2}|V(v)|^{3/2}+(h\ell(v)^{-1})^5\big)\,,
  \label{eq:scintegrals} 
\end{eqnarray}
for $v$ in the support of $\varphi_u$, since then we have $\ell(v)\leq
3\ell(u)/2$ (see \eqref{eq:bound nabla ell}) and $|V(v)|\leq
Cd(v)^{-1}\leq Cd(u)^{-1}\leq Ch^{-2}$ (see \eqref{eq:VcondD},
\eqref{eq:ell and d} and \eqref{eq:elldcom}). 
We have also used that $\beta\leq h^2$ and that, by \eqref{eq:elldcom}
and \eqref{eq:defrR} (\(\ell_0=r/4\)), 
$h^2\ell(v)^{-2}\leq C h^2r^{-2}=C \delta^{2}h^{-2}\leq Ch^{-2}$.

Combining \eqref{eq:tfphilow}, \eqref{eq:v-tfsclow}, and
\eqref{eq:scintegrals} (remembering that $d(u)\leq C d(v)$ if $v$ is in
the support of $\varphi_u$ and $d(u)>r/2=2\ell_0$)
we find, using \eqref{eq:phiuprop}, \eqref{eq:VcondD}, and
\eqref{eq:definition-f}, that
\begin{eqnarray}
  \lefteqn{\Tr\big[\Phi_-\phi_+\big(T_\b(-\ic h\nabla)
  -V(\x)-Ch^2r^{-2}\chi_{\{d(x)\le3r\}}
  -Ch^2R^{-2} \big)\phi_+\Phi_-\big]_- } \nonumber
  &&\\
  &\geq&(2\pi h)^{-3}\int
  \phi_+(v)^2\big[\mfr{1}{2}q^2-V(v)\big]_-\,dq\,dv
  -C \!\!\!\!\!\!\!\!\!\!\int\limits_{r/2<d(u)<2R+1} \!\!\!\!\!\!\!\!\!\!
  h^{-2+1/5}f(u)^{19/5}\ell(u)^{-6/5}\, du \nonumber\\
  &&{}-C\!\!\!\!\!\!\!\!\!\!\!\!\!\int\limits_{C^{-1}r< d(v)< 2R+2} \!\!\!\!\!\!\!\!\!\!
  h^{-3}(h^2
  d(v)^{-7/2}+h^2\ell(v)^{-2}f(v)^{3}+(h\ell(v)^{-1})^5)\,
  dv
  -Ch^{-2+1/10}\,.\label{eq:vtf2ndterm} 
\end{eqnarray}

If $\mu\ne0$ the error term in \eqref{eq:vtf2ndterm} is controlled as
follows:
\begin{align}\nonumber
  &\int\limits_{C^{-1}r<d(v)<
      2R+2}\!\!\!\!\!\!\!\!\!\!\big(h^{-2+1/5}f(v)^{19/5}\ell(v)^{-6/5}+ 
  h^{-1}d(v)^{-7/2}+h^{-1}\ell(v)^{-2}f(v)^{3}+h^2\ell(v)^{-5}\big)\,
  dv
  \\\nonumber&\qquad\leq\ \, 
  C \!\!\!\!\!\int\limits_{C^{-1}r<|v|<
    2R+2}\!\!\!\!\!\!\!\!\!\!\big(h^{-2+1/5}|v|^{-19/10}\min\{1,|v|\}^{-6/5}+ 
  h^{-1}|v|^{-7/2}\nonumber
  \\&{}\qquad\qquad\qquad\qquad\qquad\qquad\qquad\nonumber
  +h^{-1}\min\{1,|v|\}^{-2}
  |v|^{-3/2}+h^{2}\min\{1,|v|\}^{-5}\,\big)\,dv
  \\&\qquad\leq\ \, 
  C\big(h^{-2+1/5}(R^{11/10}+r^{-1/10})+h^{-1}r^{-1/2}
  +h^{-1}R^{3/2}+h^2R^3+h^2r^{-2}\big)
  \nonumber
 \\&\qquad\leq\ \, 
  Ch^{-2+1/10}\,,
  \label{eq:mune0} 
\end{align}
with the choices \eqref{eq:defrR} where $R=R_\mu$ is a constant. 

If $\mu=0$ we get instead 
\begin{align}
  \int\limits_{C^{-1}r<d(v)<
      2R+2}\!\!\!\!\!\!\!\!\!\!\big(&h^{-2+1/5}f(v)^{19/5}\ell(v)^{-6/5}+ 
  h^{-1}d(v)^{-7/2}+h^{-1}\ell(v)^{-2}f(v)^{3}+h^2\ell(v)^{-5}\big)\,
  dv
  \nonumber\\ 
  &\leq \, 
  C \big(h^{-2+1/5}r^{-1/10}+h^{-1}r^{-1/2}+h^2R^3+h^2r^{-2}\big)
  \leq \, Ch^{-2+1/10}.\label{eq:mu=0} 
\end{align}
If we insert the last two estimates \eqref{eq:mu=0} and
\eqref{eq:mune0} into \eqref{eq:vtf2ndterm} and then together with
\eqref{eq:vtfouter} and \eqref{eq:vtfscott} into \eqref{eq:TF-loc-2}
we arrive at a lower bound on the quantum energy corresponding to one
direction in \eqref{TF-Scott}. 

\noindent{\bf Step 2: Upper bound on the quantum energy.}

We obtain an upper bound on the quantum energy by choosing the density 
matrix
\begin{equation}\label{eq:gammaup}
  \gamma =
  \sum_{k=1}^M\theta_{r,k}\gamma_k\theta_{r,k}
  \ +   \!\!\!\!\!\! \int\limits_{d(u)<2R+1}\!\!\!\!\!\!
  \phi_+ \varphi_u \gamma_u
  \varphi_u \phi_+\,\ell(u)^{-3}\, du \,,
\end{equation}
where $\gamma_k$, $k=1,\ldots,M$, are the density matrices
$$ 
  \gamma_k=\chi\big(\theta_{r,k}\big(T_\b(-\ic h\nabla)
  -z_k|\x-r_k|^{-1}\big)\theta_{r,k}\big)
$$
and $\gamma_u$, $u\in\R^3$, are the density matrices given 
in Theorem~\ref{corollary} for the potential
$V$ with $B_\ell$
being the ball centered at $u$, $\ell=\ell(u)$, $f=f(u)$ (see
\eqref{eq:definition-f}), and 
$\phi=\phi_+\varphi_u$. Since 
\begin{equation}\label{eq:thetasum}
  \sum_{k=1}^M\theta_{r,k}^2+\phi_+^2=\phi_-^2+\phi_+^2=1\,,
\end{equation}
we immediately see from \eqref{eq:phiuprop} that $\gamma$ is a
density matrix. 

Using this density matrix as a trial state we obtain from
Theorem~\ref{corollary} that
\begin{eqnarray}
  \lefteqn{\Tr\big[T_\b(-\ic h\nabla)-V(\x)\big]_-}&&\nonumber\\ &\le& 
  \sum_{k=1}^M\Tr\big[\theta_{r,k}\big(T_\b(-\ic h\nabla)
  -V(\x)\big)\theta_{r,k}\big]_-\nonumber\\&& 
  {}+(2\pi h)^{-3}\!\!\!\!\!\!
  \int\limits_{d(u)<2R+1}\!\!\!\!\!
  \phi_+(v)^2\varphi_u(v)^2
  \big[\sqrt{\beta^{-1}q^2+\beta^{-2}}-
  \beta^{-1}-V(v)\big]_-\ell(u)^{-3}\,dvdqdu\nonumber\\
  &&{}+Ch^{-2+1/5}\!\!\!\!\!\!\!\!\!\!
  \int\limits_{r/2<d(u)<2R+1}\!\!\!\!\!\!\!\!\!\!
  f(u)^{19/5}\ell(u)^{-6/5}\,du\,,
  \label{eq:SCupper1}
\end{eqnarray}
where we have used the fact that $\phi_+$ and $\varphi_u$ have
overlapping supports only if $d(u)> r/2$. The last error term is
estimated by $Ch^{-2+1/10}$ as in the lower bound. 

Using that $\sqrt{\beta^{-1}q^2+\beta^{-2}}-\beta^{-1}\leq
\mfr{1}{2}q^2$ and the normalization of $\varphi_u$
\eqref{eq:phiuprop} we find that
\begin{eqnarray*}
  \lefteqn{(2\pi h)^{-3}\!\!\!\!\!
  \int\limits_{d(u)<2R+1}\!\!\!\!\!\phi_+(v)^2\varphi_u(v)^2
  \big[\sqrt{\beta^{-1}q^2+\beta^{-2}}-
  \beta^{-1}-V(v)\big]_-\ell(u)^{-3}\,dvdqdu}&&\\&\leq&
  (2\pi h)^{-3}\int\phi_+(v)^2
  \big[\mfr{1}{2}q^2-V(v)\big]_-\, dvdq\\
  &&{}-(2\pi h)^{-3}\!\!\!\!\!\int\limits_{d(u)>2R+1}\!\!\!\!\!
  \phi_+(v)^2\varphi_u(v)^2
  \big[\mfr{1}{2}q^2-V(v)\big]_-\ell(u)^{-3}\, dvdqdu\\
  &\leq&(2\pi h)^{-3}\int\phi_+(v)^2
  \big[\mfr{1}{2}q^2-V(v)\big]_-\,
  dvdq +Ch^{-3}\!\!\!\!\int\limits_{d(v)>2R}\!\!\!\!\!|V(v)_-|^{5/2}\,dv\,.
\end{eqnarray*}
If $\mu\ne0$ the last term vanishes by the choice of \(R=R_{\mu}\). If
$\mu=0$ it may be estimated using \eqref{eq:VcondD} and
\eqref{eq:defrR} as 
$$
  Ch^{-3}\int\limits_{d(v)>2R}|V(v)_-|^{5/2}dv\leq C\,.
$$
Together with \eqref{eq:phi_-scerror}, \eqref{eq:phi_-S},
\eqref{eq:thetasum}, and \eqref{eq:SCupper1} this gives the 
proof of \eqref{TF-Scott-trial}, and therefore finishes the proof of 
\eqref{TF-Scott}. 

\noindent{\bf Step 3: Properties of the density.}

We will now show that the density matrix $\gamma$ in
\eqref{eq:gammaup} satisfies the two requirements
\eqref{eq:rhogammaint} and \eqref{eq:rhogamma6/5}. 

The density of $\gamma$ is 
\begin{equation}\label{eq:some-dens}
  \rho_\gamma(x)=\sum_{k=1}^M\theta_{r,k}(x)^2\rho_k(x)
  +\int_{d(u)<2R+1}\varphi_u^2(x)\phi_+^2(x)\rho_u(x)\ell(u)^{-3}\,du\,,
\end{equation}
where $\rho_k$ for $k=1,\ldots,M$ is the density of the density matrix
$\gamma_k$ and $\rho_u$ for $u\in\R^3$ is the density for $\gamma_u$. 
We first control the $6/5$-norm and the $1$-norm of $\theta_{r,k}^2\rho_k$. 
If $\beta^{1/2}h^{-1}<1/2$ we use the combined Daubechies-Lieb-Yau
inequality  (Theorem~\ref{thm:L-Y-D})
with $\alpha=\beta^{1/2}h^{-1}\leq1/2$, $\nu=2z_k$, and $m=h^{-2}$ to
obtain that 
\begin{eqnarray*}
  0\geq\Tr\big[\theta_{r,k}\gamma_k\theta_{r,k}(T_\beta(-{\rm
    i}h\nabla)-z_k|{\hat x}-r_k|^{-1})\big]
  &\geq& \mfr12\,\Tr\big[\theta_{r,k}\gamma_k\theta_{r,k}T_\beta(-{\rm
    i}h\nabla)\big] 
  \\&&{}-Cz_k^{5/2}h^{-2}-Ch^{-3}z_k^{5/2}r^{1/2}-Cz_k^4h^2\\&\geq&
  \mfr12\,\Tr\big[\theta_{r,k}\gamma_k\theta_{r,k}T_\beta(-{\rm
    i}h\nabla)\big]-Ch^{-2}\delta^{-1/2}\,, 
\end{eqnarray*}
where the constant $C$ depends on $z_k$. Hence we have that
\begin{equation}\label{eq:kinTrace}
  \Tr\big[T_\beta(-{\rm
    i}h\nabla)\theta_{r,k}\gamma_k\theta_{r,k}\big]\leq 
  Ch^{-2}\delta^{-1/2}=Ch^{-5/2}\,. 
\end{equation}
Using \eqref{Daub manybodyBIS} with $\alpha=\beta^{1/2}h^{-1}\leq1$
and $m=h^{-2}$, \eqref{eq:kinTrace} implies that
\begin{eqnarray}\label{eq:first 6/5}\nonumber
  \int (\theta_{r,k}^2\rho_k)^{6/5}&\leq& Ch^{-36/25}
  \left(\int_{h^2(\theta_{r,k}^2\rho_k)^{1/3}\leq \beta^{-1/2}h}
    h^2(\theta_{r,k}^2\rho_k)^{5/3}\right)^{18/25}r^{21/25} 
  \\&&{}+C\left(\int_{h^2(\theta_{r,k}^2\rho_k)^{1/3}>
      \beta^{-1/2}h}\beta^{-1/2}h(\theta_{r,k}^2\rho_k)^{4/3}\right)^{18/20} 
  r^{3/10}\nonumber\\&\leq&
  Ch^{-36/25}h^{-9/5}h^{21/25}+Ch^{-9/4}h^{3/10}\leq Ch^{-12/5}\,,
\end{eqnarray}
where we have used that $r=h$ and that \(h\) is bounded above by a
constant. Likewise we find
$$
  \int \theta_{r,k}^2\rho_k\leq C h^{-3/2}\,.
$$

The case when $1/2\leq \beta^{1/2} h^{-1}\leq 1$ is more
complicated. We have to treat the region within the radius $r_-=h^2$
from the nucleus $z_k$ differently. Let
$\widetilde\theta_\pm(x)=\theta_\pm(|x-r_k|/h^2)$. Using the
relativistic IMS formula (Theorem~\ref{IMS}) and
Theorem~\ref{IMS-error-est} with $\ell=h^2/2$, $m=h^{-2}$,
$\alpha=\beta^{1/2}h^{-1}$, and $\Omega=\{|x-r_k|<3h^2\}$ we find that  
\begin{eqnarray*}
  \lefteqn{0\geq\Tr\big[\theta_{r,k}\gamma_k\theta_{r,k}
  (T_\beta(-{\rm i}h\nabla)-z_k|{\hat x}-r_k|^{-1})\big]}&&\\ 
  &\geq&
  \Tr\big[\widetilde\theta_-\gamma_k\widetilde\theta_-
  (T_\beta(-{\rm i}h\nabla)-z_k|{\hat x}-r_k|^{-1}-Ch^{-2})\big]\\&&{} 
  +\Tr\big[\theta_{r,k}\widetilde\theta_+\gamma_k\theta_{r,k}
  \widetilde\theta_+(T_\beta(-{\rm i}h\nabla)-z_k|{\hat x}-r_k|^{-1}-
  h^{-2}\chi_{\Omega})\big]-Ch^{-2}\,.
\end{eqnarray*}
To treat the first term we use the inequality (see \eqref{eq:new-critical})
$$
  \sqrt{-\D}-\frac{2}{\pi|{\hat x}|}\geq A_s(-\Delta)^{s}-B_s\,,
$$
which holds for all $0\leq s<1/2$ and $A_s,B_s>0$ being constants
depending only on $s$. Hence, using that $h$ is bounded above by a
constant and that $1\le\beta^{-1/2}h\leq2$ we get
\begin{eqnarray*}
  0\geq\Tr\big[\widetilde\theta_-\gamma_k\widetilde\theta_-
  (T_\beta(-{\rm i}h\nabla)-z_k|{\hat x}-r_k|^{-1}-Ch^{-2})\big]
  &\geq&
  \Tr\big[\widetilde\theta_-\gamma_k\widetilde\theta_-(A_s(-\Delta)^{s}-
  C_sh^{-2})\big]\,.
\end{eqnarray*}
We appeal to the standard (Daubechies)-Lieb-Thirring inequality
$$
  \Tr\big[(-\Delta)^{s}\widetilde\theta_-\gamma_k\widetilde\theta_-\big]
  \geq c\int (\widetilde\theta_-^2\rho_k)^{(3+2s)/3}\,,
$$
which holds for all $s\in(0,3)$. We obtain
that (with
all constants depending on $0<s<1/2$)
\begin{eqnarray*}
  \Tr\big[\widetilde\theta_-\gamma_k\widetilde\theta_-
  (T_\beta(-{\rm i}h\nabla)-z_k|{\hat x}-r_k|^{-1}-Ch^{-2})\big]
  &\geq&c\int (\widetilde\theta_-^2\rho_k)^{(3+2s)/3}-Ch^{-2}\int
  (\widetilde\theta_-^2\rho_k)\\
  &\geq&(c/2)\int
  (\widetilde\theta_-^2\rho_k)^{(3+2s)/3}-Ch^{(4s-3)/s}\,. 
\end{eqnarray*}
Using the Daubechies inequality (Theorem~\ref{Daubechies}) we find as
above that  
\begin{eqnarray*}
  \lefteqn{\Tr\big[\theta_{r,k}\widetilde\theta_+\gamma_k\theta_{r,k}
  \widetilde\theta_+(T_\beta(-{\rm i}h\nabla)-z_k|{\hat x}-r_k|^{-1}-
  h^{-2}\chi_{\Omega})\big]}&&\\&\geq&
  c\,\Tr\big[\theta_{r,k}\widetilde\theta_+\gamma_k\theta_{r,k}\widetilde\theta_+ 
  T_\beta(-{\rm i}h\nabla)\big]-Ch^{-5/2}\,.
\end{eqnarray*} 
By choosing $s$ sufficiently close to $1/2$ and using that $h$ is
bounded by a constant we conclude that  
$$
0\geq c\int
  (\widetilde\theta_-^2\rho_k)^{(3+2s)/3}+c\,\Tr\big[\theta_{r,k}
  \widetilde\theta_+\gamma_k\theta_{r,k}\widetilde\theta_+
  T_\beta(-{\rm i}h\nabla)\big]-Ch^{-5/2}\,.
$$
As above it follows from this, choosing $s$ sufficiently close to
$1/2$, that we still have
\begin{equation}\label{eq:thetarkrho}
  \int (\theta_{r,k}^2\rho_k)^{6/5}\leq Ch^{-12/5},\quad
  \int \theta_{r,k}^2\rho_k\leq C h^{-3/2}\,.
\end{equation}
Using that $r=h$ and that from \eqref{eq:VcondD} $|V(x)|\leq C
d(x)^{-1}$ we also have 
\begin{equation}\label{eq:thetarV}
  \int (h^{-3}\theta_{r,k}^2|V_-|^{3/2})^{6/5}\leq Ch^{-12/5}\,,\quad
  \int h^{-3}\theta_{r,k}^2|V_-|^{3/2}\leq C h^{-3/2}\,.
\end{equation}

We move to the second term in \eqref{eq:some-dens}.
By the rescaled semi-classics (Theorem~\ref{corollary}) 
we have on the support of $\varphi_u\phi_+$ that (for \(f(u)\), see
\eqref{eq:definition-f}) 
$$
  \left|\rho_u(x)-2^{1/2}(3\pi^2)^{-1} h^{-3}|V(x)_-|^{3/2}\right|
  \leq Ch^{-2-1/10}f(u)^{21/10}\ell(u)^{-9/10}+Ch^{-2}|V(x)_-|^{3/2}\,,
$$
where we have used that on the support of $\varphi_u\phi_+$ we have
$|V(x)|\leq C d(u)^{-1}\leq Cr^{-1} \leq C h^{-1}\leq
Ch\beta^{-1}$, since $d(u)\geq r/2$ if $\varphi_u\phi_+$ is non-vanishing.  
We moreover have on the support of $\varphi_u\phi_+$ that 
$|V(x)_-|^{3/2}\leq C f(u)^3\leq Cf(u)^{21/10}\ell(u)^{-9/10} $.  For
$r/2<d(u)\leq 1$ this is because $\ell(u)^{-1}\geq
d(u)^{-1}=f(u)^{2}\geq f(u)$ (see \eqref{eq:ell and d}) and for
$d(u)>1$ we simply use that $\ell(u)\leq 1$ and $f(u)\leq1$.  Hence
\begin{equation}\label{eq:local-dens}
  \big\|\varphi_u^2\phi_{+}^{2}\big(\rho_\gamma-2^{1/2}(3\pi^2)^{-1}
      h^{-3}|V_-|^{3/2}\big)\big\|_{6/5} 
  \leq
  Ch^{-2-1/10}f(u)^{21/10}\ell(u)^{8/5}\,,
\end{equation}
Using \eqref{eq:gammaup} and \eqref{eq:thetasum}, we have that
\begin{eqnarray}\label{eq:long 6/5}
  \lefteqn{\big\|\rho_\gamma-2^{1/2}(3\pi^2)^{-1}
      h^{-3}|V_-|^{3/2}\big\|_{6/5}}&&\\ \nonumber
  &\leq&
  \sum_{k=1}^M\big(\|\theta_{r,k}^2\rho_k\|_{6/5}+Ch^{-3}
  \|\theta_{r,k}^2|V_-|^{3/2}\|_{6/5}\big)\\&& \nonumber
  +\int_{d(u)<2R+1}
  \big\|\varphi_u^2\phi_+^2\big(\rho_u-2^{1/2}(3\pi^2)^{-1}
  h^{-3}|V_-|^{3/2}\big)\big\|_{6/5}\ell(u)^{-3}du
  \\&&+\int_{d(u)>2R+1} \nonumber
  Ch^{-3}\big\|\varphi_u^2\phi_+^2|V_-|^{3/2}\big\|_{6/5}\ell(u)^{-3}du\,. 
\end{eqnarray}
The last term is non-zero only in the case $\mu=0$ in which case it is
easily seen by \eqref{eq:VcondS} and \eqref{eq:defrR} to be bounded by
$Ch^{-3/2}$. Thus, combining \eqref{eq:first 6/5}--\eqref{eq:local-dens},
\eqref{eq:long 6/5} implies that
\begin{eqnarray*}
  \big\|\rho_\gamma-2^{1/2}(3\pi^2)^{-1}
  h^{-3}|V_-|^{3/2}\big\|_{6/5}\leq Ch^{-2}
  +Ch^{-2-1/10}\int_{C^{-1}r<d(u)<2R+1}f(u)^{21/10}\ell(u)^{-7/5}du\,.
\end{eqnarray*}
The last integral is easily seen to be bounded and we arrive at 
\eqref{eq:rhogamma6/5}. 

To control the integral of the density we estimate
\begin{eqnarray*}
  \lefteqn{\left|\int\rho_\gamma(x)\,dx-2^{1/2}(3\pi^2)^{-1} h^{-3}\int
  |V(x)_-|^{3/2}\,dx\right|}&&\\
  &\leq&\sum_{k=1}^M\big(\|\theta_{r,k}^2\rho_k\|_{1}+Ch^{-3}
  \|\theta_{r,k}^2|V_-|^{3/2}\|_{1}\big)\\&& 
  +\int_{d(u)<2R+1}\left|\int
    \varphi_u^2\phi_+^2\big(\rho_u(x)-2^{1/2}(3\pi^2)^{-1} 
    h^{-3}|V(x)_-|^{3/2}\big)\,dx\right|\ell(u)^{-3}\,du
  \\&&+\int_{d(u)>2R+1}
  Ch^{-3}\big\|\varphi_u^2\phi_+^2|V_-|^{3/2}\big\|_{1}\ell(u)^{-3}\,du\,.
\end{eqnarray*}
As before the last term is bounded by $Ch^{-3/2}$.
For the middle term we again see from the rescaled
semi-classics (Theorem~\ref{corollary}) that 
\begin{eqnarray*}
  \lefteqn{\left|\int \varphi_u^2\phi_+^2\big(\rho_u(x)-2^{1/2}(3\pi^2)^{-1}
    h^{-3}|V(x)_-|^{3/2}\big)dx\right|}&&\\
  &\leq& Ch^{-2+ 1/5}
  f(u)^{9/5}\ell(u)^{9/5}+Ch^{-1}\int\varphi_u^2(x)\phi_+^2(x)|V(x)_-|^{5/2}dx\\
  &\leq& Ch^{-2+ 1/5}
  f(u)^{9/5}\ell(u)^{9/5}+Ch^{-1} f(u)^5 \ell(u)^3\,,
\end{eqnarray*}
where we have used that $\beta\leq h^2$.
The estimate (\ref{eq:rhogammaint}) follows since both integrals
$$
\int f(u)^{9/5}\ell(u)^{9/5}\ell(u)^{-3} du\quad \hbox{and} \quad \int f(u)^5 du 
$$
are bounded (recall that \(f(u)\) is given in \eqref{eq:definition-f}).
This finishes the proof of Theorem~\ref{thm:TF-semicl}, except for the
continuity of the function \(\mathcal{S}\) from
Lemma~\ref{lem:Coulomb}. We will need a lemma to prove this. This
lemma also gives an alternative characterization of the function
$\mathcal S$.
\begin{lemma}[{\bf Scott-corrected pushed-up Hydrogen}]\label{cor:alt-S}
Let \(\mathcal{S}:[0,2/\pi]\to\R\) be the function from
Lemma~\ref{lem:Coulomb}. Then there exists a constant \(C>0\) such
that, for all \(\alpha\in[0,2/\pi]\) and \(\kappa\in(0,1]\), 
\bea\label{eq:est-S}  \nonumber
  \lefteqn{\Big|\,\displaystyle
  \Tr\big[\sqrt{-\alpha^{-2}\Delta+\alpha^{-4}}-\alpha^{-2}-|{\hat
    x}|^{-1} 
  +\kappa\big]_-} 
  \\&&\qquad\qquad
  \displaystyle{}-(2\pi)^{-3}\int\big[\mfr12
  p^2-|v|^{-1}+\kappa\big]_-\,dpdv
 -\mathcal S(\alpha)\Big|
    \le C\kappa^{1/20}\,.
\eea
Here, as before, 
$\sqrt{-\alpha^{-2}\Delta+\alpha^{-4}}-\alpha^{-2} =
-\Delta/2$, 
when $\alpha=0$.
\end{lemma}
\begin{proof}[Proof of Lemma~\ref{cor:alt-S}] 
A simple rescaling, changing $x\to\kappa^{-1}\pi x/2$, gives
$$
  \Tr\big[\sqrt{-\alpha^{-2}\Delta+\alpha^{-4}}-\alpha^{-2}-|{\hat
    x}|^{-1} 
  +\kappa\big]_-=\kappa\,
  \Tr\big[\sqrt{-\beta^{-1}h^2\Delta
  +\beta^{-2}}-\beta^{-1}-\frac{2}{\pi|{\hat x}|}
  +1\big]_-\,,
$$
where 
$\beta=\kappa\alpha^2$ and $h=2\kappa^{1/2}/\pi$.
We have $\beta\leq h^2$.

The semi-classical integral may be rewritten in the same fashion,
$$
  (2\pi)^{-3}\int\big[\mfr12
  p^2-|v|^{-1}+\kappa\big]_-dpdv=\kappa\,(2\pi h)^{-3}\int\big[\mfr12
  p^2-\frac{2}{\pi|v|}+1\big]_-\,dpdv\,.
$$
Since the potential $V(x)=\frac{2}{\pi|x|}-1$ satisfies the 
assumptions of Theorem~\ref{thm:TF-semicl} we see that there exists a
constant \(C>0\) such that
\beax
  \lefteqn{\Big|\Tr\big[\sqrt{-\beta^{-1}h^2\Delta
  +\beta^{-2}}-\beta^{-1}-\frac{2}{\pi|{\hat x}|}
  +1\big]_-}
  \\&-& (2\pi h)^{-3}\int\big[\mfr12
  p^2-\frac{2}{\pi|v|}+1\big]_-\,dpdv 
  -h^{-2}\frac{4}{\pi^2}\mathcal{S}(\alpha)\Big|
  \le C\,h^{-2+1/10}\,.
\eeax
Using that $h=2\kappa^{1/2}/\pi$ gives \eqref{eq:est-S}.
\end{proof}
We can now, using the alternative characterization of the function
\(\mathcal{S}\) in Lemma~\ref{cor:alt-S}, finish the proof of
Theorem~\ref{thm:TF-semicl}.

\noindent{\bf Step 4: Continuity of the function \(\mathcal{S}\).}

We recall that
\begin{align}
  T_{\beta}(p)=
  \begin{cases}
    \ \sqrt{\beta^{-1}p^2+\beta^{-2}}-\beta^{-1} \ ,& \beta>0 \\
    \qquad\  \frac12p^2 \ ,\qquad\quad\, & \beta=0
  \end{cases}\,.
\end{align}
It suffices to prove continuity of
\[
  \Tr\big[T_{\alpha^2}(-{\rm i}\nabla)-|{\hat x}|^{-1}+\kappa\big]_{-}
  =
  \Tr\big[\sqrt{-\alpha^{-2}\D+\alpha^{-4}}-\alpha^{-2}-|{\hat
  x}|^{-1}+\kappa\big]_{-}
\] 
at all \(\alpha_0\in[0,2/\pi]\), for any \(\kappa\in(0,1]\)
fixed. Then continuity of \(\mathcal{S}\) follows from
\eqref{eq:est-S} by uniform
convergence as \(\kappa\to0\).

We first prove the continuity at \(\alpha_0=0\).

Let \(\chi_{>}=\chi_{|p|\ge\lambda}\),
\(\chi_{<}=\chi_{|p|\le\lambda}\) for some \(\lambda>0\) to be chosen
below. Note that \((\Gamma_1-\Gamma_2)(\Gamma_1-\Gamma_2)^*\ge0\)
implies
\(\Gamma_1\Gamma_2^*+\Gamma_2\Gamma_1^*\le\Gamma_1\Gamma_1^*+\Gamma_2\Gamma_2^*\). 
Using this with \(\Gamma_1=\varepsilon^{1/2}\chi_{<}|{\hat x}|^{-1/2}\),
\(\Gamma_2=\varepsilon^{-1/2}\chi_{>}|{\hat x}|^{-1/2}\) for some
\(\varepsilon>0\) which we choose later, we get the operator inequality
\bea\label{eq:anotherC-S}
  \lefteqn{T_{\alpha^{2}}(\hat{p})-|{\hat x}|^{-1}+\kappa}
  \\&\ge&
  \chi_{>}\big(T_{\alpha^{2}}(\hat{p})-(1+\varepsilon^{-1})|{\hat
    x}|^{-1}+\kappa\big)\chi_{>} 
  +\chi_{<}\big(T_{\alpha^{2}}(\hat{p})-(1+\varepsilon)|{\hat
    x}|^{-1}+\kappa\big)\chi_{<}\,.
  \nonumber
\eea

Here and in the sequel we write \(T_{\alpha^{2}}(\hat{p})\) for the operator
\(T_{\alpha^{2}}(-{\rm i}\nabla)\) (and similarly for other operators).
Since \(T_{\alpha^{2}_1}\ge T_{\alpha^{2}_2}\) for \(\alpha_1\le \alpha_2\), and
\(T_{\alpha^{2}}(p)\ge\alpha^{-1}|p|-\alpha^{-2}\), \eqref{eq:anotherC-S}
implies that, if \(\alpha\in(0,A]\) for some \(A>0\), then for all
$\varepsilon>0$ 
\begin{eqnarray}
  T_{\alpha^{2}}(\hat{p})-|{\hat x}|^{-1}+\kappa&\ge&
  \chi_{>}\big(A^{-1}|\hat{p}|-A^{-2}
  -(1+\varepsilon^{-1})|{\hat x}|^{-1}+\kappa\big)\chi_{>} 
  \nonumber\\&&{}+\chi_{<}\big(T_{\alpha^{2}}(\hat{p})-(1+\varepsilon)|{\hat
    x}|^{-1}+\kappa\big)\chi_{<}\,. \label{eq:momentumsep}
\end{eqnarray}
Since \(|\hat{p}|-2/(\pi|{\hat x}|)\ge0\), we have that
\beax
  \frac12A^{-1}|\hat{p}|-(1+\varepsilon^{-1})|{\hat x}|^{-1}\ge0\,,
\eeax
if \(A\le \varepsilon/(2\pi)\), and now assuming 
$\varepsilon\leq1$. Furthermore,  for
\(\lambda\ge2A^{-1}\) we have that 
\beax
  \chi_{>}\big(\frac12A^{-1}|\hat{p}|-A^{-2}\big)\chi_{>}\ge0\,.
\eeax
This implies that, if \(\varepsilon\leq 1\), \(\lambda\ge2A^{-1}\),
\(\alpha\in(0,A]\) and \(A\le \varepsilon/(2\pi)\), then
by (\ref{eq:momentumsep})
\begin{eqnarray}\label{eq:lowmomentum}
  T_{\alpha^{2}}(\hat{p})-|{\hat x}|^{-1}+\kappa&\ge&
  \chi_{<}\big(T_{\alpha^{2}}(\hat{p})-(1+\varepsilon)|{\hat
    x}|^{-1}+\kappa\big)\chi_{<}\,. 
\end{eqnarray}

Since, by a Taylor-expansion, \(T_{\alpha^{2}}(p)\ge T_{0}(p)-(\alpha
p^2)^2/8\), and since \(\chi_{<}=\chi_{|p|\le\lambda}\), we have that,
still for \(\alpha\in(0,A]\), 
\bea\label{eq:lowMomentaGone}
  T_{\alpha^{2}}(\hat{p})-|{\hat x}|^{-1}+\kappa\ge
  \chi_{<}\big(T_{0}(\hat{p})-\alpha^2\lambda^4/8-(1+\varepsilon)|{\hat x}|^{-1} 
  +\kappa\big)\chi_{<}\,. 
\eea
Let
\beax
  \gamma_{\alpha,\kappa}=\chi_{(-\infty,0]}
  \big(T_{\alpha^{2}}(\hat{p})-|{\hat x}|^{-1}+\kappa\big)\,.  
\eeax
Then \eqref{eq:lowMomentaGone} and the fact that \(T_0\ge
T_{\alpha^{2}}\) imply that,
for \(\alpha\in(0,A]\), 
\bea\label{eq:firstBigIneq}\nonumber
  \Tr\big[T_0(\hat{p})-|{\hat x}|^{-1}+\kappa\big]_{-}
  &\ge&\Tr\big[T_{\alpha^{2}}(\hat{p})-|{\hat x}|^{-1}+\kappa\big]_{-}
  =\Tr\big[\gamma_{\alpha,\kappa}\big(T_{\alpha^{2}}(\hat{p})-|{\hat
    x}|^{-1}+\kappa\big)\big] 
  \\&\geq&\Tr\big[\gamma_{\alpha,\kappa}\chi_{<}\big(T_{0}(\hat{p})-\alpha^2\lambda^4/8 
  -(1+\varepsilon)|{\hat x}|^{-1}+\kappa\big)\chi_{<}\big]\,.
\eea 

If \(\kappa\in(0,1]\),
\(\alpha\in(0,A]\), $\lambda\geq 2A^{-1}$, and $A\leq 1/(2\pi)$ we
will show the a priori estimate
\bea\label{eq:boundedGammas}
  \Tr\big[\gamma_{\alpha,\kappa}\chi_{<}\big]\le C\kappa^{-3/2} \quad
  \text{ and } \quad \Tr\big[\gamma_{\alpha,\kappa}\chi_{<}\,|{\hat
  x}|^{-1}\chi_{<}\big]\le C\kappa^{-1/2}\,.  
\eea 
The combined Daubechies-Lieb-Yau inequality \eqref{ineq:ImprovedDLY}
gives that for positive constants \(C_1, C_2\) such that \(\alpha\le
2/(C_1\pi)\), we have 
\beax 
  \Tr\big[T_{\alpha^{2}}(\hat{p})-C_1|{\hat
  x}|^{-1}+C_2\kappa\big]_{-} &\ge& {}-C\alpha^{1/2}-C\int_{|x|<
  C\kappa^{-1}}\big(|x|^{-1}+\kappa\big)^{5/2}\,dx
  \\&&{}-C\alpha^{3}\int_{\alpha<|x|<
  C\kappa^{-1}}\big(|x|^{-1}+\kappa\big)^4\,dx \ge{}-C\kappa^{-1/2}\,.
\eeax 
If $\alpha\in (0,A]$ and $A\leq 1/(2\pi)$ then \(\alpha\leq 4/(5\pi)\)
and hence we obtain from
\eqref{eq:lowmomentum} with $\varepsilon=1$ that
\beax 
  0&\ge&
  \Tr\big[T_{\alpha^{2}}(\hat{p})-|{\hat x}|^{-1}+\kappa\big]_{-} 
  =\Tr\big[\gamma_{\alpha,\kappa}\big(T_{\alpha^{2}}(\hat{p})-|{\hat 
  x}|^{-1}+\kappa\big)\big]\\
  &\geq&\Tr\big[\gamma_{\alpha,\kappa}\chi_<\big(T_{\alpha^{2}}(\hat{p})-2|{\hat 
  x}|^{-1}+\kappa\big)\chi_<\big]\\
  &=&\Tr\big[\chi_<\gamma_{\alpha,\kappa}\chi_<\big(T_{\alpha^{2}}(\hat{p})
  -\frac{5/2}{|{\hat x}|}+\frac12\kappa\big)\big]
  +\frac12\Tr\big[\gamma_{\alpha,\kappa}\chi_<|{\hat
  x}|^{-1}\chi_<\big]+\frac{\kappa}{2}\Tr\big[\gamma_{\alpha,\kappa}\chi_<\big] 
  \\&\ge&{}-C\kappa^{-1/2}
  +\frac12\Tr\big[\gamma_{\alpha,\kappa}\chi_<|{\hat 
  x}|^{-1}\chi_<\big]+\frac{\kappa}{2}\Tr\big[\gamma_{\alpha,\kappa}\chi_<\big]. 
\eeax 
This gives \eqref{eq:boundedGammas}.

Choose \(\lambda=2A^{-1}\), \(A=\epsilon/(2\pi)\). We combine
\eqref{eq:firstBigIneq} and \eqref{eq:boundedGammas} and use the
variational principle to conclude that for
\(\alpha\in(0,\epsilon/(2\pi)]\), \(\epsilon<1\), and \(\kappa\in(0,1]\), 
\beax
  \Tr\big[T_0(\hat{p})-|{\hat x}|^{-1}+\kappa\big]_{-}
  &\ge&
  \Tr\big[T_{\alpha^{2}}(\hat{p})-|{\hat x}|^{-1}+\kappa\big]_{-}
  \\&\ge&\Tr\big[\big(\chi_{<}\gamma_{\alpha,\kappa}\chi_{<}\big)
      \big(T_{0}(\hat{p})-|{\hat x}|^{-1}+\kappa\big)\big]
  -C\kappa^{-3/2}(\alpha^2\varepsilon^{-4}+\varepsilon)
  \\&\ge& \Tr\big[T_{0}(\hat{p})-|{\hat x}|^{-1}+\kappa\big]_{-}
  -C\kappa^{-3/2}(\alpha^2\varepsilon^{-4}+\varepsilon)\,.
\eeax
Finally choose \(\epsilon=\alpha^{2/5}\); then \(\alpha\le
(2\pi)^{-5/3}\) implies that \(\alpha\in(0,\epsilon/(2\pi)]\) and
\(\epsilon<1\).

Therefore we have proved that for
any \(\alpha\le (2\pi)^{-5/3}\) and \(\kappa\in(0,1]\), 
\beax
  \Tr\big[T_0(\hat{p})-|{\hat x}|^{-1}+\kappa\big]_{-}
  &\ge&
  \Tr\big[T_{\alpha^{2}}(\hat{p})-|{\hat x}|^{-1}+\kappa\big]_{-}
  \\&\ge& \Tr\big[T_{0}(\hat{p})-|{\hat x}|^{-1}+\kappa\big]_{-}
  -C\kappa^{-3/2}\alpha^{2/5}\,,
\eeax
which proves continuity from the right of
\(\Tr\big[T_{\alpha^{2}}(\hat{p})-|{\hat x}|^{-1}+\kappa\big]_{-}\) at
\(\alpha_0=0\) for any \(\kappa\in(0,1]\) fixed. (Notice that the
above has not been optimized in \(\kappa\).)

We now prove the continuity at any \(\alpha_0\in(0,2/\pi)\). Note
first that, for \(0<\alpha_1\le\alpha_2\),
\bea\label{eq:ineqKinEnergy} 
  T_{\alpha^{2}_1}(p)\ge T_{\alpha^{2}_2}(p)\ge
  (\alpha_2^{-1}\alpha_1)^2\, T_{\alpha^{2}_1}(p)\,. 
\eea
Assume first that \(\alpha>\alpha_0\), and let \(\gamma_{\alpha,\kappa}\)
be defined as above. Then, using
\eqref{eq:ineqKinEnergy} and the variational principle,
\beax
  \Tr[T_{\alpha^{2}}(\hat{p})-|{\hat x}|^{-1}+\kappa\big]_{-} 
  &\le&\Tr\big[T_{\alpha^{2}_0}(\hat{p})-|{\hat
    x}|^{-1}+\kappa\big]_{-} 
  \le\Tr\big[\gamma_{\alpha,\kappa}\big(T_{\alpha^{2}_0}(\hat{p})-|{\hat x}|^{-1}+\kappa\big)\big] 
  \\&\le&\Tr\big[\gamma_{\alpha,\kappa}\big(T_{\alpha^{2}}(\hat{p})-|{\hat x}|^{-1}+\kappa\big)\big] 
  +[(\alpha\alpha_0^{-1})^2-1]\,\Tr\big[\gamma_{\alpha,\kappa}T_{\alpha^{2}}(\hat{p})\big]
  \\&=&  \Tr[T_{\alpha^{2}}(\hat{p})-|{\hat x}|^{-1}+\kappa\big]_{-}
   +[(\alpha\alpha_0^{-1})^2-1]\,\Tr\big[\gamma_{\alpha,\kappa}T_{\alpha^{2}}(\hat{p})\big]\,.
\eeax
It remains to show that 
\([(\alpha\alpha_0^{-1})^2-1]\,\Tr\big[\gamma_{\alpha,\kappa}T_{\alpha^{2}}(\hat{p})\big]\to0\)
as \(\alpha\to\alpha_0\). 
For this, it obviously suffices to show that
\(\Tr\big[\gamma_{\alpha,\kappa}T_{\alpha^{2}}(\hat{p})\big]\) is
uniformly bounded for, say, \(\alpha\in(\alpha_0,A]\) for some
\(A\in(\alpha_0,2/\pi)\). But this follows as in the proof of
\eqref{eq:boundedGammas}. This proves continuity from the right of 
\(\Tr\big[T_{\alpha^{2}}(\hat{p})-|{\hat x}|^{-1}+\kappa\big]_{-}\) at 
\(\alpha_0\in(0,2/\pi)\). To prove continuity from the
left, assume \(\alpha<\alpha_0\), and let \(\gamma_{\alpha_0,\kappa}\) be
defined as above. Then, by 
\eqref{eq:ineqKinEnergy} and the variational principle,
\beax
  \Tr\big[T_{\alpha^{2}}(\hat{p})-|{\hat x}|^{-1}+\kappa\big]_{-}
  &\ge&
  \Tr\big[T_{\alpha^{2}_0}(\hat{p})-|{\hat x}|^{-1}+\kappa\big]_{-}
  =\Tr\big[\gamma_{\alpha_0,\kappa}\big(T_{\alpha^{2}_0}(\hat{p})-|{\hat x}|^{-1}+\kappa\big)\big] 
  \\&=&\Tr\big[\gamma_{\alpha_0,\kappa}\big(T_{\alpha^{2}}(\hat{p})-|{\hat x}|^{-1}+\kappa\big)\big]
  +\Tr\big[\gamma_{\alpha_0,\kappa}(T_{\alpha^{2}_0}(\hat{p})-T_{\alpha^{2}}(\hat{p}))\big]
  \\&\ge& \Tr\big[T_{\alpha^{2}}(\hat{p})-|{\hat x}|^{-1}+\kappa\big]_{-}
  +[1-(\alpha_0\alpha^{-1})^2]\,\Tr\big[\gamma_{\alpha_0,\kappa}T_{\alpha^{2}_0}(\hat{p})\big]\,. 
\eeax
As before, the last trace is finite by arguments as in the proof of
\eqref{eq:boundedGammas} (since \(\alpha_0<2/\pi\)). This proves
continuity from the left, and 
therefore, continuity, of
\(\Tr\big[T_{\alpha^{2}}(\hat{p})-|{\hat x}|^{-1}+\kappa\big]_{-}\) at 
\(\alpha_0\in(0,2/\pi)\).

Finally we prove the continuity at \(\alpha_0=2/\pi\). Here, arguments
as in the proof of \eqref{eq:boundedGammas} are no longer at our
disposal. Therefore, let \(\epsilon>0\), and let
\(\gamma_{\alpha_0,\kappa}\) be defined as above, and choose
\(\phi_1,\ldots,\phi_N\in C_{0}^{\infty}(\mathbb R^3)\),
\((\phi_i,\phi_j)=\delta_{i,j}\), such that
\bea\label{ineq:firstCrit}
  \lefteqn{\Tr\big[\gamma_{N}\big(T_{\alpha^{2}_0}(\hat{p})-|{\hat
      x}|^{-1}+\kappa\big)\big]} 
  \\&\le&
  \Tr\big[\gamma_{\alpha_0,\kappa}\big(
  T_{\alpha^{2}_0}(\hat{p})-|{\hat x}|^{-1}+\kappa\big)\big] 
  +\epsilon/2
  =\Tr\big[T_{\alpha^{2}_0}(\hat{p})-|{\hat x}|^{-1}+\kappa\big]_{-}  +
  \epsilon/2\,,
  \nonumber
\eea
for \(\gamma_N(x,y)=\sum_{j=1}^N\phi_j(x)\overline{\phi_j(y)}\). 
This is possible since the operator is defined as the Friedrichs
extension from \(C_0^{\infty}(\R^3)\). 
(Here, both \(N\) and the \(\phi_j\)'s depend, of course, on
\(\epsilon\)). Recall that \(\gamma_N\) is finite dimensional and
\(\phi_j\in C_{0}^{\infty}(\mathbb{R}^3)\). 
Using this,  
\eqref{eq:ineqKinEnergy}, and the variational principle gives that
(for any \(\alpha\in(\alpha_0/2,\alpha_0)\)),
\bea\label{ineq:bisCrit}  \nonumber
  \lefteqn{\Tr\big[\gamma_{N}\big(T_{\alpha^{2}_0}(\hat{p})-|{\hat
      x}|^{-1}+\kappa\big)\big]
  =\Tr\big[\gamma_N
  T_{\alpha^{2}_0}(\hat{p})\big]+\Tr\big[\gamma_N\big(-|{\hat x}|^{-1}+\kappa\big)\big]}
  \\&\ge& \Tr\big[\gamma_N\big(T_{\alpha^{2}}(\hat{p})-|{\hat x}|^{-1}+\kappa\big)\big]
  +[(\alpha_0^{-1}\alpha)^2-1]\,\Tr\big[\gamma_NT_{\alpha^{2}}(\hat{p})\big]  \nonumber
  \\&\ge& \Tr\big[T_{\alpha^{2}}(\hat{p})-|{\hat x}|^{-1}+\kappa\big]_{-}+
  [(\alpha_0^{-1}\alpha)^2-1]\,\Tr\big[\gamma_NT_{\alpha^{2}_0/2}(\hat{p})\big]\,.
\eea
Choose now \(\delta>0\) such that
\bea\label{ineq:secondCrit}
  \alpha\in(\alpha_0-\delta,\alpha_0)\cap(\alpha_0/2,\alpha_0)
  \quad\Rightarrow\quad
  [(\alpha_0^{-1}\alpha)^2-1]\,\Tr\big[\gamma_NT_{\alpha^{2}_0/2}(\hat{p})\big]>{}-\epsilon/2\,.
\eea
Then, combining \eqref{ineq:firstCrit}, \eqref{ineq:bisCrit}, and \eqref{ineq:secondCrit}
(and using \eqref{eq:ineqKinEnergy} again) we
have proved that, for all \(\epsilon>0\) there exists \(0<\delta<\alpha_0/2\)
such that
\beax
  \lefteqn{\alpha\in(\alpha_0-\delta,\alpha_0)}
  \\&\Rightarrow& 
  \Tr\big[T_{\alpha^{2}}(\hat{p})-|{\hat x}|^{-1}+\kappa\big]_{-}
  \ge  \Tr\big[T_{\alpha^{2}_0}(\hat{p})-|{\hat x}|^{-1}+\kappa\big]_{-}
  \ge  \Tr\big[T_{\alpha^{2}}(\hat{p})-|{\hat x}|^{-1}+\kappa\big]_{-} - \epsilon\,.
\eeax
This proves the continuity from the left of
\(\Tr\big[T_{\alpha^{2}}(\hat{p})-|{\hat x}|^{-1}+\kappa\big]_{-}\) at
\(\alpha_0=2/\pi\), and therefore finishes the proof that
\(\mathcal{S}:[0,2/\pi]\to\R\) is continuous.

This completes the proof of Theorem~\ref{thm:TF-semicl}. 
\end{proof}

\section{Local relativistic semi-classical estimates using new coherent states}
\label{Local semi-classical estimates}

In this section we study the sum and the density of the negative eigenvalues of 
the localised Hamiltonian $\phi H_\b\phi$, with $\phi$ compactly supported and 
$H_\b=T_\b(-\ic h\nabla)+V(\hat{x})$. Here, \(T_\b\) is given by
\eqref{def:kinetic energy}, and \(V\) is a (sufficiently) regular
potential (see below for details). For the most part we suppress the index 
$\b$ but all estimates, in particular the constants C, will be uniform
in $\b\in[0,1]$.

We first recall the definition and the main properties of the coherent
states (operators) introduced in~\cite{SS}, where all proofs can be
found. These coherent states are denoted by $\cG_{u,q}$. Let
$1/a>h>0$. The kernel of $\cG_{u,q}$ is given by
\begin{equation}\label{new coherent states} 
  {\cG}_{u,q}(x,y) = (\pi h)^{-n/2} 
  {\rm e}^{-a\left(\frac{x+y}{2}-u\right)^2 
  +\ic q(x-y)/h -\frac{1}{4h^2a}(x-y)^2} \,.
\end{equation}
A first important property of these operators is their completeness.
\begin{lemma}[{\bf Completeness of new coherent states}]
The coherent operators $\cG_{u,q}$ satisfy 
\bea\label{eq:formula-coherent}
  \int {\cG}_{u,q}^2 \,\frac{dq}{(2\pi h)^n} = G_b(\x-u)\,,\quad
  \int {\cG}_{u,q}^2 \,\frac{du}{(2\pi h)^n} = G_b(-\ic h\nabla-q)\,,
\eea
where $\x$ denotes the operator multiplication by the position
variable $x$. Here $G_b(v)=(b/\pi)^{n/2}{\rm e}^{-bv^2}$ with
$b=2a/(1+h^2a^2)$. Note that $G_b$ has integral 1 and hence
\bea\label{eq:complete}
  \int {\cG}_{u,q}^2 \,\frac{dudq}{(2\pi h)^n} = {\bf 1} \,.
\eea
\end{lemma}

We shall consider operators of the form 
\begin{equation}\label{eq:formf}
  \int {\cG}_{u,q}\,f({\widehat A}_{u,q})\,{\cG}_{u,q}\,dudq\,, 
\end{equation}
where $f:\R\to\R$ is any polynomially bounded real function. As we
shall see in the next theorem the integrand above is a trace class
operator for each $(u,q)$. The integral above is to be understood in
the weak sense, i.e., as a quadratic form. We shall consider
situations where the integral defines bounded or unbounded operators.  
\begin{thm}[{\bf Trace identity}]\label{thm:traceidentity} 
Let $f:\R\to\R$ and $V:\R^n\to\R$  be polynomially bounded,
real-valued measurable functions and let
$$\hat{A}=B_0 + B_1{\x} - \ic h B_2\nabla$$ 
be a first order self-adjoint differential operator\footnote{The
operator $\hat{A}$ is essentially self-adjoint on Schwartz functions on $\R^n$.}
with $B_0\in\R$, $B_{1,2}\in\R^n$. Then  
${\cG}_{u,q}\,f(\hat{A}) \,{\cG}_{u,q}\,V({\x})$ is a
trace class operator (when extended from $C_0^\infty(\R^n)$) and
\begin{eqnarray*}
  {\Tr}\big[{\cG}_{u,q}\,f(\hat{A}) \,{\cG}_{u,q}\,V({\x})\big]
  &=& \int f(B_0+B_1v+B_2p)\,G_b(u-v)G_b(q-p)
  G_{(h^2b)^{-1}}(z)
  \\
  &&\qquad\qquad\qquad\qquad\qquad
  \times\,V(v+h^2ab(u-v)+z)\,dvdpdz \,.
\end{eqnarray*}
In particular, $\Tr\big[{\cG}_{u,q}^2\big]=1$.
\end{thm}
We shall also need the following extension of this theorem, where we
however only give an estimate on the trace. 
\begin{thm}[{\bf Trace estimates}] \label{trace formula} Let
  $f,\hat{A}$ be as in the previous theorem. Let moreover $\phi\in
  C^{n+4}(\mathbb R^n)$ be a bounded, real function with all
  derivatives up to order $n+4$ bounded, and let $V,F\in C^2(\R^n)$ be
  real functions with bounded second derivatives. Then, for $h<1$,
  $1<a< 1/h$ and $b=2a/(1+h^2a^2)$ we have, with $\s(u,q)=F(q)+V(u)$,
  that\footnote{The operator ${\cG}_{u,q}\,f(\hat{A}) \,{\cG}_{u,q}\,
    \phi(\hat{x}) \left(F(-\ic h\nabla) + V({\x})\right)\phi(\hat{x})$
    is originally defined on, say, $C_0^\infty(\R^n)$, but it is part
    of the claim of the theorem that it extends to a trace class
    operator on all of $L^2(\R^n)$.}
\begin{eqnarray*} 
  \lefteqn{{\Tr}\big[{\cG}_{u,q}\,f(\hat{A}) 
  \,{\cG}_{u,q}\,
  \phi(\x) \big(F(-\ic h\nabla) + V({\x})\big)\phi(\x) \big]
  }
  \\&=&
  \int f(B_0+B_1v+B_2p)\,G_b(u-v)G_b(q-p)
  \\
  &&\quad\times\,\Big[\big(\phi(v+h^2ab(u-v))^2 + E_1(u,v)\big)
  \s(v+h^2ab(u-v),p+h^2ab(q-p))\\&&{}
  \qquad\qquad\qquad\qquad\qquad\qquad\qquad
  \qquad\qquad\qquad\qquad\qquad\quad\ 
  + E_2(u,v;q,p)\Big]\, dv dp \,,
\end{eqnarray*}
with $\|E_1\|_\infty,\|E_2\|_\infty\le C h^2b$, where $C$ depends only
on 
$$ 
  \sup_{|\nu|\le n+4} \|\p^\nu \phi\|_{\infty}\,,\ 
  \sup_{|\nu|=2}\|\p^\nu
  V\|_\infty\,,\ \hbox{and}\ \sup_{|\nu|=2}\|\p^\nu
  F\|_\infty\,.
$$
(Note that the assumption $1<a<1/h$ implies $1<b<1/h$.) 
\end{thm} 
We will use the above theorem to prove an upper bound on the sum of
eigenvalues of the operator $F(-\ic h\nabla) + V({\x})$, in the case
when $F(q)=T_\beta(q)$ from \eqref{def:kinetic energy} with
$\b\in[0,1]$ (equal to $\sqrt{\b^{-1}q^2+\b^{-2}}-\b^{-1}$ for
$\b\in(0,1]$, and to \(\frac12q^2\) when \(\beta=0\)).
This is done in Lemma~\ref{lm:upperbound} below by
constructing a trial density matrix on the form (\ref{eq:formf}).  

To prove a lower bound on the sum of the negative eigenvalues one
approximates the Hamiltonian $F(-\ic h\nabla) + V(\x)$ by an operator
also represented on the form (\ref{eq:formf}).

\begin{thm}[{\bf Coherent states
    representation}]\label{thm:coherentrepresentation} 
Consider functions $F,V\in C^3(\R^n)$, for which all second and third
derivatives are bounded. Let $\s(u,q)=F(q)+V(u)$, then for
$a<1/h$ and $b=2a/(1+h^2a^2)$ we have  the representation (as
quadratic forms on $C^\infty_0(\R^n)$), 
$$
  F(-\ic h\nabla) + V(\x) = \int \,{\cG}_{u,q}\widehat H_{u,q}
  {\cG}_{u,q}\frac{du dq}{(2\pi h)^n} + {\mathbf E}\,,
$$
with the operator-valued symbol
\begin{equation}\label{eq:coherentrepresentation}
  \widehat H_{u,q}= \s(u,q)+\mfr{1}{4b}\D \s(u,q) 
  + \p_u \s(u,q)({\x}-u) + \p_q \s(u,q)(-ih\nabla -q)\,.
\end{equation}
The error term, ${\mathbf E}$, is a bounded operator with
$$
  \|\mathbf E\|\leq Cb^{-3/2}\sum_{|\nu|=3}\|\partial^\nu
  \s\|_\infty+Ch^2b\sum_{|\nu|=2}\|\partial^\nu\s\|_\infty \,.
$$
\end{thm}

Let us recall our convention that $x_-=\min\{x,0\}$ and that
$\chi=\chi_{(-\infty,0]}$ denotes the characteristic function of
$(-\infty,0]$.

The next theorem is the main result of this section.
\begin{thm}[{\bf Local relativistic semi-classics}] \label{local semi-classics}
For $n\ge3$, let $\phi\in C_0^{n+4}(\mathbb R^n)$ be supported in
a ball $B_1\subset\R^n$ of radius $1$ and let $V\in
C^3(\overline{B_1})$ be a real function. Let $0\le\beta\le1$, $h>0$, and
let $\sigma_\beta(u,q)=T_\beta(q)+V(u)$ and \(H_\b=T_\beta(-\ic
h\nabla) + V(\hat{x})\) 
with
$T_\beta(q)=\sqrt{\beta^{-1}q^2+\beta^{-2}}-\beta^{-1}$ for
\(\beta\in(0,1]\) and \(T_0(q)=\frac12q^2\).

Then 
\bea\label{local semicl}
  \Big|{\Tr}\big[\phi H_\b
  \phi\big]_{-}-(2\pi h)^{-n}\int \,\phi(u)^2 
  \sigma_\beta(u,q)_-\,du dq\Big|
  \leq Ch^{-n+6/5} \,. \nonumber
\eea
The constant $C>0$ here depends only on $\|\phi\|_{C^{n+4}}$,
$\|V\|_{C^3}$,\footnote{We use the convention that
  $\|\psi\|_{C^p}=\sup_{|\nu|\leq p} \|\partial^\nu\psi\|_\infty$.}
and the dimension $n$, but not on $\beta\in[0,1]$.
\end{thm}
The important property for our method to work is that the second and
third order derivatives of the kinetic energy function $T_\beta(q)$
are bounded uniformly in $q$ and $\beta$. Thus the error term above is 
independent of $\beta\in[0,1]$, and in particular the same as for the 
non-relativistic case, \({}-h^2\D/2+V\), which corresponds to the limit
$\b\to0$. We 
prove upper and lower bounds and start with the lower bound.

\begin{lemma}[{\bf Lower bound on ${\Tr}[\phi
    H_\b\phi]_-$}]\label{local semi-classics-lower} 
Under the same conditions as in Theorem \ref{local semi-classics},
$$ 
  {\Tr}[\phi H_\b\phi]_- 
  \ge (2\pi h)^{-n}\int \,\phi(u)^2 \s_\b(u,q)_- \, du dq
  -Ch^{-n+6/5} \,.
$$
The constant $C>0$ here depends only on $\|\phi\|_{C^{n+4}}$,
$\|V\|_{C^3}$, and the dimension $n$, but not on $\b\in[0,1]$.
\end{lemma}
\begin{proof} 
Since $\phi$ has support in the ball $B_1$ we may assume
without loss of generality that $V\in C^3_0(\R^n)$ with the
support in a ball $B_2$ of radius $2$
and that the norm $\| V\|_{C^3}$ refers to the supremum over all 
of $\R^n$. We shall not explicitly follow how the 
error terms depend on $\|\phi\|_{C^{n+4}}$ and $\|V\|_{C^3}$. 
All constants denoted by $C$ depend on $\|\phi\|_{C^{n+4}}$,
$\|V\|_{C^3}$, and the dimension $n$ but, in particular, not on $\b$.

We use the Daubechies inequality (Theorem~\ref{Daubechies}) to control
various error estimates. Since $T_\b(q)\ge T_1(q)$ for $\beta\in[0,1]$
we may use it with $\beta=1$. Then, uniformly in $\beta\in[0,1]$,
$$ 
  {\Tr}[\phi H_\b\phi]_- \ge C\|\phi\|_\infty^2\int\limits_{u\in B_1}
  \, \s_1(u,q)_- 
  \,\frac{du dq}{(2\pi h)^n} \geq {}-C h^{-n}\,.
$$
Consider some fixed $0<\tau<1$ (independent of $h$ and $\beta$).
If $h\geq \tau$ then we get that
$$ 
  {\Tr}[\phi H_\b\phi]_- \ge \int \, \phi(u)^2 \s_\b(u,q)_-
  \,\frac{du dq}{(2\pi h)^n} - C\tau^{-6/5}h^{-n+6/5}\,.
$$
We are therefore left with considering $h<\tau$. 

If we now use the inequality $[x+y]_-\ge [x]_- + [y]_-$, which we will
do frequently without further mentioning, and Theorem
\ref{thm:coherentrepresentation}, we have that  
\begin{eqnarray}
  {\Tr}[\phi H_\b\phi]_-&\geq& {\Tr}\left[\int \phi 
  \,{\cG}_{u,q}\widehat{H}_{u,q}^{(\varepsilon)}
  {\cG}_{u,q}\phi\,\frac{du dq}{(2\pi h)^n}\right]_-\nonumber \\
  &&{}+ {\Tr}\big[\phi\big(\varepsilon \sqrt{-\beta^{-1}h^2\Delta +
       \beta^{-2}} - 
  \varepsilon \beta^{-1} - C(b^{-3/2}+h^2b)\big)
  \phi\big]_- \,. \label{eq:LTerror}
\end{eqnarray}
Here, $0<\varepsilon<1/2$ and
$$
  \widehat H_{u,q}^{(\varepsilon)}=\widetilde{\s}(u,q)+\mfr{1}{4b}\D
  \widetilde{\s}(u,q)+\p_u\widetilde{\s}(u,q)({\x}-u) + \p_q
  \widetilde{\s}(u,q)(-ih\nabla -q)
$$
with $\widetilde{\s}(u,q)=(1-\varepsilon)T_\b(q)+V(u)$. The second
trace can be estimated from below using the Daubechies inequality
(Theorem~\ref{Daubechies}) with
\(\alpha=\beta^{1/2}h^{-1}\), \(m=h^{-2}\). Then
\begin{eqnarray}\lefteqn{
  {\Tr}\big[\phi\big(\varepsilon \sqrt{-\beta^{-1}h^2\Delta +
        \beta^{-2}} - 
  \varepsilon \beta^{-1} - C(b^{-3/2}+h^2b)\big)\phi\big]_-
  }\nonumber\\
  &=&\varepsilon\,{\Tr}\big[\phi\big(\sqrt{-\beta^{-1}h^2\Delta +
        \beta^{-2}} - 
  \beta^{-1} - C\varepsilon^{-1}(b^{-3/2}+h^2b)\big)\phi\big]_-
  \nonumber
  \\\nonumber
  &\ge&{}-C\varepsilon h^{-n}\int_{B_1}
  \big(\varepsilon^{-1}(b^{-3/2}+h^2b)\big)^{1+n/2}\,dx
  \\&&{}-C\varepsilon\beta^{n/2}h^{-n}\int_{B_1}
  \big(\varepsilon^{-1}(b^{-3/2}+h^2b)\big)^{n}\,dx\,.
  \label{Derror}
\end{eqnarray}
We shall eventually choose
$\varepsilon=\frac{1}{4}(b^{-3/2}+h^2b)$. Note that then
$\varepsilon<1/2$, and that the bound in (\ref{Derror}) is
${}-Ch^{-n}(b^{-3/2}+h^2b)$, uniformly for \(\beta\in[0,1]\). 

By bringing the negative part inside in \eqref{eq:LTerror} we obtain
the lower bound, 
$$
  {\Tr}[\phi H_\b\phi]_-\geq \int {\Tr}\Big[\phi 
    \,{\cG}_{u,q}\big[\widehat H_{u,q}^{(\varepsilon)}\big]_-
    {\cG}_{u,q}\phi\Big]\,\frac{du dq}{(2\pi h)^n}  
  - C h^{-n}(b^{-3/2}+h^2b)\,.
$$

We first consider the integral over $u$ outside the ball $B_2$ of
radius \(2\), where $V=0$. Using Theorem \ref{thm:traceidentity} (with
\(f(t)=[t]_{-}\), and 
$V$ replaced by $\phi^2$) and $\int \phi^2\le C$, we get that this
part of the integral is 
\begin{eqnarray*}
  \lefteqn{
  (1-\varepsilon)\displaystyle\int\limits_{u\not\in 
  B_2}\left[T_{\b}(q)+ (n+(n-1)\beta q^2)/[4b(1+\beta q^2)^{3/2}]
  + q\cdot(p-q)/\sqrt{1+\beta q^2}\,\right]_- 
  }\\
  &&\qquad\quad\times\,G_b(q-p)G_b(u-v)
  G_{(h^2b)^{-1}}(z)\phi(v+h^2ab(u-v)+z)^2  
  \,dvdpdz\, 
  \frac{du dq}{(2\pi h)^n}
  \\
  &\geq&(1-\varepsilon) \int\limits_{z\in B_1}\phi(z)^2  
  \int\limits_{u\not\in B_2} G_b(u-v)
  G_{(h^2b)^{-1}}(v+h^2ab(u-v)-z)\, du dv dz
  \\
  &&\qquad\qquad\qquad\qquad\qquad\times  \int \big[T_{\b}(q) +
  q\cdot(p-q)/\sqrt{1+\beta  
    q^2}\,\big]_-   
  G_b(q-p) \,\frac{dq dp}{(2\pi h)^n}
  \\
  &=&(1-\varepsilon) \int\limits_{z\in B_1}\phi(z)^2
  \int\limits_{(1-h^2ab)u\not\in B_2-v} G_b(u) G_{(h^2b)^{-1}}(v-z)\,
  du dv dz
  \\
  &&\qquad\qquad\qquad\qquad\qquad\qquad\qquad\times  \int
  \big[T_{\b}(q) + q\cdot p/\sqrt{1+\beta 
    q^2}\,\big]_-   
  G_b(p) \,\frac{dq dp}{(2\pi h)^n}\,.
\eeax
The integration over $u,v$ is obviously bounded by 1. In fact, the
$u$-integration can be shown to be exponentially small, i.e., less
than $C \,{\rm e}^{-Cb}$, but this will not be necessary.

The domain of integration for the variables $q,p$ is contained in the
set $\{(q,p)\,|\, |q|\le 2|p|\}$. Then,
\beax 
  \lefteqn{\int \big[T_{\b}(q) + q\cdot p/\sqrt{1+\beta q^2}\,\big]_-  
  \,G_b(p) \,dq dp
  }
  \\&\ge& {}-C\int\limits_{|q|<2|p|} \frac{|q||p|}{\sqrt{1+\b q^2}}
  \,G_b(p)\,dqdp 
  \ge{}-C\int |p|^{n+2} G_b(p)\,dp \,= \,{}-C\,b^{-(n+2)/2}\,.
\eeax
It follows that the integral over $u\not\in B_2$ is bounded from below
by \({}-Ch^{-n}b^{-3/2}\), since \(b>1\).

For the integral over $u\in B_2$ we use
Theorem~\ref{thm:traceidentity} as before.
This time, expanding $\phi^2$ to second order in \(z\) at the point
$z=0$ and using the crucial fact (which we shall use without mentioning
later) that, for any \(\lambda>0\),
\bea\label{eq:Gauss-integrals-one}
  \int x_j \,G_\lambda(x) \,dx = 0\,,\quad\int
  |x|^{m}\,G_\lambda(x)\,dx=C\, \lambda^{-m/2}\,, 
\eea
implies that
\begin{eqnarray}
  {\Tr}[\phi H_\b\phi]_-\geq&\displaystyle\int\limits_{u\in B_2}&
  \!\!\!\!\!\!
  \big[\phi(v+h^2ab(u-v))^2 + C h^2 b\big]
  G_b(u-v)G_b(q-p) \nonumber
  \\
  &&\times\left[H_{u,q}^{(\varepsilon)}(v,p)\right]_-\,\frac{du
    dq}{(2\pi h)^n}dv dp  
    -Ch^{-n}(b^{-3/2}+h^2b)\label{eq:phiHphi}\,,
\end{eqnarray}
where
$$
  H_{u,q}^{(\varepsilon)}(v,p) =
  \widetilde{\s}(u,q)+\mfr{1}{4b}\Delta\widetilde{\s}(u,q) +
  \partial_u\widetilde{\s}(u,q)(v-u) 
  + \partial_q \widetilde{\s}(u,q)(p-q)\,.
$$
The rest of the proof is simply an estimate of the integral in
\eqref{eq:phiHphi}. This analysis is an elementary but tedious exercise
in calculus. For the convenience of the reader it is given in detail in
Appendix~\ref{app:B} below. 
\end{proof}
\begin{lemma}[{\bf Construction of a trial density
    matrix}]\label{lm:upperbound} Under the same conditions as in
  Theorem \ref{local semi-classics} there exists a density matrix
  $\gamma$ on $L^2(\R^n)$ such that 
\bea\label{eq:lemmaupper}
  {\Tr}\big[\phi (T_\b(-\ic h\nabla) + V(\hat{x}))\phi\gamma\big]\le 
 \int \, \phi(u)^2 \s_\b(u,q)_-\, \frac{du dq}{(2\pi h)^n} 
  + Ch^{-n+6/5} \,.
\eea
Moreover, the density $\rho_\c$ of $\gamma$ satisfies
\begin{equation}\label{eq:rhogammaprop1}
  \left|\rho_\gamma(x)-(2\pi h)^{-n}\omega_n \,
  |V_-|^{n/2}(2+\beta|V_-|)^{n/2}(x)\right|\leq Ch^{-n+9/10}\,,
\end{equation}
for (almost) all $x\in B_1$ and
\begin{equation}\label{eq:rhogammaprop2}
  \left|\int\phi(x)^2\rho_\gamma(x)\,dx-(2\pi h)^{-n}\omega_n\int\phi(x)^2
  \,|V_-|^{n/2}(2+\beta|V_-|)^{n/2}(x)\,dx\right|\leq Ch^{-n+6/5}\,,
\end{equation}
where $\omega_n$ is the volume of the unit ball $B_1$ in $\R^n$.
The constants $C>0$ in the above estimates depend only on
$n, \|\phi\|_{C^{n+4}}$, and $\|V\|_{C^3}$, but not on
$\beta\in[0,1]$. 
\end{lemma}
It is convenient to introduce the function
\be 
\label{eta}
    \eta(t)=n\int_0^{\infty} \chi[T_\b(p)+t]\,|p|^{n-1}\,d|p| 
    = |t_-|^{n/2}(2+\b|t_-|)^{n/2}\,.
\ee
(Recall that $\chi$ is the characteristic function of $\R_-$.)
\begin{proof} We will occasionally drop the index $\beta$ in $H_\b$ and
$\s_\b$. It is important to realize, however, that all estimates are
uniform in $\beta$. We first note that since $T_1(p)\leq T_\b(p)\leq
T_0(p)=p^2/2$ we have that 
\begin{eqnarray}\label{eq:uniformsigma}
  |p|-C\leq \s(v,p)\leq\frac12p^2+C.
\end{eqnarray}

Let us start by choosing some fixed $0<\tau<1$. For $h\ge
\tau$  and for some $C>0$ we have by (\ref{eq:uniformsigma}) that 
$$
  \int \phi(u)^2 \s(u,q)_-\,\frac{du dq}{(2\pi h)^n}
  + C\tau^{-6/5} h^{-n+6/5}\geq 0 \,,
$$
{and} that for any $s>0$,
$$ 
  (2\pi h)^{-n}\eta(V(x)) \leq C\tau^{-s}h^{-n+s}\,.
$$  
If $h\geq\tau$ we may therefore use $\gamma=0$, and $s=9/10$
and $s=6/5$ for \eqref{eq:rhogammaprop1} and \eqref{eq:rhogammaprop2},
respectively. From now on we assume that $h<\tau$ and, if necessary,
that $\tau$ is small enough depending only on $\phi$ and $V$. Also,
as for the lower bound, we may assume that $V\in C_0^3(\R^n)$ with
support in the ball $B_{3/2}$ concentric with $B_1$ and of radius
$3/2$. 

In analogy to the previous proof for the lower bound we define now for
each $(u,q)$ an operator $\hat{h}_{u,q}$ by 
$$
  \hat{h}_{u,q}=\left\{\begin{array}{cl}
  \s(u,q)+ \frac{1}{4b}\Delta\s(u,q) 
  + \nabla\s(u,q)\cdot({\x}-u,-\ic h\nabla -q)&\hbox{ if } u\in B_2\\
  0&\hbox{ if } u\not\in B_2
  \end{array}\right.\,.
$$
The corresponding function is
$$
  {h}_{u,q}(v,p)=\left\{\begin{array}{cl}
  \s(u,q) + \frac{1}{4b}\Delta\s(u,q) 
  + \nabla\s(u,q)\cdot(v-u,p-q)&\hbox{ if } u\in B_2\\
  0&\hbox{ if } u\not\in B_2
  \end{array}\right.\,.
$$
As for the lower bound we shall choose $a=h^{-4/5}$; then $a<h^{-1}$.
In fact, we will assume that
$(1-h^2ab)\geq 1/2$. Recall here 
that $b=2a/(1+h^2a^2)$ (i.e., in particular $a\leq b\leq 2a$).

Similar to \eqref{eq:atildeapp} (for \(\varepsilon=0\)) we have for
$u\in B_2$ that
\begin{eqnarray}
  \lefteqn{\big|h_{u,q}(v,p)-\s(v,p)-\xi_{v,p}(u-v,q-p)\big|}\nonumber\\
  &&\leq C|u-v|(b^{-1}+|u-v|^2)+C|q-p|(b^{-1}+|q-p|^2)\,,\label{eq:happ}
\end{eqnarray}
where
$$
 \xi_{v,p}(u,q) = \mfr{1}{4b}\Delta\s(v,p) 
                - \mfr{1}{2}\sum_{i,j}\p_i\p_j T_\b(p)q_i q_j
                - \mfr{1}{2}\sum_{i,j}\p_i\p_j V(v)u_iu_j\,.
$$

Recalling that $\chi$ is the characteristic function of $\R_-$ we define
\be \label{trial density}
   \c = \int {\cG}_{u,q}\,\chi\big[\hat{h}_{u,q} \big]\,
             {\cG}_{u,q}\,\frac{dudq}{(2\pi h)^n} \,.
\ee 
Since ${\bf 0}\le\chi\big[\hat{h}_{u,q}\big]\le\mathbf1$ it follows
from \eqref{eq:complete} that ${\bf 0}\le\c\le{\bf 1}$.

We now calculate ${\Tr}[\c\phi H_\b\phi]={\Tr}\big[\c\phi (T_\b(-\ic h\nabla) +
V(\x))\phi\big]$. {F}rom Theorem \ref{trace formula} we have that
\begin{eqnarray}\label{eq:first-sc-int-in app}\nonumber
  \lefteqn{
  {\Tr}\big[\c\phi(T_\b(-\ic h\nabla) + V(\hat{x}))\phi\big]}\\ 
  &=&\!\!\!\!
  \int \,\chi[{h}_{u,q}(v,p)]\,G_b(u-v)G_b(q-p)
  \Big[E_2(u,v;q,p) +
  \\
  &&\,\big(\phi(v+h^2ab(u-v))^2 + E_1(u,v)\big)
  \s(v+h^2ab(u-v),p+h^2ab(q-p))\Big] \,\frac{dudq}{(2\pi h)^n} dv 
  dp\,, \nonumber
\end{eqnarray}
where $E_1,E_2$ are functions such that
$\|E_1\|_\infty+\|E_2\|_\infty\leq Ch^2b$. The rest of the proof of 
(\ref{eq:lemmaupper}) is a tedious, but elementary analysis of this
integral. A detailed analysis is presented in Appendix~\ref{app:B}
below. 

It remains to estimate the density $\rho_\c$ together with
$\int\phi(x)^2\rho_\c(x)\,dx$. By Theorem~\ref{thm:traceidentity} and
\eqref{trial density}, $\gamma$ is easily seen to be a trace class
operator with density
\begin{equation}\label{eq:rhointegral}
  \rho_\gamma(x)
  =\int\chi\big[h_{u+v,q+p}(v,p)\big]G_b(u)G_b(q)
  G_{(h^2b)^{-1}}(x-v-h^2abu)\,dvdp\, 
  \frac{dudq}{(2\pi h)^n} \,.
\end{equation}

The proof of \eqref{eq:rhogammaprop1} and \eqref{eq:rhogammaprop2}
again relies on a detailed analysis of this integral. As for the
estimate on the energy above this analysis is an exercise
in calculus. Although it is still elementary this analysis is
more complicated than in the case of the energy.
For the convenience of the reader the analysis is given in detail in
Appendix~\ref{app:B} below. 
\end{proof}

\begin{appendix}

\section{Various Proofs} 
\label{sect:App}

In this appendix we collect proofs of various results mentioned
in Section~\ref{sect:prelim}.
\begin{proof}[Proof of Theorem~\ref{thm:new-critical} {\rm ({\bf
      Operator inequality critical Hydrogen})}]
Let \(f\in\mathcal{S}(\R^3)\) and \(t>0\) (to be chosen below). By
Schwarz' inequality, 
\bea\label{initial-ineq}\nonumber
  \frac{2}{\pi}\int_{\R^3}\frac{|f(x)|^2}{|x|}\,dx
  &=&\frac{1}{\pi^3}\int_{\R^3}\int_{\R^3}\frac{\overline{{\hat
        f}(p)}{\hat
      f}(q)}{|p-q|^2}\Big(\frac{|p|^2+|p|^t}{|q|^2+|q|^t}\Big)^{1/2} 
   \Big(\frac{|q|^2+|q|^t}{|p|^2+|p|^t}\Big)^{1/2}\,dpdq
  \\&\le& \frac{1}{\pi^3}\int_{\R^3}\int_{\R^3}\frac{|{\hat
      f}(p)|^2}{|p-q|^2}
    \,\frac{|p|^2+|p|^t}{|q|^2+|q|^t}\,dpdq\,.
\eea
We first compute the integral in \(q\). Since \((|q|^2+|q|^t)^{-1}\le
|q|^{-2}-|q|^{t-4}+|q|^{2t-6}\) we get
\beax
  \int_{\R^3}\frac{1}{|p-q|^2}\,\frac{1}{|q|^2+|q|^t}\,dq
  \le \int_{\R^3}\frac{1}{|p-q|^2}\,(|q|^{-2}-|q|^{t-4}+|q|)^{2t-6}\,dq\,.
\eeax
Note \cite[5.10 (3)]{Lieb-Loss} that, for \(0<\tau,\sigma<n\), with
\(0<\tau+\sigma<n\),
\bea\label{eq:LL}
  \int_{\R^n}|y-z|^{\tau-n}|z|^{\sigma-n}\,dz=
  \frac{c_{n-\tau-\sigma}c_{\tau}c_{\sigma}}{c_{\tau+\sigma}c_{n-\tau}c_{n-\sigma}}
  \,|y|^{\tau+\sigma-n}\,,  
\eea
where \(c_{\tau}=\pi^{-\tau/2}\Gamma(\tau/2)\). In particular, if \(n=3\), 
then 
\bea\label{eq:LLspec}
  \int_{\R^3}|y-z|^{-2}|z|^{-r}\,dz=k_r\,|y|^{1-r}\text{ for } r\in(1,3)\,,
\eea
with
\bea\label{def:k(r)}
   k_r=\pi^2\frac{\Gamma\big(\frac{r-1}{2}\big)\Gamma\big(\frac{3-r}{2}\big)}
  {\Gamma\big(\frac{4-r}{2}\big)\Gamma\big(\frac{r}{2}\big)}\,.
\eea
It follows that, for \(3<2t<5\), 
\bea\label{first-int}
  \int_{\R^3}\frac{1}{|p-q|^2}\,\frac{|p|^2+|p|^t}{|q|^2+|q|^t}\,dq 
 &\le&
 k_2\,|p|+(k_2-k_{4-t})\,|p|^{t-1}
 \\&&{}+(k_{6-2t}-k_{4-t}\big)\,|p|^{2t-3}+k_{6-2t}\,|p|^{3t-5}\,.\nonumber
\eea
We see from \eqref{def:k(r)} that $k$ is symmetric with respect to
$r=2$. Using $\Gamma(1+z) = 
z\Gamma(z)$ in the denominator in \eqref{def:k(r)} with $z=1-r/2$
and the relation $\Gamma(z)\Gamma(1-z) = \pi/ \sin(\pi z)$ (for $0<z<1$) in the
denominator and numerator we obtain 
$$ k_r = -\pi^2\, \frac{\tan(\pi r/2)}{1-r/2}\,,
$$
which shows that $k$ is
decreasing on $(1,2)$ and increasing on $(2,3)$.

Hence from \eqref{first-int}, choosing further \(t>5/3\), we find, for
positive constants \(A_{(t-1)/2}\), \(B_{(t-1)/2}\), that
\bea\label{second-int}
  \int_{\R^3}\frac{1}{|p-q|^2}\,\frac{|p|^2+|p|^t}{|q|^2+|q|^t}\,dq 
  &\le&
  k_2\,|p|-\pi^3 A_{(t-1)/2}|p|^{t-1}+\pi^3B_{(t-1)/2}\,.
\eea
Since \(k_2=\pi^3\), this and \eqref{initial-ineq} implies that
\bea
  \frac{2}{\pi}\int_{\R^3}\frac{|f(x)|^2}{|x|}\,dx
  &\le& \int_{\R^3}|\hat{f}(p)|\big(|p|-A_{(t-1)/2}|p|^{t-1}+B_{(t-1)/2}\big)\,dp\,,
\eea
which implies the operator inequality,  for all \(t\in(5/3,2)\), 
\bea
  \sqrt{-\D}-\frac{2}{\pi|\x|}&\ge& A_{(t-1)/2}(-\D)^{(t-1)/2}-B_{(t-1)/2}\,.
\eea
Choosing \(t=2s+1\) proves \eqref{eq:new-critical} for
\(s\in(1/3,1/2)\). For \(s\in[0,1/3]\), \eqref{eq:new-critical}
follows from the existence 
of positive constants \(A_{(\tau-1)/2}\), \(B_{(\tau-1)/2} \), given \(\tau\in[1,5/3],
t\in(5/3,2)\) and positive constants \(A_{(t-1)/2}, B_{(t-1)/2}\), such that
\beax A_{(t-1)/2}|p|^{t-1}- B_{(t-1)/2}&\ge& A_{(\tau-1)/2}|p|^{\tau-1}-A_{(\tau-1)/2}\,.
\eeax
\end{proof}

\noindent{\bf Integral representation for the relativistic kinetic energy.}
We shall here give a self-contained presentation of the integral
formulas for the relativistic kinetic energy. The relativistc kinetic
energy will be given in terms of the modified Bessel functions of the
second kind, $K_\nu$. To identify the modified Bessel functions we 
use that \cite[9.6.23]{Abramowitz-Stegun}
\begin{equation}\label{eq:K0}
  K_0(t)=\int_1^\infty \frac{{\rm e}^{-w t}}{\sqrt{w^2-1}}\,dw\,,\quad t>0\,,
\end{equation}
and the recursion relation \cite[9.6.28]{Abramowitz-Stegun}
\begin{equation}\label{eq:K-recursion}
  K_{\nu+1}(t)=-t^{\nu}\frac{d}{dt}(t^{-\nu}K_\nu(t))\,,\quad t>0\,.
\end{equation}
We emphasise that we use these properties only as {\it definitions} of the
Bessel functions, and derive all other properties of these
functions that we need. Note that $K_{\nu}:\R_+\to\R$ are smooth functions.

Consider the function $G^m_n\in L^1(\R^n)$ (the Yukawa potential) whose
Fourier transform is
$$
  \widehat G^m_n(\xi)=(2\pi)^{-n/2}(|\xi|^2+m^2)^{-1}\,.
$$
Using that $v^{-1}=\int_0^\infty {\rm e}^{-uv}\,du$ we get from the
Fourier transform of Gaussian functions the following integral
representation for $G$, 
\begin{equation}\label{eq:G-int-representation}
  G^m_n(z)=\int_0^\infty (4\pi u)^{-n/2} {\rm e}^{-m^2u-|z|^2/(4u)}\,du\,.
\end{equation}
It follows from this that $G$ is non-negative, smooth for $z\ne0$, and
indeed in $L^1(\R^n)$.

For odd $n$ the above integral can be explicitly calculated. For even
$n$ it is as we shall now see expressible as a modified Bessel
function $K_\nu$ of integer order $\nu$.
By a simple change of variables (\(2w=v+v^{-1}\) with \(v=2mu/|z|\))
in the integral \eqref{eq:K0} we see from \eqref{eq:G-int-representation} 
that 
$
  G^m_{2}(z)=(2\pi)^{-1}K_0(m|z|)\,.
$
{F}rom the recursion formula \eqref{eq:K-recursion} we then find
inductively that for even $n$
\begin{equation}\label{eq:G-K}
  G^m_{n}(z)=m^{(n-2)/2}(2\pi)^{-n/2} |z|^{-(n-2)/2}K_{(n-2)/2}(m|z|)\,.
\end{equation}
In fact, the same formula holds for all $n$, but we do not wish to
discuss the modified Bessel functions of fractional order (one could
simply take this formula as their definition). 
\begin{lemma}\label{eq:heat-kernel}
The heat kernel for the operator $\sqrt{-\Delta+m^2}$ on $L^2(\R^n)$
is given by  
\begin{eqnarray*} 
  \lefteqn{\exp(-t\sqrt{-\Delta+m^2})(x,y)}&&\\&=&-2\partial_t
  G^m_{n+1}(x-y,t)\\&=&
  2\left(\frac{m}{2\pi}\right)^{(n+1)/2}
  \frac{t}{(|x-y|^2+t^2)^{(n+1)/4}}K_{(n+1)/2}(m(|x-y|^2+t^2)^{1/2}) 
\end{eqnarray*}
for $t>0$.
\end{lemma}
\begin{proof}
It suffices to show that the two tempered distributions on $\R^{n+1}$,
$$
  \frac{t}{|t|}\exp(-|t|\sqrt{-\Delta+m^2})(x,0)\ \ \hbox{ and }\  \ 
  -2\partial_t G^m_{n+1}(x,t)\,,
$$
have the same Fourier transform. The Fourier transform as a function
of $\xi=(p,s)$ with $p\in\R^n$ and $s\in \R$ of the first distribution
is 
$$
  (2\pi)^{-(n+1)/2}\big(-\int_{-\infty}^0{\rm e}^{-\ic ts+t\sqrt{p^2+m^2}}\,dt
  +\int_{0}^\infty{\rm e}^{-\ic ts-t\sqrt{p^2+m^2}}\,dt\big)= 
  \frac{-2\ic s (2\pi)^{-(n+1)/2}}{|p|^2+s^2+m^2}\,.
$$
The Fourier transform of the second distribution above is 
$$
  -2\ic s\,\widehat G^m_{n+1}(p,s)=\frac{-2\ic s (2\pi)^{-(n+1)/2}}{|p|^2+s^2+m^2}\,.
$$
The last identity in the lemma follows from \eqref{eq:K-recursion} and
\eqref{eq:G-K}.
\end{proof}
If we set $x=y$ in the above lemma we find the following integral
formula for the modified Bessel function
$$
  K_{(n+1)/2}(t)
  =\frac{1}{2}\left(\frac{t}{2\pi}\right)^{(n-1)/2}
  \int_{\R^n}{\rm e}^{-t\sqrt{p^2+1}}\,dp\,,\quad t>0\,.
$$
For $n=3$ this simplifies to
\begin{equation}\label{eq:K2formula}
  K_2(t)=t\int_0^\infty{\rm e}^{-t\sqrt{s^2+1}}s^2\,ds
\end{equation}
from which we immediately get the estimate
\begin{equation}\label{eq:K2-estimate}
  K_2(t)\leq Ct^{-2}{\rm e}^{-t/2}\,.
\end{equation}

\begin{proof}[Proof of Theorem~\ref{IMS} {\rm ({\bf Relativistic IMS
	formula})}]
By scaling, it suffices to prove the statement for \(\alpha=1\).
We start from the identity
$$ \big( f,(\sqrt{-\D+m^2}-m) f\big) = \int |f(x)-f(y)|^2 F(x-y)\, dxdy
$$
with 
\be \label{Formel-kernel} F(x-y) = \frac{m^2 }{4\pi^2} \,\frac{K_2(m|x-y|)}{|x-y|^2}\,,
\ee
where \(K_2\) is the modified Bessel function of second order defined
above (see \eqref{eq:K0}--\eqref{eq:K-recursion}). The identity
follows from Lemma~\ref{eq:heat-kernel} (for a proof, see
\cite[7.12]{Lieb-Loss}).
Then,
\beax
  \lefteqn{\big( f,(\sqrt{-\D+m^2}-m) f\big)
  }\nonumber\\
  &=&\int |f(x)-f(y)|^2 F(x-y)\, dxdy
  \nonumber\\
  &=&\int\int_{\mathcal{M}} \big[\theta_u(x)^2|f(x)|^2 +
  \theta_u(y)^2|f(y)|^2\big] \, F(x-y)\,d\mu(u)dxdy 
  \nonumber\\
  &&+\,\int\int_{\mathcal{M}}\big[-\mfr{1}{2}(\theta_u(x)^2+\theta_u(y)^2) 
  + \theta_u(x) \theta_u(y) - \theta_u(x)  
   \theta_u(y) \big] 
  \nonumber\\
  &&\phantom{+}\qquad\qquad\qquad\qquad\qquad\quad\times\,\big[\,\overline{f(x)}
  f(y) 
  + f(x)\overline{f(y)}\,\big]\,F(x-y)\,d\mu(u)dxdy
  \nonumber\\
  &=&\int\int_{\mathcal{M}} |\theta_u(x)f(x)-\theta_u(y)f(y)|^2 F(x-y)\, d\mu(x)dxdy
  \nonumber\\
  &&+\,\int\int_{\mathcal{M}}\big[-\mfr{1}{2}(\theta_u(x)^2+\theta_u(y)^2) 
  + \theta_u(x) \theta_u(y)  
  \big] \nonumber\\
  &&\phantom{+}\qquad\qquad\qquad\qquad\qquad\quad\times\
  \,\big[\,\overline{f(x)} f(y) +
  f(x)\overline{f(y)}\,\big]\,F(x-y)\,d\mu(u)dxdy 
  \nonumber\\
  &=&\int_{\mathcal M} \big(\theta_u f,(\sqrt{-\D+m^2}-m) \theta_u f\big) \, d\mu(u) 
  \nonumber\\
  && -\, \int\int_{\mathcal{M}} \big(\theta_u(x)-\theta_u(y)\big)^2 f(x) \overline{f(y)} F(x-y)\,
  d\mu(u) dx dy\,.
\eeax
This proves \eqref{eq:IMS} with \(L\) given by \eqref{IMS-error1}--\eqref{IMS-error2}. 
We now show that $\|L_{\theta_u}\| \le C m^{-1}\|\nabla 
\theta_u\|_\infty^2$ for fixed \(u\). 
By \eqref{IMS-error2}, Young's inequality, and \eqref{eq:K2-estimate}, 
\bea\label{eq:bound-L}\nonumber
  \big|(f,L_{\theta_u}f)\big|&\leq&
  \frac{m^2}{4\pi^2}\|\nabla\theta_u\|_\infty^2\int |f(x)|\,|f(y)|\,
  K_2(m|x-y|)\, dx dy\,\\\nonumber
  &\leq& C\,m^2\|\nabla\theta_u\|_\infty^2\|f\|_2^2 \int_0^\infty t^2 K_2(mt)\,dt \\
  &=& C\,m^{-1}\|\nabla\theta_u\|_\infty^2\|f\|_2^2 \,.
\eea
This proves that \(L_{\theta_u}\) is a bounded operator.
\end{proof}

\begin{proof}[Proof of Theorem~\ref{IMS-error-est} {\rm ({\bf
      Localisation error})}] 
Again, by scaling, it suffices to prove the statement for \(\alpha=1\).
With \(\chi_\Omega\) the characteristic function of \(\Omega\) (and
\(L\equiv L_\theta\)) we
have from the representation \eqref{IMS-error2} of $L$, since
\(\theta\) is constant on \(\Omega^c\), that 
\be
  L=\chi_{\Omega}L\chi_{\Omega}+(1-\chi_{\Omega})L\chi_{\Omega}
  +\chi_{\Omega}L(1-\chi_{\Omega}).
\ee
If $\Gamma_1$, $\Gamma_2$ are bounded operators, then
$(\Gamma_1-\Gamma_2)(\Gamma_1-\Gamma_2)^*\ge0$ implies that
$\Gamma_1\Gamma_2^* + \Gamma_2\Gamma_1^* \le \Gamma_1\Gamma_1^* +
\Gamma_2\Gamma_2^*$. Using this with $\Gamma_1  =
\e^{1/2}\chi_{\Omega} , \Gamma_2 = \e^{-1/2} (1-\chi_{\Omega}) L$ for
some $\e>0$ which we choose later, we get
\be\label{inequ:L}
    L \le \chi_{\Omega}L \chi_{\Omega} + \e\chi_{\Omega} + \e^{-1}
    (1-\chi_{\Omega}) L^2 (1-\chi_{\Omega})\,.
\ee
To bound the first term on the right side recall that $\|L\| \le C
m^{-1}\|\nabla \theta\|_\infty^2$ (see \eqref{eq:bound-L}).

Let us now look at the third term in \eqref{inequ:L}. 
Since \(\theta\) is constant on
\(\Omega^c\) and \(\dist(\Omega^c,\supp\,\nabla\theta)\ge \ell\), using 
\eqref{IMS-error2} gives
\beax 
  \Tr\big[(1-\chi_{\Omega}) L^2 (1-\chi_{\Omega})\big] &=&
  \int_{x\in\Omega^c,y\in\Omega, |x-y|> \ell} L(x,y)^2\, dxdy 
  \\
  &\le& C m^4 \|\nabla\theta\|_\infty^4
  \int_{x\in\Omega^c,y\in\Omega,|x-y|> \ell} 
  K_2(m|x-y|)^2\, dxdy\,.
\eeax
Using \eqref{eq:K2-estimate},
\begin{align*}
  \int_{x\in\Omega^c,y\in\Omega,|x-y|> \ell} 
  &K_2(m|x-y|)^2\, dxdy
  \\
  &\le C {\rm e}^{-m\ell}
  \int_{x\in\Omega^c,y\in\Omega,|x-y|>\ell} 
  (m|x-y|)^{-4} \, dxdy
  \\
  &= C {\rm e}^{-m\ell} |\Omega| \int_{\ell}^{\infty}(mt)^{-4}t^2\,dt
  = C m^{-3}(m\ell)^{-1}{\rm e}^{-m\ell} |\Omega|\,.
\end{align*}
This gives the bound
\beax 
  \Tr\big[(1-\chi_{\Omega}) L^2 (1-\chi_{\Omega})\big] 
  &\le& C\ell^{-1}\|\nabla\theta\|_\infty^{4} 
  {\rm e}^{-m\ell}|\Omega| \,.
\eeax
Finally, we choose $\e=m^{-1}\|\nabla\theta\|_\infty^2$.
Then by the above the two first terms in \eqref{inequ:L} are bounded
by $Cm^{-1}\|\nabla\theta\|_\infty^2\chi_{\Omega}$, and the trace of
the third term (which we denote $Q_{\theta}$) is bounded by $Cm
\ell^{-1}{\rm e}^{-m\ell}\|\nabla\theta\|_\infty^2|\Omega|$.
\end{proof}

For the proof of the combined Daubechies-Lieb-Yau inequality
(Theorem~\ref{thm:L-Y-D}) we need 
the following inequality \cite{AD-JPS}.
\begin{lemma}\label{lem:int-ineq}
For \(f\in\mathcal{S}(\R^3)\), 
\bea
    \int_{\mathbb R^3}
    \frac{{\rm e}^{-m^2\pi^{-1}|x|^2}}{|x|}\,|f(x)|^2\,dx
    \le \frac{\pi}{2}\frac{1}{\sqrt{2}-1}\,\big(f,(\sqrt{-\D+m^2}-m)f\big)\,.  
\eea
\end{lemma}
\begin{proof}
 Let \(\mu=m^2\pi^{-1}\). Then
\beax
  I=\int_{\mathbb R^3}\frac{{\rm e}^{-\mu|x|^2}}{|x|}\,|f(x)|^2\,dx
      =\frac{1}{2\pi^2}\int_{\mathbb R^3}\int_{\mathbb R^3}
      \overline{\hat{f}(p_1)}\frac{1}{|p_1-p_2|^2}\,\hat{g}(p_2)\,dp_1dp_2\,,
\eeax
with \(g(x)=f(x){\rm e}^{-\mu|x|^2}\). Writing \(\hat{g}(p_2)\)
explicitly as the convolution with the Fourier transform of \({\rm
  e}^{-\mu|x|^2}\) and then applying the Schwarz inequality we get that
\beax
  I\le \frac{1}{16\pi^2(\pi\mu)^{3/2}}\int_{\mathbb R^3}
  \int_{\mathbb R^3}\int_{\mathbb R^3}
  \frac{|\hat{f}(p_1)|^2{\rm e}^{-|q|^2/(4\mu)}}{|p_2-p_1-q|^2}
  \frac{|p_1|^2}{|p_2|^2}\,dp_1dp_2dq\,.
\eeax
Since \cite[5.10 (3)]{Lieb-Loss}
\beax
   \int_{\mathbb R^3}\frac{1}{|p_2-p_1-q|^2}\frac{1}{|p_2|^2}\,dp_2
   =\frac{\pi^3}{|p_1-q|}\,,
\eeax
we have
\beax
  I\le \frac{1}{16\pi^{1/2}\mu^{3/2}}\int_{\mathbb
    R^3}\int_{\mathbb R^3}
  \frac{|\hat{f}(p_1)|^2{\rm e}^{-|q|^2/(4\mu)}}{|p_1-q|}
  |p_1|^2\,dp_1dq\,.
\eeax
By Newton's theorem \cite[9.7 (5)]{Lieb-Loss},
\beax
   \int_{\mathbb R^3}\frac{{\rm e}^{-|q|^2/(4\mu)}}{|p_1-q|}\,dq
   &=& \frac{1}{|p_1|}\int_{|q|<|p_1|}{\rm e}^{-|q|^2/(4\mu)}\,dq
   + \int_{|q|>|p_1|}\frac{{\rm e}^{-|q|^2/(4\mu)}}{|q|}\,dq
   \\&=&{}\frac{8\pi\mu}{|p_1|}\int_0^{|p_1|}{\rm e}^{-r^2/(4\mu)}\,dr
   \le 8\pi\mu \min\big\{1,\frac{(\pi\mu)^{1/2}}{|p_1|}\big\}\,.
 \eeax
Substituting \(\mu=m^2\pi^{-1}\) we
find that
\beax
  I\le \frac{\pi}{2m}\int_{\mathbb R^3}
  |\hat{f}(p_1)|^2\min\{|p_1|^2,m|p_1|\}\,dp_1\,,
\eeax
from which the claim follows since \(\sqrt{t^2+1}-1\ge
(\sqrt{2}-1)\min\{t^2,t\}\) for \(t\ge0\).
\end{proof}
\begin{proof}[Proof of Theorem~\ref{thm:L-Y-D} {\rm ({\bf Combined
       Daubechies-Lieb-Yau})}] 
We may assume that $W(x)\leq 0$ otherwise we simply replace $W$ by
$W_{-}$.

Assume first that \(\nu\alpha\le3/(16\pi M)\). By the Daubechies
inequality \eqref{eq:Daub usefull},  
\bea\label{eq:firstDaubInImproved}\nonumber
  \lefteqn{\Tr\big[\sqrt{-\alpha^{-2}\D
  +m^2\alpha^{-4}}-m\alpha^{-2}+W(\x)\big]_{-}} 
  &&\\&\ge&\tfrac12\,\Tr\big[\sqrt{-\alpha^{-2}\D
  +m^2\alpha^{-4}}-m\alpha^{-2}
  +2W\chi_{\{d_{\bf R}(x)<\alpha m^{-1}\}}\big]_{-}
  \\ &&{}- Cm^{3/2}\int_{d_{\bf R}(x)>\alpha m^{-1}}|W(x)|^{5/2}\,dx
  -C\alpha^{3}\int_{d_{\bf R}(x)>\alpha m^{-1}}|W(x)|^{4}\,dx\,.
  \nonumber
\eea
The assumption on the positions of the nuclei implies that
\(\chi_{\{d_{\bf R}(x)<\alpha m^{-1}\}}=\sum_{j=1}^{M} 
\chi_{\{|x-R_j|<\alpha m^{-1}\}}\), and so, using the assumption on
\(W\), we obtain
\bea\label{eq:secondIneq}\nonumber
    \lefteqn{\Tr\big[\sqrt{-\alpha^{-2}\D+m^2\alpha^{-4}}-m\alpha^{-2}
  +2W\chi_{\{d_{\bf R}(x)<\alpha m^{-1}\}}\big]_{-}}
  &&\\&\ge& \frac{1}{M}\sum_{j=1}^{M}
  \Tr\big[\sqrt{-\alpha^{-2}\D+m^2\alpha^{-4}}-m\alpha^{-2}
  -\big(\,\frac{2\nu M}{|{\hat x}-R_j|}+C\nu
  Mm\alpha^{-1}\big)\chi_{\{|x-R_j|<\alpha m^{-1}\}}\big]_{-}
  \nonumber
  \\&=&\Tr\big[\sqrt{-\alpha^{-2}\D+m^2\alpha^{-4}}-m\alpha^{-2}
  -\big(\,2\nu M|{\hat x}|^{-1}+C\nu
  Mm\alpha^{-1}\big)\chi_{\{|x|<\alpha m^{-1}\}}\big]_{-}\,.
\eea
The last equality follows from the translation invariance of
\({}-\D\). By scaling,
\bea\label{eg:oneScale}\nonumber
  \lefteqn{\Tr\big[\sqrt{-\alpha^{-2}\D+m^2\alpha^{-4}}-m\alpha^{-2}
   -\big(2\nu M|{\hat x}|^{-1}+C\nu
   Mm\alpha^{-1}\big)\chi_{\{|x|<\alpha m^{-1}\}}\big]_{-}}
   &&\\&=&\alpha^{-2}\,\Tr\big[\sqrt{-\D+m^2}-m
   -\big(\gamma|{\hat x}|^{-1}+C\gamma m\big)
   \chi_{\{|x|<m^{-1}\}}\big]_{-}\,,
\eea
with \(\gamma=2M\nu\alpha\le 3/(8\pi)\). Using
Lemma~\ref{lem:int-ineq} and the Daubechies inequality, we get that
\beax
  \lefteqn{\Tr\big[\sqrt{-\D+m^2}-m
  -\big(\gamma|{\hat x}|^{-1}+C\gamma m\big)\chi_{\{|x|<m^{-1}\}}\big]_{-}}
  &&\\&\ge& \big(1-\tfrac{4\pi}{3}\gamma\big)\,
  \Tr\big[\sqrt{-\D+m^2}-m-\gamma(1-\tfrac{4\pi}{3}\gamma\big)^{-1}
  \big(\,\frac{1-{\rm e}^{-m^2\pi^{-1}|{\hat x}|^2}}
  {|{\hat x}|}+Cm\big)\chi_{\{|x|<
  m^{-1}\}}\big]_{-}
  \\&\ge&{}-C\gamma^{5/2}m^{3/2}\int_{|x|<
  m^{-1}}\big(|x|^{-1}+m\big)^{5/2}\,dx
  -C\gamma^{4}\int_{|x|<m^{-1}}\Big(\,
  \frac{1-{\rm e}^{-m^2\pi^{-1}|x|^2}}{|x|}+m\Big)^4\,dx\,, 
\eeax
where we have used that \(\sqrt{2}-1\ge 3/8\) and \(\gamma\le 3/(8\pi)\).

Note that
\beax
  \int_{|x|<
    m^{-1}}\big(|x|^{-1}+m\big)^{5/2}\,dx
  \le Cm^{-1/2}\ , \quad
  \int_{|x|<m^{-1}}\Big(\,\frac{1-{\rm e}^{-m^2\pi^{-1}|x|^2}}
  {|x|}+m\Big)^4\,dx
  \le Cm\,,
\eeax 
and so
\bea\label{eq:oneMore}
  \Tr\big[\sqrt{-\D+m^2}-m
  -\big(\gamma|{\hat x}|^{-1}+C\gamma m\big)
  \chi_{\{|x|<m^{-1}\}}\big]_{-}
  \ge{}-C(\gamma^{5/2}+\gamma^{4})m\,.
\eea
Combining \eqref{eg:oneScale} and \eqref{eq:oneMore}, and using
\(\gamma=2M\nu\alpha\), \(\nu\alpha\le 2/\pi\), we get that 
\bea\label{eg:thirdOneNow}\nonumber
  \lefteqn{\Tr\big[\sqrt{-\alpha^{-2}\D+m^2\alpha^{-4}}-m\alpha^{-2}
  -\big(2\nu M|{\hat x}|^{-1}+C\nu
  Mm\alpha^{-1}\big)\chi_{\{|x|<\alpha m^{-1}\}}\big]_{-}}
  &&\\&\ge&{}-C\alpha^{-2}(\gamma^{5/2}+\gamma^{4})m\ge
  {}-C\nu^{5/2}\alpha^{1/2}m\,.
  \hspace{6cm}
\eea
Combining \eqref{eq:firstDaubInImproved}, \eqref{eq:secondIneq}, and
\eqref{eg:thirdOneNow}, yields \eqref{ineq:ImprovedDLY} for
\(\nu\alpha\le 3/(16\pi M)\).  

Assume now that \(\nu\alpha\in [3/(16\pi M),2/\pi]\). 

Let \(\theta\in C_0^\infty(\R^3)\) satisfy \(0\le\theta(x)\le1\),
\(\theta(x)=1\) for \(|x|\le \alpha m^{-1}/4\), \(\theta(x)=0\) for
\(|x|\ge \alpha m^{-1}/2\), \((1-\theta^2)^{1/2}\in C^{1}(\R)\), and 
$$
  \|\nabla\theta\|_\infty\leq C\alpha^{-1} m,\quad
  \|\nabla(1-\theta^2)^{1/2}\|_\infty\leq C\alpha^{-1} m\,. 
$$
Let \(\theta_j(x)=\theta(x-R_j)\), \(j=1,\ldots,M\), and
\(\theta_{M+1}(x)=(1-\sum_{j=1}^M\theta_j^2)^{1/2}\) (the latter is
well-defined due to the assumption $\min_{k\ne \ell}|R_k-R_\ell|>2\alpha
m^{-1}$). 
The relativistic IMS formula and the localisation estimate 
\eqref{IMS-err-op-est} used for 
\(\Omega_j\), \(j=1,\ldots,M\), being the balls centered at \(R_j\) with
radii \(3\alpha m^{-1}/4\) and \(\ell=\alpha 
m^{-1}/4\), and \(\Omega_{M+1}\) being the (disjoint) union of the
same balls
 and \(\ell=\alpha m^{-1}/8\), gives
the operator inequality 
\bea\label{eq:firstBigLocal} \nonumber
 \lefteqn{\sqrt{-\alpha^{-2}\Delta+m^2\alpha^{-4}}-m\alpha^{-2}+W({\hat x})}&& 
 \\&\geq&
 \theta_{M+1}\big(\sqrt{-\alpha^{-2}\Delta
 +m^2\alpha^{-4}}-m\alpha^{-2}-Cm\alpha^{-2}\sum_{j=1}^M \chi_{\Omega_j}
 +W({\hat x})\big)\theta_{M+1} 
 \\&+&
 \sum_{j=1}^M \theta_j\big(\sqrt{-\alpha^{-2}\Delta
 +m^2\alpha^{-4}}-m\alpha^{-2}-Cm\alpha^{-2}+W({\hat x})\big)\theta_j
 -\sum_{j=1}^{M+1}Q_j\,,
 \nonumber
\eea
with 
\beax
  \Tr\big[Q_j\big]\le C\,m\alpha^{-2} \,.
\eeax
 Here we have used that
\(\theta_i\chi_{\Omega_j}\theta_i=\delta_{ij}\,\theta_i{}^2\),
\(i,j\in\{1,\ldots,M\}\),
\(\theta_{M+1}\chi_{\Omega_{M+1}}\theta_{M+1}=0\), and
\(\theta_i\chi_{\Omega_{M+1}}\theta_i\le\theta_i^2\), \(i\neq M+1\). 

Using the Daubechies
inequality on  the first term in \eqref{eq:firstBigLocal} and the assumption
on \(W\) in the second (noticing that
\(\theta_j(x)/d_{\bf R}(x)=\theta_j(x)/|x-R_j|\) due to the assumption 
\(\min_{k\ne \ell}|R_k-R_\ell|>2\alpha m^{-1}\)), 
we get from this that 
\bea \nonumber
  \lefteqn{\Tr[\sqrt{-\alpha^{-2}\Delta+m^2\alpha^{-4}}
  -m\alpha^{-2}+W({\hat x})]_{-}}&& 
  \\&\ge&
  \nonumber
  {}-Cm^{3/2}\!\!\!\!\!\!\!\!
  \int\limits_{d_{\bf R}(x)>\alpha m^{-1}/4}
  \!\!\!\!\!\!\!\!|W(x)|^{5/2}\,dx
  {}-C\alpha^{3}\!\!\!\!\!\!\!\!\int\limits_{d_{\bf R}(x)>\alpha m^{-1}/4}
  \!\!\!\!\!\!\!\!|W(x)|^{4}\,dx
  -C\,m\alpha^{-2}
  \\&&{}-C\sum_{j=1}^{M}\big(
   m^{3/2}(m\alpha^{-2})^{5/2}+\alpha^{3}(m\alpha^{-2})^{4}\big)
  \big|\{x\,|\,\tfrac14\alpha m^{-1}<|x-R_j|< \tfrac34\alpha
  m^{-1}\}\big|
  \label{eq:before-End}
  \\
  &&{}+ \sum_{j=1}^M\Tr\big[\theta_j\big(\sqrt{-\alpha^{-2}\Delta
  +m^2\alpha^{-4}}-m\alpha^{-2}-Cm\alpha^{-2}-\frac{\nu}{|{\hat x}-R_j|}-C\nu
  m\alpha^{-1}\big)\theta_j\big]_{-}\,.
  \nonumber
\eea
By the translation invariance of \(-\D\), the last term equals
\beax
  M\,\Tr\big[\theta\big(\sqrt{-\alpha^{-2}\Delta
  +m^2\alpha^{-4}}-Cm\alpha^{-2}-C\nu
  m\alpha^{-1}-\nu|{\hat x}|^{-1}\big)\theta\big]_{-}\,,
\eeax
and using the Lieb-Yau inequality 
(and the properties of \(\theta\)
and that \(\nu\alpha\le 2/\pi\)), 
\bea\label{eq:one-Lieb-Yau}\nonumber
  \lefteqn{\Tr\big[\theta\big(\sqrt{-\alpha^{-2}\Delta
  +m^2\alpha^{-4}}-Cm\alpha^{-2}-C\nu
  m\alpha^{-1}-\nu|{\hat x}|^{-1}\big)\theta\big]_{-}}
  &&\\&\ge&
  \alpha^{-1}\,\Tr\big[\theta\big(\sqrt{-\Delta}
  -\frac{2}{\pi}|{\hat x}|^{-1}-Cm\alpha^{-1}\big)\theta\big]_{-}
  \ge{}-Cm\alpha^{-2}\,.
\eea
Further, by the assumption on \(W\), and the assumption 
\(\min_{k\ne \ell}|R_k-R_\ell|>2\alpha m^{-1}\), 
\beax
  \lefteqn{m^{3/2}\!\!\!\!\!\!\!\!\!\!\!\!\!\!\int
  \limits_{\alpha m^{-1}/4<d_{\bf R}(x)<\alpha m^{-1}}
  \!\!\!\!\!\!\!\!\!\!\!\!\!\!|W(x)|^{5/2}\,dx}
  &&\\{\ }\qquad
  &\le& \sum_{j=1}^{M}\ C\nu^{5/2}m^{3/2}
  \!\!\!\!\!\!\!\!\!\!\!\!\!\!
  \int\limits_{\alpha m^{-1}/4<|x-R_j|<\alpha m^{-1}}
  \!\!\!\!\!\!\!\!\!\!\!\!\!\!\!\!\!\!\!
  \big(|x-R_j|^{-1}+\alpha^{-1}m\big)^{5/2}\,dx
  \le C\nu^{5/2}\alpha^{1/2}m\,,
\eeax
and, since \(\nu\alpha\le2/\pi\), 
\beax
  \lefteqn{\alpha^{3}\!\!\!\!\!\!\!\!\!\!\!\!\!\!\int
  \limits_{\alpha m^{-1}/4<d_{\bf R}(x)<\alpha m^{-1}}
  \!\!\!\!\!\!\!\!\!\!\!\!\!\!|W(x)|^{4}\,dx}
  &&\\{\ }\qquad
  &\le& \sum_{j=1}^{M}\ C\nu^{4}\alpha^{3}
  \!\!\!\!\!\!\!\!\!\!\!\!\!\!
  \int\limits_{\alpha m^{-1}/4<|x-R_j|<\alpha m^{-1}}
  \!\!\!\!\!\!\!\!\!\!\!\!\!\!\!\!\!\!\!
  \big(|x-R_j|^{-1}+\alpha^{-1}m\big)^{4}\,dx
  \le C\nu^{4}\alpha^{2}m \le C\nu^{5/2}\alpha^{1/2}m\,.
\eeax
It follows from this, \eqref{eq:one-Lieb-Yau}, and
\eqref{eq:before-End} that 
\beax
  \lefteqn{\Tr[\sqrt{-\alpha^{-2}\Delta+m^2\alpha^{-4}}
  -m\alpha^{-2}+W({\hat x}) ]_{-}}&&
  \\&\ge&
  {}-Cm\alpha^{-2}-C\nu^{5/2}\alpha^{1/2}m-Cm^{3/2}\!\!\!\!\!\!\!\!
  \int\limits_{d_{\bf R}(x)>\alpha m^{-1}}
  \!\!\!\!\!\!\!\!|W(x)|^{5/2}\,dx
  {}-C\alpha^{3}\!\!\!\!\!\!\!\!\int\limits_{d_{\bf R}(x)>\alpha m^{-1}}
  \!\!\!\!\!\!\!\!|W(x)|^{4}\,dx\,.
\eeax
Since \(m\alpha^{-2}\le C\nu^{5/2}\alpha^{1/2}m\) when \(\nu\alpha\in
[3/(16\pi M),2/\pi]\) this proves \eqref{ineq:ImprovedDLY} in this
case. This finishes the proof of Theorem~\ref{thm:L-Y-D}.  
\end{proof}

\begin{proof}[Proof of Theorem~\ref{thm:correlation} {\rm ({\bf
      Correlation inequality})}] The proof essentially uses
  superharmonicity and positive definiteness of $|x|^{-1}$ (see
  \cite{Lieb-Loss}). By superharmonicity and the spherical properties
  of $\Phi_s$,
$$ |x-y|^{-1} \ge \int \Phi_s(z-x)\,|z-z'|^{-1}\, \Phi_s(z'-y) \, dz dz'\,.
$$
Also note that for the Coulomb energy,
$D(\Phi_s,\Phi_s)=s^{-1}D(\Phi,\Phi)=Cs^{-1}$. Therefore, we
immediately get that 
\beax 
  \lefteqn{\sum_{1\le i<j\le N}|x_i-x_j|^{-1}
  }\\
  &\ge& \sum_{1\le i<j\le N} \int
  \Phi_s(z-x_i)\,|z-z'|^{-1}\,\Phi_s(z'-x_j) \, dz dz'  
  \\
  &=& \mfr{1}{2} \int \Big(\sum_{1\le j \le N} \Phi_s(z-x_j)\Big)
  \,|z-z'|^{-1}\, 
  \Big(\sum_{1\le j \le N} \Phi_s(z'-x_j)\Big)\, dz dz' 
  \\
  &&-\,\mfr{N}{2} \int \Phi_s(z) |z-z'|^{-1}\, \Phi_s(z') \, dz dz' 
  \\
  &=&\mfr{1}{2} \int \Big(\sum_{1\le j \le N} \Phi_s(z-x_j) -
  \rho(z)\Big)\, 
  |z-z'|^{-1}\, \Big(\sum_{1\le j \le N} \Phi_s(z'-x_j) -
  \rho(z')\Big)\, dz dz' 
  \\
  &&+\,\int \rho(z) \,|z-z'|^{-1}\,\sum_{1\le j \le N}
  \Phi_s(z'-x_j)\, dzdz' 
  \\
  && -\,\mfr{1}{2} \int \rho(z)\, |z-z'|^{-1}\, \rho(z') \, dzdz' -  
  N D(\Phi_s,\Phi_s)
  \\
  &\ge&\sum_{1\le j\le N}\big(\rho * |x|^{-1} * \Phi_s \big)(x_j) 
     - D(\rho) - CNs^{-1}\,.
\eeax
In the last inequality we have used the positive definiteness of
$|x|^{-1}$ and dropped the first term. This proves inequality 
\eqref{inequ:correlation}. 
\end{proof}

\begin{proof}[Proof of Corollary~\ref{estimate related to correlation}
{\rm{({\bf Estimate on $\rho^{\rm TF}*|x|^{-1}*(\delta_0-\Phi_t)$})}}]
Let, with \(d_{\bf r}\) as in \eqref{ddefinitionBIS},  
\begin{equation}\label{gdefinition}
   g_{\bf r}(x)=\min\{d_{\bf r}(x)^{-1/2}, d_{\bf r}(x)^{-2}\}\,.
 \end{equation}
We claim that for some constant $C>0$,
\be\label{grad-rho} 
   \big|\nabla \rho^{{\rm TF}} * |x|^{-1}\big| \,\le\, C g_{\bf r}(x) \,.
\ee
To prove this we distinguish two regions. 

First, let $d_{\bf r}(x)\ge1$. Then $g_{\bf r}(x)=d_{\bf r}(x)^{-2}$. Using
that $|\nabla V^{{\rm TF}}|(x) \le C g_{\bf r}(x)^2 d_{\bf r}(x)^{-1}
= d_{\bf r}(x)^{-5}$ we obtain 
\beax
  \big|\nabla \rho^{{\rm TF}}* |x|^{-1}\big|
  &=& \Big|\nabla\Big(\sum_{j=1}^{M} z_j |x-r_j|^{-1} - V^{{\rm
      TF}}(x)\Big)\Big| 
  \\
  &\le& \sum_{j=1}^{M} z_j |x-r_j|^{-2} + d_{\bf r}(x)^{-5}
  \le M \max_j\{z_j\} \big(\min_j |x-r_j|\big)^{-2} + d_{\bf
    r}(x)^{-5} 
  \\
  &=& C d_{\bf r}(x)^{-2} + d_{\bf r}(x)^{-5}
  \le Cg_{\bf r}(x)\,.
\eeax

If on the other hand $d_{\bf r}(x)<1$, then, by using
\eqref{eq:tfeqgeneral} and \eqref{eq:tfdf}, 
\beax 
  \big|\nabla \rho^{{\rm TF}} * |x|^{-1}\big|
  &\le& \int \big|\nabla \rho^{{\rm TF}}(y)\big|\, |x-y|^{-1} \, dy 
  \le C \int g_{\bf r}(y)^3 d_{\bf r}(y)^{-1}\,|x-y|^{-1} \, dy 
  \\
  &=&C\sum_{j=1}^{M}\int\limits_{d_{\bf r}(y)=|y-r_j|} 
  g_{\bf r}(y)^3 d_{\bf r}(y)^{-1}\,|x-y|^{-1} \, dy\,.
\eeax
Since \(g_{\bf r}(y)\le d_{\bf r}(y)^{-1/2}=|y-r_j|^{-1/2}\) when
\(d_{\bf r}(y)=|y-r_j|\) this implies that 
\beax 
  \big|\nabla \rho^{{\rm TF}} * |x|^{-1}\big|
  &\le& C\sum_{j=1}^M \int |y-r_j|^{-5/2} \, |x-y|^{-1}\, dy
  \\&=&C\sum_{j=1}^M |x-r_j|^{-1/2}
  \le C M d_{\bf r}(x)^{-1/2} 
  = C g_{\bf r}(x)\,.
\eeax      
This finishes the proof of \eqref{grad-rho}.

Let us now proceed to prove inequality (\ref{correl:upper
  bound}). That the difference is positive is again just
superharmonicity of $|x|^{-1}$. It is easy to see that
\be\label{d-inequ} |d_{\bf r}(x)-d_{\bf r}(y)| \le |x-y|\,.
\ee
In the case when $d_{\bf r}(x)\ge 2t$ we can conclude that
\beax 
  \rho^{{\rm TF}}*|x|^{-1} - \rho^{{\rm TF}}*|x|^{-1}*\Phi_t
  &\le&\sup_{|z-x|\le t} \big\{\left|\nabla \rho^{{\rm
        TF}}*|z|^{-1}\right| \big\} 
     \int |x-y|\,\Phi_t(x-y)\, dy
  \\
  &\le& Ct \sup_{|z-x|\le t} g_{\bf r}(z)
  \le Ctg_{\bf r}(x)\,.
\eeax
In the last step we have used that if $d_{\bf r}(x)\ge 2t$ and
$|z-x|\le t$ (this condition stems from the support of $\Phi_t$), then
inequality (\ref{d-inequ}) guarantees that $\frac{1}{2} d_{\bf
  r}(x)\le d_{\bf r}(z)\le \frac{3}{2} d_{\bf r}(x)$. This in turn
implies $(\frac{3}{2})^{-2} g_{\bf r}(x) \le g_{\bf r}(z)\le
(\frac{1}{2})^{-1/2} g_{\bf r}(x)$. 

If, on the other hand, $d_{\bf r}(x)\le 2t$, then we claim that
\be\label{est:Holder}
   \left| \rho^{{\rm TF}}*|x|^{-1} - \rho^{{\rm TF}}*|y|^{-1} \right|
   \le C |x-y|^{1/2} \,.
\ee
This can be seen as follows. 
\beax
  \left|\rho^{{\rm TF}}*|x|^{-1} - \rho^{{\rm TF}}*|y|^{-1}\right|
  &=&\left|\int_0^1 \frac{d}{d\theta} \Big(\rho^{{\rm TF}}*|\theta
    x+(1-\theta)y|^{-1}\Big) 
  \, d\theta\right|
  \\
  &=&\left|\int_0^1 \nabla \Big(\rho^{{\rm TF}}*|\theta
    x+(1-\theta)y|^{-1}\Big) \cdot (x-y)\, d\theta\right| 
  \\
  &\le&C \int_0^1 g_{\bf r}(\theta x+(1-\theta)y)\, |x-y| \, d\theta
  \\
  &\le&C|x-y|\sum_{j=1}^M \int_0^1 |\theta x+(1-\theta)y-r_j|^{-1/2}\,
  d\theta 
  \\
  &=&C|x-y|^{1/2} \sum_{j=1}^M \int_0^1 \left|
    \frac{\theta(x-y)}{|x-y|} +  
  \frac{y-r_j}{|x-y|}\right|^{-1/2} \, d\theta \,.
\eeax
Let \(n=(x-y)/|x-y|\), \(b=(y-r_j)/|x-y|\), and \(c=n\cdot b\). Then
\(|\theta n+b|^2\ge|\theta+n\cdot b|^2=(\theta+c)^2\). 
Therefore
\beax\
  \left|\rho^{{\rm TF}}*|x|^{-1} - \rho^{{\rm TF}}*|y|^{-1}\right|
  \le C|x-y|^{1/2} \sum_{j=1}^M \int_0^1 |\theta+c|^{-1/2}
  \, d\theta \,.
\eeax
The integral $\int_0^1 |\theta + c|^{-1/2} \, d\theta$ is bounded
uniformly for $c\in\R$. This proves \eqref{est:Holder}.

This allows us finally to show that for $d_{\bf r}(x)\le 2t$,
\beax 
  \rho^{{\rm TF}}*|x|^{-1} - \rho^{{\rm TF}}*|x|^{-1}*\Phi_t
  &=&\int \big(\rho^{{\rm TF}}*|x|^{-1} - \rho^{{\rm
      TF}}*|y|^{-1}\big) \Phi_t(x-y)\,dy 
  \\
  &\le&C\int |x-y|^{1/2} \, \Phi_t(x-y)\, dy
  =C t^{1/2}\,.
\eeax
This finishes the proof of the corollary.
\end{proof}

\section{Estimates of semi-classical integrals} 
\label{app:B}

In this appendix we give the remaining arguments on the analysis of
the integrals in the semi-classical proofs of Lemma~\ref{local
  semi-classics-lower} and Lemma~\ref{lm:upperbound}. 

\begin{proof}[Proof of Lemma~\ref{local semi-classics-lower}
  {\rm{({\bf Lower bound on ${\Tr}[\phi
    H_\b\phi]_-$}): {\bf Estimate of integral \eqref{eq:phiHphi}}}}]
It remains to estimate the integral in \eqref{eq:phiHphi}. 
Note that by Taylor's formula for
$\widetilde\sigma$ we have 
\begin{eqnarray}
  H_{u,q}^{(\varepsilon)}(v,p)
  &\geq&\widetilde{\s}(v,p)+\widetilde{\xi}_{v,p}(u-v,q-p)-
  C|u-v|(b^{-1}+|u-v|^2)\label{eq:atildeapp}
  \\
  &&{}-C|q-p|(b^{-1}+|q-p|^2)\,, \nonumber
\end{eqnarray}
where
$$
  \widetilde{\xi}_{v,p}(u,q)=\mfr{1}{4b}\Delta\widetilde{\s}(v,p) -
  (1-\varepsilon) 
  \mfr{1}{2}\sum_{i,j}\partial_i\p_j T_\b(p) q_iq_j
  -\mfr{1}{2}\sum_{i,j}\partial_i\partial_j V(v)u_iu_j\,.
$$
We have used that
$\left|\Delta\widetilde{\s}(v,p)-\Delta\widetilde{\s}(u,q)\right|\leq
C|u-v| + C |q-p|$, and similarly, when replacing $\p_i\p_jF(q)$ by
$\p_i\p_jF(p)$, and $\p_i\p_jV(u)$ by $\p_i\p_jV(v)$. We get that:
\bea\label{eq:bound-p}
  H_{u,q}^{(\varepsilon)}(v,p)\le0 \ \ \Rightarrow\ \  
  |p| \le C\big(1+|u-v|^3 + |q-p|^3\big)\,. 
\eea
(Note that this holds also for \(\varepsilon=0\), and uniformly in
\(\beta\in[0,1]\); to see the latter, use that \(T_\b(p)\ge T_1(p)\).) 
This implies that
\begin{equation}\label{eq:standard-int1}
  \int\limits_{H_{u,q}^{(\varepsilon)}(v,p)\le0} 
  \!\!\!\!\!\!\!\!
  \big(|u-v|^m + |q-p|^m\big)
  G_b(u-v) G_b(q-p) \,dp du dq \le C b^{-m/2}\,,
\end{equation}
and
\begin{equation}\label{eq:standard-int2}
  \int\limits_{H_{u,q}^{(\varepsilon)}(v,p)\le0} 
  \!\!\!\!\!\!\!\!
  \big(|u-v|^m + |q-p|^m\big)
  G_b(u-v) G_b(q-p) \,dp dv dq \le C b^{-m/2}\,.
\end{equation}
{F}rom this we obtain that
\begin{equation}\label{eq:h-n}
  \int G_b(u-v)G_b(q-p)
  \Big[H_{u,q}^{(\varepsilon)}(v,p)\Big]_-dpdqdv\geq {}-C\,, 
\end{equation} 
and hence from (\ref{eq:phiHphi}) that 
\begin{eqnarray*}
  \lefteqn{
  {\Tr}[\phi H_\b\phi]_-}
  \\
  &\geq&\displaystyle\int\limits_{u\in B_2}
  \phi(v+h^2ab(u-v))^2 G_b(u-v)G_b(q-p)
  \left[H_{u,q}(v,p)\right]_-\, 
  \frac{du dq}{(2\pi h)^n}dv dp 
  \\
  && {}- Ch^{-n} (b^{-3/2}+h^2b)\,.
\end{eqnarray*}
Here we have used the fact that the $u$-integration is over a bounded
region. {F}rom now on we may however ignore the restriction on the
$u$-integration. We note that, by using \eqref{eq:bound-p} and
\eqref{eq:standard-int1}, that \(\phi\) has 
support in \(B_1\), and that \(b>1\), we get that
\begin{eqnarray*}
  \lefteqn{\int\limits_{H_{u,q}^{(\varepsilon)}(v,p)\le0} 
  \!\!\!\!\!\!\!\!
  \phi(v+h^2ab(u-v))^2\big(|u-v|(b^{-1}+|u-v|^2)
  +|q-p|(b^{-1}+|q-p|^2\big)}\\
  &&\qquad\qquad\qquad\qquad\qquad\qquad\qquad\qquad\qquad
  \times\  G_b(u-v) G_b(q-p) \,dudqdvdp  \le C b^{-3/2}\,.
\end{eqnarray*}
Using this and
\eqref{eq:atildeapp} we find after the simple change of variables 
$u\to u+v$ and $q\to q+p$ that 
\begin{eqnarray}\label{eq:lowbd}
  \lefteqn{{\Tr}[\phi H_\b\phi]_-}
  \nonumber\\
  &\geq&\int\phi(v+h^2abu)^2 G_b(u)G_b(q)
  \nonumber\\
  &&\times\,\big[\widetilde{\s}(v,p)+\widetilde{\xi}_{v,p}(u,q) 
  - C|u|(b^{-1}+|u|^2)- C|q|(b^{-1}+|q|^2)\big]_-\,
  \frac{du dq}{(2\pi h)^n}dv dp 
  \nonumber\\
  &&{}-Ch^{-n}(b^{-3/2}+h^2b)
  \nonumber\\
  &\geq&\int\phi(v+h^2abu)^2 G_b(u)G_b(q)
  \big[\widetilde{\s}(v,p)+\widetilde{\xi}_{v,p}(u,q)\big]_- \,
  \frac{du dq}{(2\pi h)^n}dv dp 
  \nonumber\\
  &&{}-Ch^{-n}(b^{-3/2}+h^2b)\,.   
\end{eqnarray}
At this point we divide the $(v,p)$-integration into three regions 
given in terms of a parameter $\Lambda>0$ by 
\[\Omega_- = \{(v,p)\ |\ \widetilde{\s}(v,p)\leq-\Lambda\}\,,\ 
  \Omega_+ = \{(v,p)\ |\ \widetilde{\s}(v,p)\geq \Lambda \}\,,\ 
\Omega_0 \,= \{(v,p)\ |\ |\widetilde{\s}(v,p)|<\Lambda \}\,.\]
 
The parameter $\Lambda$ will be chosen such that 
$1\geq \Lambda\geq C b^{-1}$ for some sufficiently large constant
$C$. This is possible if $\tau$ is small enough and hence $b$ large
enough. Then, since all the second derivatives
of $\widetilde{\s}$ are bounded we may assume that $\frac{1}{4b}
|\Delta\widetilde{\s}(v,p)|<\Lambda/2$ for all $(v,p)$, uniformly in $\b$. 

We first consider $\Omega_+$. We see from \eqref{eq:lowbd} that we only
need to integrate over the set $\{(u,q)\,|\,C(|u|^2+|q|^2) \ge \Lambda\}$.
Also notice that $\widetilde{\s}(v,p)\ge \tfrac12|p|-C$ (since
\(T_\b(p)\ge T_1(p)\))
shows that we only need to integrate over the set $\{p\,\,|\,\, |p|\le
C(1+|q|^2+|u|^2)\}$. Therefore, 
\begin{eqnarray*}
  \int\limits_{(v,p)\in\Omega_+}
  \!\!\!\!\!\!\!
  \phi(v+h^2abu)^2 G_b(u) G_b(q)
  \big[\widetilde{\s}(v,p)+\widetilde{\xi}_{v,p}(u,q)\big]_-\,{du dq}
  dv dp
  \geq {}-C{\rm e}^{-Cb\Lambda}\,.
\end{eqnarray*}

A similar argument shows that on $\Omega_-$ we can ignore the negative
part $[\ \ ]_-$ paying the same price $-C h^{-n}{\rm e}^{-Cb\Lambda}$. 

For $(v,p)\in\Omega_-$ we estimate the integral
\beax
  \int \phi(v+h^2abu)^2 G_b(u) G_b(q)
  \big[\widetilde{\s}(v,p)+\widetilde{\xi}_{v,p}(u,q)\big]\,{dudq}
  \geq(\phi(v)^2+Ch^2b)\widetilde{\s}(v,p) - Ch^2b \,.
\eeax
Here we have expanded $\phi^2$ to second order at the point $v$ and
used the crucial fact that
\bea\label{eq:Gauss-integrals}
  \int \widetilde{\xi}_{v,p}(u,q) G_b(u) G_b(q) \, dudq =0\,.
\eea
For $(v,p)\in\Omega_-$ we have, of course, $\widetilde{\s}(v,p)=
\widetilde{\s}(v,p)_-$. Since the volume of $\Omega_-$ is bounded 
by a constant we get for the integration over $\Omega_-\cup\Omega_+$,
\begin{eqnarray}\lefteqn{
  \int\limits_{(v,p)\in\Omega_-\cup\Omega_+}
  \!\!\!\!\!\!\!\!\!
  \phi(v+h^2abu)^2 G_b(u) G_b(q)
  \big[\widetilde{\s}(v,p)+\widetilde{\xi}_{v,p}(u,q)\big]_-\,{du dq}
  dv dp}\nonumber \\ 
  &&\qquad\qquad\qquad\geq\int\limits_{(v,p)\in\Omega_-\cup\Omega_+}
  \!\!\!\!\!\!\!\!\!
  \phi(v)^2\widetilde{\s}(v,p)_- \,{dvdp} - 
  C(h^2b+ {\rm e}^{-Cb\Lambda})\,.\label{eq:Omega+-}
\end{eqnarray}
Finally, let $(v,p)\in\Omega_0$. Observe that, with
\(\vartheta(t)=(2t+\beta t^2)^{n/2}\), 
\begin{eqnarray}\label{eq:Omega-0-int}
   \int\limits_{(v,p)\in\Omega_0}\!\!\!dp =
   c_n\big(\vartheta(-[\Lambda-V(v)]_{-})-\vartheta(-[\Lambda+V(v)]_{-})\big)  
   \le C\Lambda\,,
\end{eqnarray}
by the mean value theorem (uniformly in \(v\)). Now,
\beax
   \lefteqn{\phi(v+h^2abu)^2\,G_b(u) G_b(q) 
  \big[\widetilde{\s}(v,p)+\widetilde{\xi}_{v,p}(u,q)\big]_{-}}
  \\&\ge&
   \phi(v+h^2abu)^2 \,G_b(u) G_b(q) 
  \big[\widetilde{\s}(v,p)\big]_{-}\\
  &&{}-C \phi(v+h^2abu)^2 \,G_b(u) G_b(q)(b^{-1}+|u|^2+|q|^2)\,,
\eeax
and, using the observation above and making the change of variables
\(v\to v-h^2abu\) in the \(v\)-integral, 
\beax
  \int\limits_{(v,p)\in\Omega_0}\phi(v+h^2abu)^2 \,G_b(u)
  G_b(q)(b^{-1}+|u|^2+|q|^2)\,dvdpdudq
  \le C\Lambda b^{-1}\,.
\eeax
Expanding  $\phi^2$ to first order at $v$ we have that
 \begin{eqnarray*}
   \lefteqn{\int\limits_{(v,p)\in\Omega_0}\phi(v+h^2abu)^2 \,G_b(u)
   G_b(q)\widetilde{\s}(v,p)_{-}\,dvdpdudq}
   \\&\ge&
   \int\limits_{(v,p)\in\Omega_0}\phi(v)^2
   \widetilde{\s}(v,p)_{-}\,dvdp
   \,+\, Ch^2ab\!\!\!\!\!\!\!
   \int\limits_{(v,p)\in\Omega_0, v\in{\text{\rm supp}}\,V}\!\!\!\!\!\!\!
   |u|\,G_b(u)G_b(q)\widetilde{\s}(v,p)_{-}\,dudqdvdp\\
   &\ge&\int\limits_{(v,p)\in\Omega_0}\phi(v)^2
   \widetilde{\s}(v,p)_{-}\,dvdp
   -Chb^{1/2}\Lambda^2\,.
 \end{eqnarray*}
As a consequence, 
\begin{eqnarray}\nonumber
  \lefteqn{
  \int\limits_{(v,p)\in\Omega_0}\!\!\!\!\!\!
  \phi(v+h^2abu)^2 \,G_b(u) G_b(q) \big[\widetilde{\s}(v,p)+
  \widetilde{\xi}_{v,p}(u,q)\big]_-\,{du dq} dvdp}\\
   &\qquad\qquad\qquad\geq&\int\limits_{(v,p)\in\Omega_0} 
  \!\!\!\!\!\!
  \phi(v)^2\widetilde{\s}(v,p)_-\, {dvdp}
  -C\Lambda(\Lambda hb^{1/2}+b^{-1})\,.\label{eq:Omega0}
\end{eqnarray}
Since 
\beax
  \lefteqn{\int\phi(v)^2\widetilde{\s}(v,p)_-\, {dvdp}
  }
  \\
  &=&(1-\varepsilon)^{-n} \int
  \phi(v)^2\big[\sqrt{\b^{-1}p^2 + (1-\varepsilon)^2\b^{-2}}
  -(1-\varepsilon)\b^{-1} + V(v)\big]_-\, dvdp
  \\
  &\geq&(1-\varepsilon)^{-n} \int \phi(v)^2 {\s}_\b(v,p)_-\, {dvdp} 
  \\
  &\geq&\int\phi(v)^2{\s}_\b(v,p)_-\, {dvdp} - C \varepsilon \,,
\eeax
the lemma follows from \eqref{eq:lowbd}, 
\eqref{eq:Omega+-}, and 
\eqref{eq:Omega0} if we choose $b=h^{-4/5}$.
\end{proof}

\begin{proof}[Proof of Lemma~\ref{lm:upperbound} {\rm{({\bf
        Construction of a trial density}): {\bf Estimates of integrals}}}]

We give here the remaining arguments on the
analysis of the integrals in the semi-classical proofs of
Lemma~\ref{lm:upperbound}. 

\ 

\noindent{\bf The energy: proof of \eqref{eq:lemmaupper}.}
It remains to estimate the integral in \eqref{eq:first-sc-int-in app}.  

Using
\eqref{eq:bound-p}, that \(T_\b(p)\le\frac12p^2\), and that
${h}_{u,q}(v,p)=0$ unless $u\in B_2$, we get that
$$
  \int \,\chi[{h}_{u,q}(v,p)]\,G_b(u-v)G_b(q-p)
  \big(1 + T_{\beta}(p+h^2ab(q-p))\big) \,dudq dv dp\leq C\,.
$$
This implies that
\begin{eqnarray*}
  \lefteqn{{\Tr}[\c\phi H_\b \phi]}&&\\
  &\le&
  \int \,\chi[h_{u,q}(v,p)]\,G_b(u-v)G_b(q-p) \phi\big(v+h^2ab(u-v)\big)^2 
  \\
  &&\qquad\times \,
  \s\big(v+h^2ab(u-v),p+h^2ab(q-p)\big)\,
  \frac{dudq}{(2\pi h)^n} dv dp + Ch^2b h^{-n}\,.
\end{eqnarray*}
{F}rom (\ref{eq:happ}) we may now conclude that 
\begin{eqnarray}
  \lefteqn{{\Tr}[\c\phi H_\b\phi]
  }\nonumber\\
  &\leq&\int\limits_{u\in B_2-v}
  \,\chi\big[\s(v,p)+\xi_{v,p}(u,q)-C|u|(b^{-1}
             +|u|^2)-C|q|(b^{-1}+|q|^2)\big] 
  \,G_b(u)G_b(q) \nonumber
  \\ 
  &&\qquad\times\, \phi(v+h^2abu)^2\s(v+h^2abu,p+h^2abq)
  \, \frac{dudq}{(2\pi h)^n} dv dp\, +\, Ch^2b h^{-n}\label{eq:laplace1}\,.
\end{eqnarray}
At this point we introduce the same partition of the
$(v,p)$-integration into sets $\Omega_{\pm},\Omega_0$ as in the proof
of the lower bound above (with
$\e=0$) with the same $\Lambda=b^{-1/2}=h^{2/5}$. 

Then for the integration over $\Omega_+$ we have as above that 
$C(|u|^2+|q|^2)>\Lambda$ and hence
\beax
  \lefteqn{\int\limits_{(v,p)\in\Omega_+,\atop u\in B_2-v} 
  \chi\big[\s(v,p)+\xi_{v,p}(u,q)-C|u|(b^{-1}+|u|^2)-C|q|(b^{-1}+|q|^2)\big]} 
  \\
  &&\,\times  \,\phi(v+h^2abu)^2 \s(v+h^2abu,p+h^2abq)\,
  G_b(u)G_b(q)\,du dq dv dp 
  \le C \,{\rm e}^{-c b\Lambda} \le Ch^2b\,,
\eeax
where we have used (\ref{eq:uniformsigma}) and that $\phi$ is
supported in the ball $B_1$. 

Similarily, if $(v,p)\in \Omega_-$ then for the $(u,q)$-integration 
we can safely assume that the argument of 
$\chi$ is negative to the effect of paying the same ${\rm
  e}^{-Cb\Lambda}$ price. Likewise we may ignore the restriction
$u\in B_2-v$, since $u\not\in B_2-v$ and $v+h^2abu\in B_1$ implies 
$|u|>(1-h^2ab)^{-1}>1$. 
Expanding $\phi^2$ and $\s$ to second order at $(v,p)\in\Omega_-$ and
using the fact that all their second order derivatives are bounded
together with 
\eqref{eq:Gauss-integrals-one} we get that
\beax 
  \lefteqn{\int
  \chi\big[\s(v,p)+\xi_{v,p}(u,q)-C|u|(b^{-1}+|u|^2)-C|q|(b^{-1}+|q|^2)\big]} 
  \\
  &&\qquad\qquad\qquad\qquad\quad\ \ \times  \,\phi(v+h^2abu)^2 \s(v+h^2abu,p+h^2abq)\,
  G_b(u)G_b(q)\,du dq\\ 
  &\leq& \int \Big[\big(\phi(v)^2 + h^2ab\,u\cdot\nabla(\phi^2)(v)\big)
  \, \big(\s(v,p) + h^2ab \, (u,q)\cdot\nabla\s(v,p)\big)\Big]\,
  G_b(u)G_b(q)\,du dq 
  \\
  &&\, + \,Ch^2b + C \,{\rm e}^{-Cb\Lambda}
  \\
  &\le& \phi(v)^2 \s(v,p)_- + Ch^2b \,.
\eeax
It is important here that $\s$ and $\nabla\s$ are bounded  uniformly
in $\beta\leq1$ on $\Omega_-$. This follows from
(\ref{eq:uniformsigma}) and 
$|\nabla\sigma_\beta(v,p)|\leq C(1+|p|)$. Indeed,
(\ref{eq:uniformsigma}), 
in particular, implies that $\Omega_-$ is a bounded set
(uniformly in $\beta$). 
The fact that the volume of $\Omega_-$ is
bounded also gives that the contribution from $\Omega_-$ to the
integral on the right side of (\ref{eq:laplace1}) is bounded above by
$$
(2\pi h)^{-n}\int_{\Omega_-}\phi(v)^2 \s(v,p)_- dvdp + Ch^2bh^{-n}. 
$$

Finally, we consider $(v,p)\in\Omega_0$. 
If we expand $\phi^2$ to first order at $v$ and $\s$ to second order
at $(v,p)$ and use that all second order derivatives of $\s$ are
bounded and that $\nabla\s(v,p)$ is bounded for $(v,p)\in\Omega_0$ we
obtain that 
$$
\phi(v+h^2abu)^2 \s(v+h^2abu,p+h^2abq)\leq
\phi(v)^2\s(v,p)+Ch^2ab(|u|+|q|)+Ch^4a^2b^2(|u|^2+|q|^2).
$$ 
This together with the estimate $|\chi(x+y)x
-\chi(x)x|\le |y|$ implies that 
\beax \lefteqn{\int_{u\in B_2-v}
  \chi\big[\s(v,p)+\xi_{v,p}(u,q)-C|u|(b^{-1}+|u|^2)-C|q|(b^{-1}+|q|^2)\big]}
  \\
  &&\qquad\qquad\qquad\qquad\quad\ \ \times \,\phi(v+h^2abu)^2
  \s(v+h^2abu,p+h^2abq)\,
  G_b(u)G_b(q)\,du dq\\
  &\le&
  \int\chi[\s(v,p)+\xi_{v,p}(u,q)-C|u|(b^{-1}+|u|^2)-C|q|(b^{-1}+|q|^2)]\,G_b(u)
  G_b(q) \, du dq \\
  &&\qquad\qquad\qquad\qquad\quad\times\phi(v)^2 \s(v,p)+
  C h^2b^{3/2}\\
  &\le&\phi(v)^2\s(v,p)_- + C (b^{-1}+h^2b^{3/2})\,.  
\eeax 
We have here again used that the effect of removing the restriction
$u\in B_2-v$
causes a smaller error than the last term above. 
Note that $u\in B_2-v$ and
$v+h^2abu\in B_1$ imply $|u|\leq 3(1-h^2ab)^{-1}\leq 6$ and
hence we only need to consider 
$|v|\leq |v+u|+|u|<8$. If we use that \eqref{eq:Omega-0-int}
implies
$$
  {\rm Vol}(\Omega_0\cap\{v\ |\ |v|<8\}\times\R^n)\leq C\Lambda
$$
we see that  
the contribution from $\Omega_{0}$ 
to the
integral on the right side of \eqref{eq:laplace1} is bounded above by 
$$
  (2\pi h)^{-n}\int_{\Omega_0}\phi(v)^2\s(v,p)_-\,dvdp + 
  C h^{-n}(b^{-1}+h^2b^{3/2})\Lambda\,.
$$
This finishes the proof of the upper bound on the energy in 
\eqref{eq:lemmaupper} .

\

\noindent{\bf The density: proof of \eqref{eq:rhogammaprop1} and
  \eqref{eq:rhogammaprop2}.} 
Here it remains to estimate the integral in \eqref{eq:rhointegral}.
The strategy
is to freeze the variable $|p|$ in $\xi_{v,p}$ so that the remaining
dependence on $|p|$ is explicitly integrable.
This is accomplished in the estimate 
\eqref{eq:kappa estim} below. After the $|p|$-integration the proof is
almost the same as in the non-relativistic case \cite{SS}.

We write $p=|p|\omega$ and define 
$$
  p_0=(\beta|V(v)_-|^2+2|V(v)_-|)^{1/2}\omega=\eta(V(v))^{1/n}\omega\,.
$$
We will then prove that if $u\in B_2-v$ then 
\begin{equation}\label{eq:kappa estim}
  \chi[\s(v,p) + \xi_{v,p_0}(u,q) + R(u,q)] \le
  \chi[h_{u+v,q+p}(v,p)] \le \chi[\s(v,p) + \xi_{v,p_0}(u,q) 
  - R(u,q)] \,,
\end{equation}
where 
$$
  R(u,q)=C(|u|(b^{-1}+|u|^2) + (|q|+\Lambda)(b^{-1}+|q|^2)
  + (b^{-1} +|u|^2+|q|^2)(|u|^2+|q|^2)\Lambda^{-1}) \,.
$$
We first prove \eqref{eq:kappa estim}  for $(v,p)\in\Omega_0$.
In this case we have
$$\eta(V(v)+\Lambda)^{2/n}\leq p^2\leq
\eta(V(v)-\Lambda)^{2/n}\,,$$
from which it follows that $|p^2-p_0^2|\leq C\Lambda$ with a constant
independent of $\beta\in[0,1]$.

Let \(G(t)=\sqrt{\b^{-1}t+\b^{-2}}-\b^{-1}\), so that
$T_\b(p)=G(p^2)$ (we suppress that \(G\) depends on \(\b\)).
Note that then \(\p_i\p_jT_\b(p)=4p_ip_jG''(p^2)+2\delta_{ij}G'(p^2)\),
and so, in particular, $\D T_\beta(p) = 4 G''(p^2) p^2 + 2n
G'(p^2)$. Therefore, using that \(p_i=|p|\omega_i,
p_{0,i}=|p_0|\omega_i\), 
\bea\nonumber
  \lefteqn{\left| \xi_{v,p}(u,q) - \xi_{v,p_0}(u,q) \right|
  }\\\nonumber
  &\le&\mfr{1}{4b} \left|\D\s(v,p)-\D\s(v,p_0)\right| 
   + \mfr{1}{2} \sum_{i,j}\big|\p_i\p_j [T_\b(p) - T_\b(p_0)]\big| \,|q_iq_j|
  \\\nonumber
  &\le& C(b^{-1}+|q|^2) \big(|G''(p^2)p^2 
  - G''(p_0^2)p_0^2| + |G'(p^2) - 
  G'(p_0^2)|\big)
  \\
  &\le&C(b^{-1} + |q|^2) \,|p^2-p_0^2|
  \ \le\ C\Lambda(b^{-1} + |q|^2)\,.  \label{eq:xi-ineq}
\eea
Here we have used the choice of \(p_0\) and that $G'(t)$ and $t
G''(t)$ have bounded derivatives uniformly in $\beta\in[0,1]$. If we
combine \eqref{eq:happ} with \eqref{eq:xi-ineq} we obtain that
\[ 
  \left| h_{u+v,q+p}(v,p) - \s(v,p) - \xi_{v,p_0}(u,q)\right|
  \le C |u|(b^{-1}+|u|^2) + |q|(b^{-1}+|q|^2) + C\Lambda(b^{-1} +|q|^2) \,,
\]
which is, in fact, stronger than \eqref{eq:kappa estim}.

If $(v,p)\in\Omega_+$ we see that the left inequality in
\eqref{eq:kappa estim} is only violated if $\xi_{v,p_0}(u,q)\leq-\Lambda$
and the right inequality is only violated if $h_{u+v,q+p}(v,p)\leq 0$.
Since $(v,p)\in\Omega_+$ we must in both cases have $C(|u|^2+|q|^2)>\Lambda$.
We hence get (again using
\eqref{eq:happ}) that
\begin{eqnarray*}
  \lefteqn{\big| h_{u+v,q+p}(v,p)-\s(v,p) -\xi_{v,p_0}(u,q)\big|}&&\\
   &\le &  \big| h_{u+v,q+p}(v,p)-\s(v,p)-\xi_{v,p}(u,q)\big|+
   \big|\xi_{v,p}(u,q)\big| +\big|\xi_{v,p_0}(u,q)\big|  
  \\
  &\le& C|u|(b^{-1}+|u|^2) + C|q|(b^{-1}+|q|^2)  + C(b^{-1} +
  |u|^2+|q|^2) 
  \\
  &\le& C|u|(b^{-1}+|u|^2) + C|q|(b^{-1}+|q|^2)  + C(b^{-1} +
  |u|^2+|q|^2) (|u|^2+|q|^2) 
  \Lambda^{-1}\,,
\end{eqnarray*}
which gives  \eqref{eq:kappa estim} in this case.

Finally, if $(v,p)\in\Omega_-$ then the left inequality in
\eqref{eq:kappa estim} is only violated if
$h_{u+v,q+p}(v,p)\geq0$ and the right inequality is only violated if
$\xi_{v,p_0}(u,q)\geq\Lambda$. In both cases this implies that
$C(|u|^2+|q|^2)>\Lambda$ and hence the same argument as for $\Omega_+$
proves \eqref{eq:kappa estim}.

Using \eqref{eq:kappa estim} we can estimate the density in
\eqref{eq:rhointegral}
from above and below. We will discuss the lower bound on the density. 
The upper bound is similar.
Performing the $|p|$-integral in \eqref{eq:rhointegral} we obtain 
\bea\label{eq:rholower1} 
  \rho_\gamma(x)
  &\geq& \int\limits_{u\in B_2-v} 
  \!\!\!\!\!\! \Xi(u,q,v,\omega) 
  G_b(u) G_b(q)
  G_{(h^2b)^{-1}}(x-v-h^2abu)\, dvd\omega\,
  \frac{dudq}{(2\pi h)^n}
  \\
  &=&  \!\!\! \!\!\!\int\limits_{(1-h^2ab)u\in B_2-v} 
  \!\!\!\!\!\!\!\!\!\!\!\!
  \Xi(u,q,v-h^2abu,\omega)
  G_b(u) G_b(q) G_{(h^2b)^{-1}}(x-v) \, dvd\omega\,
  \frac{dudq}{(2\pi h)^n} \,,    \nonumber 
\eea
where
$$ 
  \Xi(u,q,v,\omega) =
  n^{-1}\eta\left(V(v)+\xi_{v,p_0}(u,q)+R(u,q)\right)\,. 
$$
We have that 
\begin{eqnarray}
  \lefteqn{V(v-h^2abu)+\xi_{v-h^2abu,p_0}(u,q)+R(u,q)}\nonumber
  \\&\leq& V(v)-h^2ab u\nabla V(v)+\xi_{v,p_0}(u,q)
  +R(u,q)+Ch^4a^2b^2|u|^2+Ch^2ab|u|(b^{-1} + |u|^2)\nonumber\\
  &\leq&V(v)-h^2ab u\nabla V(v)+\xi_{v,p_0}(u,q)
  +R(u,q)+Ch^4a^2b^2|u|^2\,.\label{eq:Xiboundint}
\end{eqnarray}
(In the last line we have used that \(h^2ab\le1\) and the definition
of \(R(u,q)\).) We now use that
\be \label{eta expansion}
   \left| \eta(s) - \eta(t) - \eta'(t)(s-t)\right|
   \le \left\{\begin{array}{ll}
       C|s-t|^{3/2}+C(|s|+|t|)|s-t|^2\,,&n=3\\C(|s|^{\frac{n}{2}-2} +
       |t|^{\frac{n}{2}-2}+|s|^{n-2}+|t|^{n-2})|s-t|^2\,,
   &n\ge4\end{array}\right.\,.  
\ee
We continue with the case $n=3$ and leave $n\ge4$ to the
reader.

If we use the fact that $\eta'(V(v))$ is bounded independently of
$\beta\in[0,1]$ we obtain from \eqref{eq:Xiboundint} and \eqref{eta
  expansion}
\begin{eqnarray}
  n\,\Xi(u,q,v-h^2abu,\omega)&\geq&\eta(V(v))+\eta'(V(v))(\xi_{v,p_0}(u,q)-h^2ab 
  u\nabla
  V(v))\nonumber\\&&{}-C\big(b^{-1}+|q|^2+|u|^2+h^2ab|u|+R(u,q)
  \big)^{3/2}
  \nonumber\\&&{}-C\big(b^{-1}+|q|^2+|u|^2+h^2ab|u|+R(u,q)\big)^{2}
  \nonumber\\&&{}-CR(u,q)-Ch^4a^2b^2|u|^2-Ch^2ab|u|\,.\label{eq:Xibound}
\end{eqnarray}
It is now crucial that (see \eqref{eq:Gauss-integrals-one} and
\eqref{eq:Gauss-integrals}) 
$$ 
\int (\xi_{v,p_0}(u,q)-h^2ab
u\nabla
V(v))\,G_b(u) G_b(q)\, du dq = 0\,.
$$
Since $v\in {\rm supp}(V)\subset B_{3/2}$ and
$(1-h^2ab)u\not\in B_2-v$ implies $|u|>1/2$ we find
\be \label{mu-estimate2}
  \Bigg|\,\int\limits_{(1-h^2ab)u\in B_2-v}(\xi_{v,p_0}(u,q)-h^2ab 
  u\nabla V(v)) \,G_b(u) G_b(q)\, du dq\,\Bigg|  \le C \,{\rm
    e}^{-b/5} \le C h^{6/5}\,. 
\ee
Combining this with \eqref{eq:Xibound} and inserting it into
\eqref{eq:rholower1} we arrive at  (recall that $\Lambda=b^{-1/2}$)
\bea\label{eq:rholower2}
  \rho_\c(x) &\geq& (2\pi h)^{-3}\w_3 \int\eta\big[V(v)\big]
  G_{(h^2b)^{-1}}(x-v)\, dv 
  \nonumber\\   
  &&\,-\, Ch^{-3}\big(h^{6/5}+b^{-3/2}+(h^2ab)^{3/2}b^{-3/4}\big)\,.
\eea
Here we have removed the constraint $(1-h^2ab)u\in B_2-v$ by the same
argument as above. 

We shift the $v$-coordinate by $x$, and then expand
$\eta\big[V(x+v)\big]$ in the integral at $x$ by expanding \(V\) to
second order at \(x\) and using \eqref{eta expansion}. Then
\beax
  \eta\big[V(x+v)\big]\geq \eta\big[V(x)\big]
  +\eta'\big[V(x)] \nabla V(x)\cdot v -C(|v|^{3/2}+|v|^3)\,. 
\eeax  
Then we obtain from \eqref{eq:rholower2} (using
\eqref{eq:Gauss-integrals-one}) that
$$ 
  \rho_\c(x) - (2\pi h)^{-3}\w_3\, \eta\big[V(x)\big] \geq -C
  h^{-3}\big(h^{6/5}+(h^2b)^{3/4}\big) \geq{}-C h^{-3+9/10}\,.
$$
This finishes the proof of \eqref{eq:rhogammaprop1}.

Lastly, we prove \eqref{eq:rhogammaprop2}. By integrating 
\eqref{eq:rholower2} we see that  
\beax
  \lefteqn{\int \phi(x)^2\rho_\c(x)\, dx}
  \\
  &\geq&(2\pi h)^{-3}\w_3 \int \phi(x)^2\,
  G_{(h^2b)^{-1}}(x-v)\,\eta\big[V(v)\big]\, dxdv  
  - C h^{-3+6/5}
  \\
  &\geq&(2\pi h)^{-3}\w_3\int \phi(v)^2 \,\eta\big[V(v)\big]\, dv - C
  h^{-3+6/5}\,. 
\eeax
In the last step we have expanded $\phi^2$ to second order at \(v\)
to obtain that (see also \eqref{eq:Gauss-integrals-one})
$$ \int \phi(x)^2  G_{(h^2b)^{-1}}(x-v)\, dx \le \phi(v)^2 + C h^{6/5}\,.
$$
This finishes the proof of \eqref{eq:rhogammaprop2} and therefore of
Lemma~\ref{lm:upperbound}.
\end{proof}
\end{appendix}
\begin{acknowledgement}
The authors thank the Erwin Schr\"odinger Institute (Vienna, Austria)
for hospitality at various visits. T\O S and WLS wishes to thank the
Department of Mathematics at University of Copenhagen for the
hospitality at numerous visits.

Work of the authors was partially supported by two EU TMR
    grants, and by various grants from The Danish National Research
    Foundation. 
\end{acknowledgement}
%

\providecommand{\bysame}{\leavevmode\hbox to3em{\hrulefill}\thinspace}
\providecommand{\MR}{\relax\ifhmode\unskip\space\fi MR }
\providecommand{\MRhref}[2]{%
  \href{http://www.ams.org/mathscinet-getitem?mr=#1}{#2}
}
\providecommand{\href}[2]{#2}

\end{document}